\documentclass [oneside,reqno,12pt,a4paper]{article}
\usepackage{amsfonts,latexsym,amsmath,amssymb,amsthm}
\usepackage{graphicx}
\usepackage{caption}
\usepackage{xcolor}
\usepackage{longtable}
\usepackage[export]{adjustbox}
\def\a{\alpha}
\newcommand\be{\beta}
\newcommand\de{\delta}

\newcommand\vth{\vartheta}
\newcommand\et{\eta}
\newcommand\ep{\epsilon}

\newcommand\ka{\kappa}
\newcommand\la{\lambda}

\newcommand\ph{\varphi}

\newcommand\ps{\psi}

\newcommand\rh{\rho}
\newcommand\si{\sigma}

\newcommand\om{\omega}

\newcommand\A{{\cal A}}
\newcommand\B{{\cal B}}

\newcommand\Reals{\mathbb R}
\newcommand\xix{X}
\newcommand\yix{Y}
\newcommand\dom{h}
\newcommand\abs[1]{\left|#1\right|}

\def\pder#1#2{\frac{\partial#1}{\partial#2}}
\newcommand\wb{{\bar w}}

\newcommand{\arctanh}{\operatorname{arctanh}}
\newcommand{\sgn}{\operatorname{sgn}}
\newtheorem{theorem}{Theorem}
\newtheorem{lemma}[theorem]{Lemma}
\newtheorem{remark}[theorem]{Remark}

\newtheorem{proposition}[theorem]{Proposition}
\newtheorem{corollary}[theorem]{Corollary}
\newtheorem*{proposition*}{Proposition}
\numberwithin{equation}{section}
\numberwithin{theorem}{section}
\begin{document}
\parindent=15pt
\normalbaselineskip20pt
\baselineskip20pt
\global\hoffset=-15truemm
\global\voffset=-10truemm
\allowdisplaybreaks[1]
\title{Two-locus clines on the real line \\ with a step environment}
\date{}
\maketitle

\begin{center}
{\large
Reinhard B\"urger

\vspace{2cm}
Department of Mathematics,\\
University of Vienna,\\
Austria
}
\end{center}
\vspace{5cm}

{
\normalbaselineskip15pt
\baselineskip15pt{
{\obeylines 
{\bf Address for correspondence:}
\bigskip 
Reinhard B\"urger
Fakult\"at f\"ur Mathematik 
Universit\"at Wien 
Oskar-Morgenstern-Platz 1
A-1090 Wien 
Austria
E-mail: reinhard.buerger@univie.ac.at
Phone: +43 1 4277 50784
}
}
}

\vspace{1.5cm}

\newpage
\section*{Abstract}
The shape of allele-frequency clines maintained by migration-selection balance depends not only on the properties of migration and selection, but also on the dominance relations among alleles and on linkage to other loci under selection. We investigate a two-locus model in which two diallelic, recombining loci are subject to selection caused by an abrupt environmental change. The habitat is one-dimensional and unbounded, selection at each locus is modeled by step functions such that in one region one allele at each locus is advantageous and in the other deleterious. We admit an environmentally independent, intermediate degree of dominance at both loci, including complete dominance. First, we derive an explicit expression for the single-locus cline with dominance, thus generalizing classical results by Haldane (1948). We show that the slope of the cline in the center (at the step) or, equivalently, the width of the cline, is independent of the degree of dominance. Second, under the assumption of strong recombination relative to selection and migration, the first-order approximations of the allele-frequency clines at each of the loci and of the linkage disequilibrium are derived. This may be interpreted as the quasi-linkage-equilibrium approximation of the two-locus cline. Explicit asymptotic expressions for the clines are deduced as $x\to\pm\infty$. For equivalent loci, explicit expressions for the whole clines are derived. The influence of dominance and of linkage on the slope of the cline in the center and on a global measure of steepness are investigated. This global measure reflects the influence of dominance. Finally, the accuracy of the approximations and the dependence of the shape of the two-locus cline on the full range of recombination rates is explored by numerical integration of the underlying system of partial differential equations. 

\vskip2.5cm
{\bf Key words:} Selection; Migration; Recombination; Dominance; Linkage disequilibrium; Dispersal; Geographical structure; Population structure

\newpage
\setcounter{section}{0}
\section{Introduction}
A cline describes a gradual change in genotypic or phenotypic frequency as a function of spatial location. Such clines are frequently observed in natural populations and are an important and active research area in evolutionary biology and ecology (e.g., Endler 1977, Hoffman et al.\ 2002, Lohman et al.\ 2017). Clines typically occur in species distributed along an environmental gradient, for instance in temperature, where alternative phenotypes or genotypes are better adapted to the different extremes of the environment. Dispersal leads to mixing, reduces local adaptation, and entails a continuous, often sigmoidal, change in type frequencies across space. The study of clines can be used to obtain insight into the relative strengths of the evolutionary and ecological forces acting on this species.

Haldane (1948) devised a model in terms of a reaction-diffusion equation which approximates migration by diffusion and assumes that there is a step environment on the real line such that one allele is advantageous if $x>0$ and the other if $x<0$. He derived explicit expressions for the cline, the spatially non-constant stationary solution, in terms of hyperbolic functions for the two cases of no dominance and of a completely dominant allele. The slope of the cline in the center, i.e., at the environmental step, can be expressed in terms of the selection intensity and the migration variance. He used this relation to infer the strength of selection on a population of deer mouse. 

The mathematical theory of clines became a very active research area in the 1970s. Various patterns of spatial variation in fitnesses were investigated (e.g., environmental pockets, periodic changes), as were variation (e.g., barriers) or asymmetry in migration (e.g., Slatkin 1973; Nagylaki 1975, 1976, 1978). These works focused on the derivation of explicit results about the shape of clines. At about the same time and motivated by this work, Conley (1975), Fleming (1975), Fife and Peletier (1977), and Henry (1981) developed and employed advanced mathematical methods to investigate existence, uniqueness, and stability of clinal solutions under a variety of assumptions about fitnesses, i.e., for quite general classes of functions that describe selection caused by the environment. 

Lou and Nagylaki (2002, 2004, 2006) extended much of the previous work on spatially varying selection in several directions. In most of their analyses, migration is modeled by general elliptic operators on bounded domains in arbitrary dimensions and for wide classes of fitness functions. Such elliptic operators arise if migration is spatially inhomogeneous or anisotropic (Nagylaki 1989). In addition, they studied the maintenance of clines at multiallelic loci, which does not only add realism but also produces new and interesting phenomena.

In the present work, we focus on the role of dominance and of linkage between loci. In particular, we investigate how the shape of a cline at one locus is affected by linkage to a second locus. We choose a step environment on the real line, such that in each of the two regions ($x>0$, $x<0$) one of the alleles at each locus is advantageous, the other deleterious. The strength of selection acting at the two loci may be different, though. The choice of a step environment on the whole real line (as opposed to a bounded interval) has the advantage that explicit results are obtained more readily.
 
We admit arbitrary intermediate dominance, including complete dominance, and assume that its degree is independent of the environment. The influence of dominance on the maintenance of single-locus clines has been studied before. In particular, on bounded domains the number and stability of clines depend on the degree of dominance (Henry 1981, Lou and Nagylaki 2002, Nagylaki and Lou 2008, Lou et al.\ 2010, Nakashima et al.\ 2010). Explicit results about the shape of a cline seem to be rare and confined to the two cases of no dominance and complete dominance (e.g., Haldane 1948, Nagylaki 1975).  

The first study of a two-locus cline model is due to Slatkin (1975), who showed numerically that the linkage disequilibrium generated between the two loci tends to steepen the clines. Barton (1983, 1986, 1999) derived general results about the consequences of linkage and linkage disequilibrium among multiple loci. Although he derived them for and applied them to hybrid zones, they are also of relevance in our context. 

One of the novel features of our work is that we derive an analytically explicit solution for the single-locus cline with dominance (Section \ref{sec:one_locus}). In particular, we show that the slope of the cline in its center or, equivalently, the width of the cline is independent of the degree of dominance. Our main achievement is the derivation of the first-order perturbation of the allele-frequency clines at each of the loci if recombination is strong relative to selection and diffusion (Section \ref{sec:approx_twoloc}). In other words, we derive the quasi-linkage-equilibrium approximation for the two-locus cline. In Sections \ref{sec:properties_twoloc} and \ref{sec:global_steepness}, we study two measures of steepness of clines and their dependence on dominance and linkage. We derive the asymptotic properties of the allele-frequency clines of the two-locus system in Section \ref{sec:asymptotics}. For equivalent loci, we obtain an analytically explicit expression for the strong-recombination approximation of the two-locus cline (Section \ref{sec:explicit}). In Section \ref{sec:norec}, we briefly treat the case of no recombination. Finally, we provide numerical checks of the accuracy of our approximations and illustrate the dependence of two-locus clines on the full range of recombination rates by numerical integration of the system of partial differential equations (Section \ref{sec:numerics}).

\section{The model}\label{sec:model}
We consider a monoecious, diploid population that occupies a linear, unbounded habitat in which it is uniformly distributed and mates locally at random. Fitness of individuals depends on location and is determined by two diallelic loci, $\A$ and $\B$, which recombine at rate $r\ge0$. We model dispersal by diffusion on the real line $\Reals=(-\infty,\infty)$, and assume it is homogeneous, isotropic, and genotype-independent, with migration variance $\si^2$.

The frequencies of the gametes $AB$, $Ab$, $aB$, and $ab$, at position $x\in\Reals$ and time $t$ are $p_1=p_1(x,t)$, $p_2=p_2(x,t)$, $p_3=p_3(x,t)$, and $p_4=p_4(x,t)$, respectively, where $p_i\ge0$ and $\sum_{i=1}^4 p_i=1$. Let $D=p_1p_4-p_2p_3$ denote the usual measure of linkage disequilibrium, and let $\mathbf p=(p_1,p_2,p_3,p_4)^T$. If $w_{ij}(x)$ is the fitness of the diploid genotype $ij$ ($i,j\in\{1,2,3,4\}$) at location $x\in\Reals$, then $w_i = w_i(x,\mathbf p) = \sum_{j=1}^4 w_{ij}(x)p_j$ is the marginal fitness of gamete $i$, and $\wb = \wb(x,\mathbf p) = \sum_{i=1}^4 w_ip_i$ is the population mean fitness. 

Throughout, we use primes, $\vphantom{f}'$\,, and dots, $\dot{\vphantom{p}}$\,, to indicate  partial derivatives with respect to $x$ and $t$, respectively. We assume that (i) the three evolutionary forces selection, migration, and recombination are of the same order of magnitude and sufficiently weak, (ii) migration is genotype independent and spatially uniform and symmetric, and (iii) Hardy-Weinberg proportions obtain locally. Defining $\et_1=\et_4=-\et_2=-\et_3=1$ and proceeding as in Nagylaki (1975, 1989), we derive the following diffusion approximation for the evolution of gamete frequencies:
\begin{subequations}\label{dynamics_pi}
\begin{equation}
	\dot p_i = \frac{\si^2}{2}p_i'' + p_i(w_i-\wb) - \et_i r D \,,\quad i\in\{1,2,3,4\}\,, \label{dynamics_pi_a}
\end{equation}
subject to the initial conditions
\begin{equation}
	0\le p_i(x,t)\le 1\,,\quad  \sum_{i=1}^4 p_i(x,t)=1 \quad\text{for $t=0$ and every $x\in\Reals$} \label{dynamics_pi_b}
\end{equation}
\end{subequations}
(cf.\ Slatkin 1975). Solutions satisfy the constraints \eqref{dynamics_pi_b} for every $t\ge0$.
 
Throughout, we assume absence of epistasis. Then the genotypic fitnesses can be written as
\begin{equation}\label{fitscheme}
\begin{tabular}{c|ccc}
	&  $BB$ & $Bb$ & $bb$\\
\hline
	$AA$ & $\a(x)+\be(x)$ & $\a(x)+\dom_B\be(x)$ & $\a(x)-\be(x)$\\
	$Aa$ & $\dom_A\a(x)+\be(x)$ & $\dom_A\a(x)+\dom_B\be(x)$ & $\dom_A\a(x)-\be(x)$\\
	$aa$ & $-\a(x)+\be(x)$ & $-\a(x)+\dom_B\be(x)$ & $-\a(x)-\be(x)$\\
\end{tabular}\;,
\end{equation}
where it is easy to show that for a continuous-time model this scaling is general because absence of epistasis is assumed (Appendix \ref{app:fitness}). 
We could have introduced spatially dependent dominance coefficients, $\dom_A(x)$ and $\dom_B(x)$. However, in view of our applications, we refrained from doing so. In order to have unique single-locus clines (Fife and Peletier 1981), we assume throughout
\begin{equation}\label{intermediate_dom}
	-1 \le \dom_A \le 1 \quad\text{and}\quad -1 \le \dom_B \le 1\,.
\end{equation}
Dominance is absent at locus $\A$ ($\B$) if $\dom_A=0$ ($\dom_B=0$). 

For our purposes, it will be convenient to follow the evolution of the allele frequencies $p_A=p_1+p_2$ and $p_B=p_1+p_3$, and the linkage disequilibrium $D$. With the abbreviations
\begin{equation}\label{vth}
	\vth_A = 1+\dom_A-2\dom_Ap_A\,,\quad \vth_B = 1+\dom_B-2\dom_Bp_B\,, 
\end{equation}
straightforward calculations yield 
\begin{subequations}\label{w_i-wb}
\begin{align}
	w_1(x,\mathbf p) - \wb(x,\mathbf p) &= \a(x)(1-p_A)\vth_A + \be(x)(1-p_B)\vth_B\,, \\
	w_2(x,\mathbf p) - \wb(x,\mathbf p) &= \a(x)(1-p_A)\vth_A - \be(x)p_B \vth_B \,, \\
	w_3(x,\mathbf p) - \wb(x,\mathbf p) &= -\a(x)p_A \vth_A  + \be(x)(1-p_B)\vth_B \,, \\
	w_4(x,\mathbf p) - \wb(x,\mathbf p) &= -\a(x)p_A \vth_A - \be(x)p_B \vth_B \,.
\end{align}
\end{subequations}

Now it is easy to show that the system of differential equations \eqref{dynamics_pi_a} with the fitnesses \eqref{w_i-wb} is equivalent to
\begin{subequations}\label{eq:ABD_add}
\begin{align}
    \dot p_A &= p_A'' + \la\a(x) p_A(1-p_A)\vth_A + \la\be(x)\vth_B D\,, \label{eq:ABD_a} \\
    \dot p_B &= p_B'' + \la\be(x) p_B(1-p_B)\vth_B + \la\a(x)\vth_A D\,, \label{eq:ABD_b} \\
    \dot D  &= D'' + 2p_A' p_B'+ \la[\a(x)(1-2p_A)\vth_A + \be(x)(1-2p_B)\vth_B]D - \rh D \,, \label{eq:ABD_D}
\end{align}
\end{subequations}
where time has been rescaled such that 
\begin{equation}\label{scaled_pars}
	\la=2/\si^2  \quad\text{and}\quad  \rh=r\la=2r/\si^2\,.
\end{equation}
The constraints \eqref{dynamics_pi_b} on the $p_i$ are transformed to
\begin{subequations}\label{constraints_twoloc}
\begin{equation}\label{constraint_pA_pB}
	0\le p_A\le1\,, \;  0\le p_B\le1\,,
\end{equation}
and 
\begin{equation}\label{constraint_D}
	-\min\{p_A p_B,(1-p_A)(1-p_B)\} \le D \le \min\{p_A (1-p_B),(1-p_A)p_B\}\,,
\end{equation}
\end{subequations}
where these inequalities hold for every $(x,t)\in\Reals\times[0,\infty)$. We will impose the boundary conditions
\begin{equation}\label{boundary_cond1}
	p_A'(\pm\infty,t) =0\,,\; p_B'(\pm\infty,t) =0\,, \text{ and } D'(\pm\infty,t) =0\,, \text{ for every } t\ge0.
\end{equation}

Following Barton and Shpak (2000), we offer the following interpretation of the terms in \eqref{eq:ABD_D}. The first term represents diffusion of linkage disequilibrium; the second term is due to migration, which mixes populations with different allele frequencies, and therefore generates associations between the loci by the Wahlund effect; the third term arises from the indirect effects on $D$ of direct selection on each locus; the fourth term describes the decay of $D$ caused by recombination.

\begin{remark}\rm
Because we assume absence of epistasis and Hardy-Weinberg proportions, only marginal selection coefficients matter. Therefore, the above model with constant selection on genotypes and dominance is equivalent to a model with linear frequency-dependent selection on genotypes without dominance. For a model with dominance and linear frequency dependence, see Mallet and Barton (1989).
\end{remark}

Now we specialize fitness further. For the rest of this paper, we assume a so-called step environment. More precisely, we assume that $\a(x)$ and $\be(x)$ are the following step functions on the real line:
\begin{subequations}\label{step_functions}
\begin{equation}\label{step_function_alpha}
	\a(x)= \begin{cases} \hphantom{-}\a_+ \;&\text{if } x\ge0\,, \\
						 -\a_- \;&\text{if } x<0\,, \end{cases}
\end{equation}
and
\begin{equation}\label{step_function_beta}
	\be(x)= \begin{cases} \hphantom{-}\be_+ \;&\text{if } x\ge0\,, \\
						 -\be_- \;&\text{if } x<0\,, \end{cases}
\end{equation}
where we posit 
\begin{equation}\label{sign_ab}
	\a_+ >0, \; \a_- >0, \; \be_+ >0, \; \be_- > 0.
\end{equation}
\end{subequations}
Under the assumption that both functions change sign once and at the same spatial location, this is general. For instance, the case $\be_+<0<-\be_-$ is obtained by relabeling the alleles $B$ and $b$. Under the assumptions \eqref{step_functions}, selection favors alleles $A$ and $B$ on $[0,\infty)$, and $a$ and $b$ on $(-\infty,0)$; cf.\ Appendix \ref{app:fitness}. Without loss of generality, we could set $\a_+=1$ by rescaling $\la$, but we refrain from doing so. 

By elliptic regularity (Gilbarg and Trudinger 2001), stationary solutions $(p_A,p_B,D)$ of \eqref{eq:ABD_add} with \eqref{step_functions} are continuous and have continuous first derivatives. In addition, $(p''_A,p''_B,D'')$ is continuous and uniformly bounded on every set $(-\infty,-\de)\cup(\de,\infty)$, where $\de>0$, but discontinuous at $x=0$. We refer to a non-constant stationary solution $(p_A,p_B,D)$ of \eqref{eq:ABD_add} as a two-locus cline, and we call $p_A$ and $p_B$ the (marginal) one-locus allele-frequency clines. 

Throughout, if we write $f(\pm\infty)=f^\pm$ then both limits $f^+=\lim_{x\to\infty}f(x)$ and $f^-=\lim_{x\to-\infty}f(x)$ exist and are finite. Furthermore, each equation in which plus and minus signs appear superimposed holds if either every upper or every lower sign is chosen.

For stationary solutions of \eqref{eq:ABD_add} with \eqref{step_functions}, we impose the boundary conditions
\begin{equation}\label{boundary_cond1a}
	p_A'(\pm\infty) =0\,,\; p_B'(\pm\infty) =0\,,\;D'(\pm\infty) =0\,,
\end{equation}
which are a consequence of \eqref{boundary_cond1}.
Under the assumption that the limits $p_A'(\pm\infty)$, $p_A''(\pm\infty)$, $p_B'(\pm\infty)$, $p_B''(\pm\infty)$, $D'(\pm\infty)$, and $D''(\pm\infty)$ exist, \eqref{boundary_cond1a} is easy to show, and a two-locus cline satisfies 
\begin{equation}\label{boundary_cond0}
	p_A(+\infty)=1\,,\; p_A(-\infty)=0\,,\; p_B(+\infty)=1\,,\; p_B(-\infty)=0\,,\text{ and } D(\pm\infty) =0
\end{equation}
(provided $r>0$); see also Remark \ref{boundary_cond_1loc}. 
We conjecture that the stationary solutions of \eqref{eq:ABD_add} with \eqref{step_functions} obeying the constraints \eqref{constraints_twoloc} automatically satisfy \eqref{boundary_cond1a} and \eqref{boundary_cond0}. 

\begin{remark}\rm
We observe from \eqref{eq:ABD_a} or \eqref{eq:ABD_b} that if one locus is fixed, the cline at the other locus is independent of which allele is fixed. Formally, the reason is that if one locus is fixed, then $D=0$. This independence of the other locus is a specific property of continuous-time models. In discrete-time models, the allele frequency at migration-selection equilibrium depends on the fitness of fixed alleles at other loci (see e.g., Aeschbacher and B\"urger 2014, where both discrete-time and continuous-time continent-island models are treated). This difference results from the fast decay of higher-order terms when approximating a discrete-time model by a continuous-time model.
\end{remark}
 
The following simple observation of an invariance property of (time-dependent) solutions of \eqref{eq:ABD_add} with \eqref{step_functions} will be useful later and provide an explanation for the qualitatively different conditions of existence and stability of clines in unbounded and bounded domains with a step environment. For arbitrary, fixed $\a_\pm$ and $\be_\pm$, let 
\begin{equation}\label{pi}
	\pi(x,t;\la,\rh)=(p_A(x,t;\la,\rh),p_B(x,t;\la,\rh),D(x,t;\la,\rh))
\end{equation}
denote the time-dependent solution of \eqref{eq:ABD_add} with \eqref{step_functions}, where the dependence on $\la$ and $\rh$ is shown explicitly. 

\begin{lemma}\label{lem:invariance}
The following statements hold for arbitrary positive constants $c$, $c_1$, $c_2$, for arbitrary (positive) parameters $\la$ and $\rh$, and for every $(x,t)\in\Reals\times[0,\infty)$.

{\rm(i)}
\begin{equation}\label{invariance_c}
	\pi(x,t;c\la,c\rh)=\pi(\sqrt{c}x,ct;\la,\rh)
\end{equation}
and
\begin{equation}\label{invariance}
	\pi(x,t; c_1\la,c_2\rh) = \pi(\sqrt{c_1}x,c_1t;\la,\frac{c_2}{c_1}\rh) =  \pi(\sqrt{c_2}x,\frac{c_2}{c_1}t;\frac{c_1}{c_2}\la,\rh) \,.
\end{equation}

{\rm(ii)} If $\pi(x;\la,\rh)$ is a two-locus cline for $\la$ and $\rh$, then a two-locus cline exists for $c\la$ and $c\rh$. It is given by $\pi(x;c\la,c\rh)=\pi(\sqrt{c}x;\la,\rh)$.

{\rm(iii)} If $\pi(x;\la,\rh)$ is a globally asymptotically stable two-locus cline for $\la$ and $\rh$, then $\pi(x;c\la,c\rh)=\pi(\sqrt{c}x;\la,\rh)$ is a globally asymptotically stable two-locus cline for $c\la$ and $c\rh$. 
\end{lemma}

\begin{proof}
(i) Let $\tilde\pi(x,t) = \pi(\sqrt{c}x,ct;\la,\rh)$ and use the abbreviation $z=(\sqrt{c}x,ct;\la,\rh)$. Then, by the chain rule, $\tilde\pi'(x,t) = \sqrt{c}\pi'(z)$, $\tilde\pi''(x,t) = c\pi''(z)$, and $\dot{\tilde\pi}(x,t) = c\dot\pi(z)$. From these identities, the fact that $\a(x)$ and $\be(x)$ are step functions with step at 0, and by applying \eqref{eq:ABD_add}, we obtain
\begin{align}
	\dot{\tilde D}(x,t) &= c \dot D(z) \notag \\
		&=c\Bigl\{D''(z) + 2p_A'(z)p_B'(z) \notag \\
		&\quad + \la[\a(\sqrt{c}x)(1-2p_A(z))\vth_A(z)+\be(\sqrt{c}x)(1-2p_B(z))\vth_B(z)]D(z) - \rh D(z)  \Bigr\} \notag \\
		&=\tilde D''(x,t) + 2\tilde p_A'(x,t)\tilde p_B'(x,t) \notag \\
		&\quad +  c\la[\a(x)(1-2\tilde p_A(x,t))\tilde \vth_A(x,t)+\be(x)(1-2\tilde p_B(x,t))\tilde \vth_B(x,t)]\tilde D(x,t) \notag \\
		&\quad - c\rh \tilde D(x,t)\,.		
\end{align}
Analogous calculations hold for $\tilde p_A$ and $\tilde p_B$.
Because solutions are unique (the proof of uniqueness in Theorem 1 of Kolmogoroff et al.\ 1937 does not require the Lipschitz condition in $x$), we have $\tilde\pi(x,t)=\pi(x,t;c\la,c\rh)$. This proves \eqref{invariance_c}, and \eqref{invariance} follows immediately. 

(ii) follows from (i) because clines are stationary solutions.

(iii) follows because (i) establishes a one-to-one relation between the domains of attraction of the two stationary solutions $\pi(x;\la,\rh)$ and $\pi(x;c\la,c\rh)$. 
\end{proof}

\begin{remark}\rm
Assume that $(p_A,p_B,D)$ is a two-locus cline and $D(x)\ge0$ holds everywhere. Then \eqref{eq:ABD_a} implies that $p_A'$ is strictly monotone increasing on $(-\infty,0)$ and strictly monotone decreasing on $[0,\infty)$. Because the cline satisfies $p_A'(\pm\infty)=0$, it follows that $p_A'(x)$ is maximized at $x=0$, and minimized at $x=\pm\infty$. An analogous statement holds for $p_B$. We conjecture that $D(x)>0$ on $\Reals$ holds for every possible parameter combination satisfying our assumptions.
\end{remark}

\section{Single-locus clines with dominance}\label{sec:one_locus}
For a step environment on the real line, Haldane (1948) derived explicit solutions for a single-locus cline if dominance is either absent ($\dom=0$) or complete ($\dom=-1$). We admit spatially constant dominance coefficients satisfying 
\begin{equation}\label{dom_cond_1loc}
	-1\le\dom\le1
\end{equation}
and assume that $\a(x)$ is given by \eqref{step_function_alpha}. Throughout this section, we write $\dom$ instead of $\dom_A$. Then the single-locus cline is the solution of 
\begin{subequations}\label{1-locus problem}
\begin{equation}\label{1-locus ODE}
	P'' + \la\a(x) P(1-P)(1+\dom-2\dom P) = 0
\end{equation}
satisfying
\begin{equation}\label{0lePle1}
	0 <P(x) < 1\,,
\end{equation}
and obeying the boundary conditions
\begin{equation}\label{P(infty)}
	P(-\infty)=0 \,,\; P(+\infty) =1\,,
\end{equation}
and
\begin{equation}\label{Pprime(infty)}
	P'(-\infty)=P'(+\infty) =0\,.
\end{equation}
\end{subequations}

\begin{remark}\label{boundary_cond_1loc}\rm
Solutions of \eqref{1-locus ODE} and \eqref{0lePle1} satisfy the boundary conditions \eqref{P(infty)} and \eqref{Pprime(infty)} automatically: Let, e.g., $x>0$. Then \eqref{1-locus ODE} implies $P''(x)<0$ for every $x>0$. Because $P(x)$ is bounded by \eqref{0lePle1}, we infer that $P''(+\infty)=0$ and that $P'(x)$ is monotone decreasing in $x$. Since $P'(x)$ must be bounded from below, the limit $P'(+\infty)$ exists and equals 0. Since $P''(+\infty)=0$, \eqref{1-locus ODE} yields $P(+\infty)=0$ or $P(+\infty)=1$. The properties of $P'$, or $P''$, imply the latter.
\end{remark}
Our goal is to derive an explicit solution for this cline. We start by collecting some simple facts. 

From \eqref{1-locus ODE}, \eqref{0lePle1}, and \eqref{dom_cond_1loc}, we obtain
\begin{equation}\label{P''(x)}
	P''(x)  \begin{cases} < 0 \;&\text{ if } x>0\,, \\
						  > 0 \;&\text{ if } x<0\,.
			\end{cases}
\end{equation}
From \eqref{P(infty)}, \eqref{0lePle1}, and \eqref{P''(x)}, we infer that $P'(x)$ is maximized at $x=0$ and is strictly monotone declining to 0 as $x\to\pm\infty$, i.e.,
\begin{equation}
	P'(\pm\infty) = 0\,.
\end{equation}

To state the main result of this section in compact form, we define
\begin{subequations}\label{abbrevs_Adom}
\begin{equation}\label{a_pm}
	a_+ = \sqrt{\la\a_+}\,, \quad a_- = \sqrt{\la\a_-}\,,
\end{equation}
\begin{alignat}{2}
	\xix_+ &= x a_+\sqrt{1-\dom} &\quad&\text{ if } \dom<1\,, \label{xix+}\\
	\xix_- &= x a_-\sqrt{1+\dom} &\quad&\text{ if } \dom>-1\,, \label{xix-}
\end{alignat}
\begin{equation}\label{A_pm_dom}
	A_+ = F_+(a_0,\dom) \,,\quad A_- = F_-(a_0,\dom) \,,
\end{equation}
where
\end{subequations}
\begin{subequations}\label{F_pm}
\begin{align}
	F_+(y,\dom) &= \begin{cases}
		\dfrac{2+(1-3\dom)y+\sqrt3\sqrt{1-\dom}\sqrt{\phi(1-y,-\dom)}}{1-y} \quad&\text{if $\dom<1$}\,, \\
		\dfrac{\sqrt3\sqrt{1+3y}}{2\sqrt{1-y}}&\text{if $\dom=1$}\,,
	\end{cases} \label{F_+}\\
	F_-(y,\dom) &= \begin{cases}
		1/F_+(1-y,-\dom) \quad&\text{if $\dom>-1$}\,, \\
		F_+(1-y,1) \quad&\text{if $\dom=-1$}\,,
		\end{cases} \label{F_-}
\end{align}
\end{subequations}
\begin{equation}\label{phi}
	\phi(y,\dom)=3-2y+3\dom(1-y)^2\,,
\end{equation}
and $a_0$ is the unique solution in $(0,1)$ of the quartic equation 
\begin{equation}\label{get_a0_dom}
	a_0^2\phi(a_0,\dom) = \frac{\a_+}{\a_++\a_-}\,.
\end{equation}
The functions $F_\pm$ are well defined and positive if $0<y<1$. They are strictly monotone increasing in $y$.

In addition, it will be convenient to write
\begin{equation}\label{Z_pm}
	Z_\pm(x) = A_\pm e^{X_\pm} \,. 
\end{equation}

\begin{theorem}\label{thm:1-locus_cline}
Let $\a(x)$ be given by \eqref{step_function_alpha} and assume \eqref{dom_cond_1loc}.
Then there exists a unique continuously differentiable solution $P(x)$ of \eqref{1-locus problem}. If $x\ge0$, it is given by
\begin{subequations}\label{P(x)_dom}
\begin{equation}\label{P(x+)_dom}
	P(x) = \begin{cases}
			1 - \dfrac{6(1-\dom)}{Z_+ + 2(1-3\dom) + (1+3\dom)Z_+^{-1}} \quad&\text{if $\dom<1$}\,,\vspace{.3cm} \\
			1 - \dfrac{12}{9+4(xa_++A_+)^2}\quad&\text{if $\dom=1$}\,. 
	\end{cases}
\end{equation}
If $x<0$, it is given by
\begin{equation}\label{P(x-)_dom}
	P(x) = \begin{cases}
			\dfrac{6(1+\dom)}{(1-3\dom)Z_- + 2(1+3\dom) + Z_-^{-1}} \quad&\text{if $\dom>-1$}\,, \vspace{.3cm}\\
		\dfrac{12}{9+4(xa_--A_-)^2}	\quad&\text{if $\dom=-1$} \,. 	
	\end{cases}
\end{equation}
\end{subequations}
\end{theorem}

\begin{proof}
Following Haldane's method, we multiply \eqref{1-locus problem} by $2P'$ to obtain
\begin{equation}\label{ODE_P'^2}
	\frac{d}{dx}(P')^2 = -2\la\a(x) P(1-P)(1+\dom-2\dom P)P'\,.
\end{equation}
From \eqref{phi}, we obtain
\begin{equation}
	\pder{}{y}(y^2\phi(y,\dom))  = 6y(1-y)(1+\dom-2\dom y)\,.
\end{equation}
 
Let $x\ge0$, so that $\a(x)=\a_+$. Separating variables in \eqref{ODE_P'^2} and integrating, we find 
\begin{subequations}\label{ph^2}
\begin{equation}\label{ph^2_right}
	[P'(x)]^2 = C_1 - \frac{\la\a_+}{3}P(x)^2\phi(P(x),\dom) \,.
\end{equation}
If $x<0$, then $\a(x)=-\a_-$ and we obtain in an analogous manner
\begin{equation}\label{ph^2_left}
	[P'(x)]^2 = C_2 + \frac{\la\a_-}{3}P(x)^2\phi(P(x),\dom) \,.
\end{equation}
\end{subequations}
Because the boundary conditions are $P'(x)\to0$ and $P(x)\to1$ as $x\to\infty$, and $P'(x)\to0$ and $P(x)\to0$ as $x\to-\infty$, we find $C_1=\la\a_+/3$ and $C_2=0$. 

Since we require continuity of $P'$ at $x=0$, we equate the right-hand sides of \eqref{ph^2_right} and \eqref{ph^2_left} to obtain
\begin{equation}
	\frac{\la\a_+}{3} - \frac{\la\a_+}{3}P(0)^2\phi(P(0),\dom) =  \frac{\la\a_-}{3}P(0-)^2\phi(P(0-),\dom)\,.
\end{equation}
Setting
\begin{equation}\label{a_0}
	a_0=P(0)=P(0-)
\end{equation}
and simplifying, we infer that $a_0$ is the unique solution in $(0,1)$ of \eqref{get_a0_dom}. 

Substituting \eqref{get_a0_dom} into \eqref{ph^2_left} with $C_2=0$ and $x=0$, we find
\begin{equation}\label{P'(0)}
	P'(0) = \frac{\sqrt{\la}}{\sqrt3}\,\sqrt{\frac{\a_+\a_-}{\a_+ +\a_-}}
				= \frac{\sqrt{\la}}{\sqrt6}\,\sqrt{H(\a_+,\a_-)} = \frac{1}{\sqrt6}\sqrt{H(a_+^2,a_-^2)}\,,
\end{equation}
where $H(\a_+,\a_-)$ denotes the harmonic mean of $\a_+$ and $\a_-$. 

Because the constants in \eqref{ph^2} are $C_1=\la\a_+/3=a_+^2/3$ and $C_2=0$, we have to solve the set of differential equations
\begin{subequations}
\begin{alignat}{2}
	P' &= \frac{a_+}{\sqrt3}(1-P)\sqrt{\phi(1-P,-\dom)} &\qquad&\text{if } x\ge0 \,, \label{P'_right}\\
	P' &= \frac{a_-}{\sqrt3}P\sqrt{\phi(P,\dom)} &\qquad&\text{if } x<0  \label{P'_left}\,.
\end{alignat}
\end{subequations}
By separating variables in \eqref{P'_right}, multiplying by $\sqrt3\sqrt{1-h}$ if $\dom<1$ and by $\sqrt3$ if $\dom=1$, integrating and observing
\begin{subequations}
\begin{alignat}{2}
		\pder{}{y} \ln F_+(y,\dom) &= \frac{\sqrt3\sqrt{1-\dom}}{(1-y)\sqrt{\phi(1-y,-\dom)}} &\quad&\text{if } \dom<1\,, \\
		\pder{}{y} F_+(y,1) &=   \frac{\sqrt3}{(1-y)\sqrt{\phi(1-y,-1)}}  &\quad&\text{if } \dom=1\,,
\end{alignat}
\end{subequations}
we arrive at
\begin{subequations}\label{int_parts1}
\begin{alignat}{2}
		\ln F_+(P(x)) &= c + x a_+\sqrt{1-\dom}  &\qquad&\text{if } \dom<1\,, \\
		F_+(P(x)) &= c + x a_+  &\qquad&\text{if } \dom=1\,.
\end{alignat}
\end{subequations}
The constant $c$ has to satisfy $c=\ln F_+(a_0)$ if $\dom<1$, and $c=F_+(a_0)$ if $\dom=1$. Using the abbreviations \eqref{abbrevs_Adom}, we can rewrite \eqref{int_parts1} as 
\begin{subequations}\label{int_parts2}
\begin{alignat}{2}
	F_+(P(x)) &= A_+e^{\xix_+}  &\qquad&\text{if } \dom<1\,, \\
	F_+(P(x)) &= \xix_+ + A_+  &\qquad&\text{if } \dom<1\,.
\end{alignat}
\end{subequations}
Now, tedious but straightforward inversion of $F_+$ yields \eqref{P(x+)_dom}.

To derive \eqref{P(x-)_dom}, we observe that if $P(x,\dom,a_+)$ solves \eqref{P'_right}, then $\tilde P(x)=1-P(-x,-\dom,a_-)$ solves \eqref{P'_left}. This follows from the fact that the right-hand side of \eqref{P'_left} is obtained from that of \eqref{P'_right} by the simultaneous substitutions $P\to1-P$, $\dom\to-\dom$, and $a_+\to a_-$.
Therefore, 
\begin{equation}
	\pder{}{y} \ln F_-(y,\dom) = \pder{}{y}( -\ln F_+(1-y,-\dom)) = \frac{\sqrt3\sqrt{1+\dom}}{y\sqrt{\phi(y,\dom)}}\,,
\end{equation}
which implies \eqref{P(x-)_dom} if $\dom>-1$ because it is obtained from \eqref{P(x+)_dom} by the above substitutions. A similar argument yields \eqref{P(x-)_dom} if $\dom=-1$.
The adjustment of integration constants must be such that $P(0)=P(0-)=a_0$; see \eqref{a_0}. This is achieved by choosing $A_+$ and $A_-$ according to \eqref{A_pm_dom}.

Uniqueness follows from the above derivation.
\end{proof}

\begin{corollary}\label{cor:properties of P}
{\rm (i)} The slope of the cline in its center, $P'(0)$, given by \eqref{P'(0)}, is independent of the degree of dominance.

{\rm (ii)} For fixed step size $\a_++\a_-$ and $\la$, or fixed $a_+^2+a_-^2$, $P'(0)$ is maximized if $\a_+=\a_-$, and $P'(0)$ decays to zero as $\a_+\to0$ (or as $\a_-\to0$). This is in accordance with intuition because in these limits allele $A$ (or $a$) is nowhere advantageous and the cline vanishes.

{\rm (iii)} For fixed $a_+$ and $a_-$, $P(0)=a_0$, given by the solution of \eqref{get_a0_dom} in $(0,1)$, is strictly monotone decreasing in $h$. 

{\rm (iv)} For fixed $h$, $P(0)=a_0$ is strictly monotone decreasing in $\tilde\a=\a_+/(\a_++\a_-)= a_+^2/(a_+^2+a_-^2)$. 
\end{corollary}

\begin{proof}
(i) is evident from \eqref{P'(0)}. (ii) is also immediate. (iii) follows from implicit differentiation of \eqref{get_a0_dom}: Define
$f(h,a_0)=a_0^2\phi(a_0,\dom) - \frac{\a_+}{\a_++\a_-}$. Then 
\begin{equation}\label{deriv_a0_h}
	\frac{da_0}{dh} = \frac{-\partial f/\partial h}{\partial f/\partial a_0} = \frac{-a_0(1-a_0)}{2[1+(1-2a_0)h]} < 0
\end{equation}
because $0<a_0<1$ and $-1\le h \le 1$. (iv) is shown analogously.
\end{proof}

Figure \ref{fig:plot_dom_1loc} illustrates the dependence of the cline on the sign and degree of dominance. It shows that the slope at $x=0$ is independent of dominance, but it also demonstrates that dominance has a strong influence on the shape, the value $P(0)$ in the center, and the asymptotic properties of the cline (see also Section \ref{sec:asymptotics}). In particular, (environment-independent) dominance leads to remarkable asymmetries of the clines.

\begin{figure}[t]
	\centering
	\includegraphics[width=0.6\textwidth]{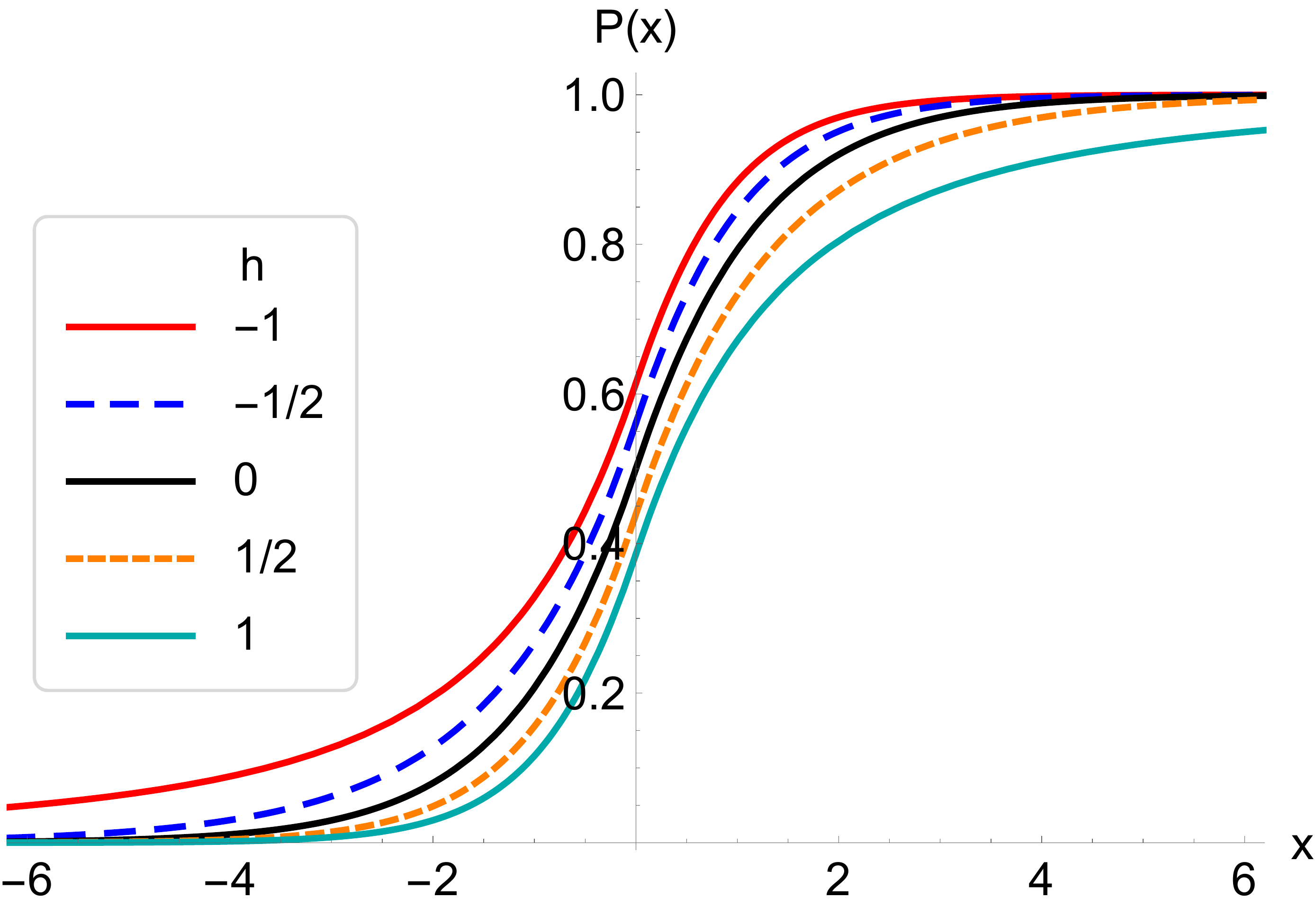}
\caption{The single-locus cline $P(x)$ in \eqref{P(x)_dom} as a function of $x$ for different values of the dominance parameter $\dom$. The other parameters are $a_+=a_-=1$.}\label{fig:plot_dom_1loc}
\end{figure}

\begin{remark}\label{rem:width}\rm
The inverse of the maximum slope, in the present model $P'(0)$, is often called the width of the cline. If $\a_+=\a_-$, the width is proportional to Slatkin's (1973) characteristic length for allele frequency variation, which, in our notation and scaling, becomes $\si/\sqrt{\a_+}=\sqrt2/\sqrt{\la \a_+}$. To extend this concept to asymmetric selection, i.e.,  to $\a_+\ne\a_-$, we define the characteristic length for locus $A$ by
\begin{equation}\label{cline_length}
	\ell_A = \frac{\si}{\sqrt{H(\a_+,\a_-)}} = \frac{\sqrt2}{\sqrt{\la H(\a_+,\a_-)}} \,.
\end{equation}
By \eqref{P'(0)} the width of the cline is
\begin{equation}\label{cline_width}
	\om_A = 1/P'(0) = \sqrt3\, \ell_A\,.
\end{equation}
Since $a_0$ does not depend on $\la$, \eqref{P(x)_dom} informs us immediately that $P$ satisfies the invariance property
\begin{equation}\label{invariance_P}
	P(x,c\la) = P(\sqrt{c}x,\la)\,;
\end{equation}
cf.\ Lemma \ref{lem:invariance}. This invariance property lies at the heart of the definition of a characteristic length.
\end{remark}

Although $P$ and $P'$ are continuous everywhere, in particular at $x=0$, the second derivative of $P$ at $x=0$ is not continuous because 
\begin{subequations}\label{P''(0)}
\begin{align}
	P''(0+) &= -\la\a_+ a_0(1-a_0)(1+\dom-2\dom a_0)<0 \,, \\
	P''(0-) &= \la\a_- a_0(1-a_0)(1+\dom-2\dom a_0)>0\,.
\end{align}
\end{subequations}
Finally, we note that
\begin{equation}
	\frac{d^n P}{dx^n}(\pm\infty)=0 \label{cline_cond1}
\end{equation}
for derivatives of every order $n\ge1$.

For later use, we note that if $-1<\dom<1$, $P'$ can be written as
\begin{equation}\label{P'(x)_dom}
	P'(x) = \begin{cases} 6a_+(1-\dom)^{3/2}\dfrac{Z_+-(1+3\dom)Z_+^{-1}}{[Z_++2(1-3\dom)+(1+3\dom)Z_+^{-1}]^2} &\quad\text{if } x\ge0\,, \vspace{.3cm} \\
			6a_-(1+\dom)^{3/2}\dfrac{Z^{-1}-(1-3\dom)Z_-}{[(1-3\dom)Z_- +2(1+3\dom)+Z_-^{-1}]^2} &\quad\text{if } x<0\,. 
	\end{cases}
\end{equation}

In the absence of dominance ($h=0$), we obtain the well known clinal solution (Haldane 1948, Nagylaki 1976)
\begin{equation}\label{P(x)}
	P(x) = 
	\begin{cases}
			-\dfrac12 + \dfrac32 \tanh^2[\tfrac12(\xix_++\ln A_+)]  &\; \text{if } x\ge0\,,\vphantom{\biggl[} \\
			 \hphantom{-}\dfrac32 - \dfrac32 \tanh^2[\tfrac12(\xix_--\ln A_-)]  &\; \text{if } x<0\,,\vphantom{\biggl[} 
	\end{cases}
\end{equation}
where 
\begin{subequations}\label{abbrevs_A}
\begin{equation}\label{A_pm}
	\ln A_+ = 2\tanh^{-1}\sqrt{\frac{1+2a_0}{3}} \,,\quad \ln A_- = 2\tanh^{-1}\sqrt{\frac{3-2a_0}{3}}\,, 
\end{equation}
and $a_0$ is the unique solution in $(0,1)$ of the cubic equation
\begin{equation}\label{get_a0}
	3a_0^2-2a_0^3 = \frac{\a_+}{\a_+ + \a_-} = \frac{a_+^2}{a_+^2+a_-^2}\,.
\end{equation}
\end{subequations}
We observe from \eqref{get_a0} that $a_0=\tfrac12$ if and only if $a_+=a_-$, and $a_0>\tfrac12$ if and only if $a_+>a_-$.

We refer to Slatkin (1973) and Nagylaki (1975, 1976) for related results on other specific environmental functions $\a(x)$, and also on the treatment of barriers to gene flow. Conley (1975) proved existence of clines on $\Reals$ for a large class of fitness functions, including the step function treated above. As noted by Nagylaki (1975), his proof applies to arbitrary dominance. Fife and Peletier (1977) proved existence, uniqueness, and global asymptotic stability of clines under rather general assumptions about selection. Although they assume that $\a'(x)$ is uniformly bounded, the proof of their local stability result (their Theorem 3) applies without this assumption if the cline exists and is unique. This is the case in the present model, whence a form of local asymptotic stability of the single-locus cline $P(x)$ is established if $-1\le\dom\le1$. I am grateful to Dr.\ Linlin Su for pointing this out. Whether their global stability result, Theorem 4, applies seems to be open.

\section{Strong recombination}\label{sec:strong_rec}
Now we return to the two-locus problem \eqref{eq:ABD_add} with \eqref{step_functions} and assume that
\begin{equation}\label{ep}
	\ep=1/\rh
\end{equation}
is sufficiently small. Our aim is to derive a first-order approximation in $\ep$ for the two-locus cline, i.e., for the stationary solution that satisfies \eqref{constraints_twoloc} and \eqref{boundary_cond1a}. Under our assumptions \eqref{step_functions} about the step environment, such an equilibrium solution is expected to exist whenever both single-locus problems have a cline, because linkage will lead to additional, indirect, selection on each locus, so that the loci will reinforce each other and the cline will persist. Because we assume the unbounded domain $\Reals$, the single-locus clines, $P$ and $Q$, exist always and are uniquely determined. The main result is Theorem \ref{thm:p(x)}. In population-genetic terms, this theorem provides the quasi-linkage-equilibrium approximation to the two-locus cline. 

\subsection{Approximation of the two-locus cline}\label{sec:approx_twoloc}
We write $P$ and $Q$ for the single-locus clines in the allele-frequencies of $A$ and $B$, respectively. They are obtained from \eqref{P(x)_dom}. For given $\ep$, we write $(p_A,p_B,D)$ for the cline solution of \eqref{eq:ABD_add}. We want to determine solutions of the form
\begin{subequations}\label{approx_eps}
\begin{equation}\label{approx_eps1}
	p_A = p_A^{(\ep)} + O(\ep^2)\,, \quad p_B = p_B^{(\ep)} + O(\ep^2)\,, \quad D = D^{(\ep)} + O(\ep^2)\,, 
\end{equation}
where
\begin{equation}\label{approx_eps2}
	p_A^{(\ep)} = P + \ep p\,, \quad p_B^{(\ep)} = Q + \ep q \,, \quad D^{(\ep)} = \ep d\,. 
\end{equation}
\end{subequations}
Here, the limit $\ep\to0$ is assumed to hold uniformly in $x\in\Reals$ for $(p_A,p_B,D)$ and for its first and second derivatives, with the restriction that at $x=0$ only left and right second derivatives exist; cf.\ \eqref{P''(0)}. From the boundary conditions \eqref{boundary_cond1a} and the fact that $P'(\pm\infty)=Q'(\pm\infty)=0$, we obtain the boundary conditions 
\begin{equation}\label{boundary_cond_pqd}
	p'(\pm\infty)=q'(\pm\infty)=d'(\pm\infty)=0\,.
\end{equation}

Before formulating our main result, we need additional notation and definitions. Let
\begin{subequations}\label{IpIm_gen}
\begin{align}
	I_+(x) &= \int_x^\infty [P'(y)]^2 Q'(y)[1+\dom_B-2\dom_B Q(y)]\, dy\,, \\
	I_-(x) &= \int_{-\infty}^x [P'(y)]^2 Q'(y)[1+\dom_B-2\dom_B Q(y)]\, dy\,.
\end{align}
\end{subequations}
We recall from Section \ref{sec:one_locus} that $P'(x)>0$ and $Q'(x)>0$ on $\Reals$. Therefore, $I_+(x)>0$ and $I_-(x)>0$ on $\Reals$. These integrals exist because $P'(y)$ and $Q'(y)$ decay rapidly as $y\to\pm\infty$ (see Theorem \ref{thm:1-locus_cline}, and Section \ref{sec:asymptotics} for asymptotic results). However, they can be computed explicitly only in special cases (Section \ref{sec:explicit}).

In addition to the abbreviations \eqref{abbrevs_Adom}, we introduce
\begin{equation}\label{b_pm}
	b_+ = \sqrt{\la\be_+}\,, \quad b_- = \sqrt{\la\be_-}\,. 
\end{equation}

\begin{theorem}\label{thm:p(x)}
Assume that for sufficiently small $\ep>0$, \eqref{eq:ABD_add} admits a two-locus cline of the form \eqref{approx_eps}. Then

(i) $d(x)$ is given by
\begin{equation}\label{D_ep}
	 d(x) = 2P'(x) Q'(x)\,.
\end{equation}

(ii) $p(x)$ is the unique continuously differentiable solution on $\Reals$ of the linear problem
\begin{gather}
	p''(x) + \la \a(x)[1-2P(x)+\dom_A(1-6P(x)+6P(x)^2)] p(x)  \notag\\
		+ 2\la \be(x) P'(x)Q'(x)[1+\dom_B-2\dom_B Q(x)] = 0\,. \label{p''inhomog}
\end{gather}

(iii) $p(x)$ is given by
\begin{equation}\label{p(x)}
	p(x) = \begin{cases}
				P'(x) k_+(x) \quad \text{if } x\ge0\,, \\
				P'(x) k_-(x)\quad \text{if } x\le0\,,
			\end{cases}
\end{equation}
where
\begin{subequations}\label{define_k_pm}
\begin{alignat}{2}
	k_+(x) &= k_0 + 2b_+^2\,\int_0^x \frac{I_+(y)}{[P'(y)]^2}\,dy&&\quad\text{if } x\ge0\,, \label{define_k_p}\\
	k_-(x) &= k_0 - 2b_-^2\,\int_x^0 \frac{I_-(y)}{[P'(y)]^2}\,dy &&\quad\text{if } x\le0\,,
\end{alignat}
\end{subequations}
and
\begin{equation}\label{k_0}
	k_0 = \frac{2\sqrt3[b_+^2 I_+(0) - b_-^2 I_-(0)]}{a_0(1-a_0)(1+\dom_A-2\dom_A a_0)a_+a_-\sqrt{a_+^2 + a_-^2}}\,.
\end{equation}

(iv) The first-order term $q(x)$ in the expansion of $p_B(x)$ is of analogous form.
\end{theorem}

\begin{proof}
To deduce the stationary solution, we equate the right-hand sides of \eqref{eq:ABD_add} to zero. 

(i) If $\rh=1/\ep$ and \eqref{approx_eps} holds, then the right-hand side of \eqref{eq:ABD_D} becomes $2P'(x)Q'(x) - d(x) + O(\ep^2)$, which yields \eqref{D_ep}.

(ii) We obtain \eqref{p''inhomog} by substituting \eqref{approx_eps} and \eqref{D_ep} into the right-hand side of \eqref{eq:ABD_a} and observing \eqref{1-locus problem}. The solution of \eqref{p''inhomog} is unique because the homogeneous equation 
\begin{equation}\label{homog_P'}
	p'' + \la \a(x)[1-2P(x)+\dom_A(1-6P(x)+6P(x)^2)] p  = 0
\end{equation}
has, up to constant multiples, the unique solution $P'(x)$. That the latter is a solution follows immediately by differentiating \eqref{1-locus problem} and observing that $\a(x)$ is the step function \eqref{step_function_alpha}.

(iii) Because $P'(x)$ is the appropriate solution of the homogeneous problem \eqref{homog_P'}, the solution $p(x)$ of the inhomogeneous problem can be obtained by variation of constants. Therefore, by substituting \eqref{p(x)} into \eqref{p''inhomog}, using that $P'$ solves \eqref{homog_P'}, and recalling \eqref{b_pm}, we obtain
\begin{equation}\label{k_pm''}
	P'(x)k_\pm''(x) + 2P''(x)k_\pm'(x) \pm 2b_\pm^2 P'(x)Q'(x)[1+\dom_B-2\dom_B Q(x)] = 0\,.
\end{equation}
If we view \eqref{k_pm''} as a set of two first-order equations for $k_\pm'$, one defined for $x\ge0$ and the other for $x<0$, we find by straightforward calculations that 
\begin{equation}\label{k_pm'(x)}
	k_\pm'(x) = 2b_\pm^2 \frac{I_\pm(x)}{[P'(x)]^2}
\end{equation}
is the only special solution that is potentially integrable (the general solution is of the form $k_\pm'(x) +c_\pm/P'(x)$; it diverges if $c_\pm\ne0$). Therefore, integration yields \eqref{define_k_pm}, and $k_+$ and $k_-$ must have the same constant $k_0$ because we require that $p(x)$ be continuous at $x=0$ (hence, everywhere).

It remains to determine $k_0=k_+(0)=k_-(0)$. Because $p(x)$ needs to be continuously differentiable at $x=0$, and because $P'(x)$ is continuous, we require $p'(0+)=p'(0-)$. Differentiation of \eqref{p(x)} shows that this holds if and only if
\begin{equation}
	 P''(0+)k_0 + P'(0)k_+'(0+) = P''(0-)k_0 + P'(0)k_-'(0-)\,.
\end{equation}
Using \eqref{k_pm'(x)}, we obtain 
\begin{equation}\label{k_pm(0)}
	k_0 = \frac{2[b_+^2 I_+(0) - b_-^2 I_-(0)]}{P'(0)[P''(0-)-P''(0+)]}
\end{equation}
and, after applying \eqref{P'(0)} and \eqref{P''(0)}, \eqref{k_0}.

(iv) is obvious.
\end{proof}

The approximation \eqref{D_ep} for the linkage disequilibrium was already derived by Barton and Shpak (2000) in a multilocus context.
Since $P(x)$ and $Q(x)$ are strictly monotone increasing, the leading term $\ep d(x) = 2\ep P'(x)Q'(x)$ of $D(x)$ is strictly positive on $\Reals$. In addition, because $P'(x)$ is positive, maximized at $x=0$, and decays monotonically to zero as $x\to\pm\infty$, the same properties are shared by $d(x)$. This does not imply that the true equilibrium solution $D(x)$ is maximized at $x=0$. If selection is not symmetric, i.e., if $a_+\ne a_-$ or $b_+\ne b_-$, then $D(x)$ will be maximized only in an $\ep$-neighborhood of 0 (see also Fig.\ \ref{fig:D_perturb_accur}). We note that $D'(0)$ exists and is continuous for every $\ep>0$, whereas $d(x)$ has different one-sided derivatives at $x=0$ for every $\ep>0$, as has $D'(x)$.

The integrals appearing in \eqref{define_k_pm} can be computed explicitly only in special cases (Section \ref{sec:explicit}). However, using integration by parts they can be reduced to expressions requiring only one-fold integration (Appendix \ref{app:one-fold_int}). Therefore, it is straightforward to compute $p(x)$ numerically. Asymptotic expansions are derived in Section \ref{sec:asymptotics}.

\begin{remark}\label{rem:kappa_k}\rm
Below, we show that the integrals $\int_0^\infty \frac{I_+(y)}{[P'(y)]^2}\,dy$ and $\int_{-\infty}^0 \frac{I_-(y)}{[P'(y)]^2}\,dy$ are finite; see \eqref{Ipm/P'_exists}. 
Therefore, we observe from \eqref{define_k_pm} that $k_+$ and $k_-$ are bounded and strictly monotone increasing on their domain. Hence, the limits 
\begin{subequations}\label{kpm_limits}
\begin{align}
	\ka_+ &= \lim_{x\to\infty} k_+(x) = k_0 + 2b_+^2\,\int_0^\infty \frac{I_+(y)}{[P'(y)]^2}\,dy\,, \\
	\ka_- &= \lim_{x\to-\infty} k_-(x) = k_0 - 2b_-^2\,\int_{-\infty}^0 \frac{I_-(y)}{[P'(y)]^2}\,dy
\end{align}
\end{subequations}
exist and satisfy $\ka_- < \ka_+$. Therefore, we can rewrite \eqref{define_k_pm} as
\begin{subequations}\label{k_pm_kappa}
\begin{alignat}{2}
	k_+(x) &= \ka_+ - 2b_+^2\,\int_x^\infty \frac{I_+(y)}{[P'(y)]^2}\,dy&&\quad\text{if } x\ge0\,, \\
	k_-(x) &= \ka_- + 2b_-^2\,\int_{-\infty}^x \frac{I_-(y)}{[P'(y)]^2}\,dy &&\quad\text{if } x\le0\,.
\end{alignat}
\end{subequations}

Although $p(x)$ is continuously differentiable at $x=0$, $k_+'(0+)\ne k_-'(0-)$. This follows directly from \eqref{k_pm'(x)}, but also from the structure of $p$ and the fact that $P'$ is not differentiable at $x=0$.
\end{remark}

\subsection{Properties of the approximate two-locus cline}\label{sec:properties_twoloc}
The following corollary summarizes some simple but important consequences of Theorem \ref{thm:p(x)}.

\begin{corollary}\label{Cor:prop_p(0)_gen}
{\rm (i)} We have
\begin{equation}\label{p'(0)_gen}
		p'(0) = \frac{2\sqrt{3}[a_-^2b_+^2 I_+(0) + a_+^2b_-^2 I_-(0) ]}{a_+a_-\sqrt{a_+^2+a_-^2}} > 0\,.
\end{equation}
Therefore, the cline $p_A=P+\ep p + O(\ep^2)$ gets steeper in its center as $\ep$ increases from 0, i.e., if linkage between the two loci gets tighter.

{\rm (ii)} The sign of
\begin{equation}\label{p(0)_gen}
	p(0) = \frac{2[b_+^2 I_+(0) - b_-^2 I_-(0)]}{a_0(1-a_0)(1+\dom_A-2\dom_A a_0)(a_+^2 + a_-^2)}
\end{equation}
can be positive, negative, or zero.

{\rm (iii)} $p(x)$ can be positive on $(-\infty,\infty)$, negative, or change sign once.

{\rm (iv)} $p(x)$ satisfies the invariance property
\begin{equation}\label{invariance_p}
	p(x,c\la) = c p(\sqrt{c}x,\la)\,.
\end{equation}
\end{corollary}

\begin{proof}
The expressions for $p(0)$ and $p'(0)$ follow straightforwardly from Theorem \ref{thm:p(x)}, \eqref{P'(0)} and \eqref{P''(0)}. Because the integrals $I_\pm(0)$ are strictly positive, $p'(0)>0$. 

(iii) Remark \ref{rem:kappa_k} shows that $k_\pm(x)$ is strictly monotone increasing and bounded. Whether $\ka_+$ and $\ka_-$ have the same or different sign depends on the parameters (see Fig.\ \ref{fig:p_perturb}). Because $P'$ is strictly positive on $\Reals$, statement (iii) now follows from \eqref{p(x)}.

(iv) The invariance property \eqref{invariance_p} follows easily from \eqref{p''inhomog} and \eqref{invariance_P} because $\a$ and $\be$ are step functions with step at $x=0$.
\end{proof}

\begin{figure}[t]
	\centering
	\includegraphics[width=1.0\textwidth]{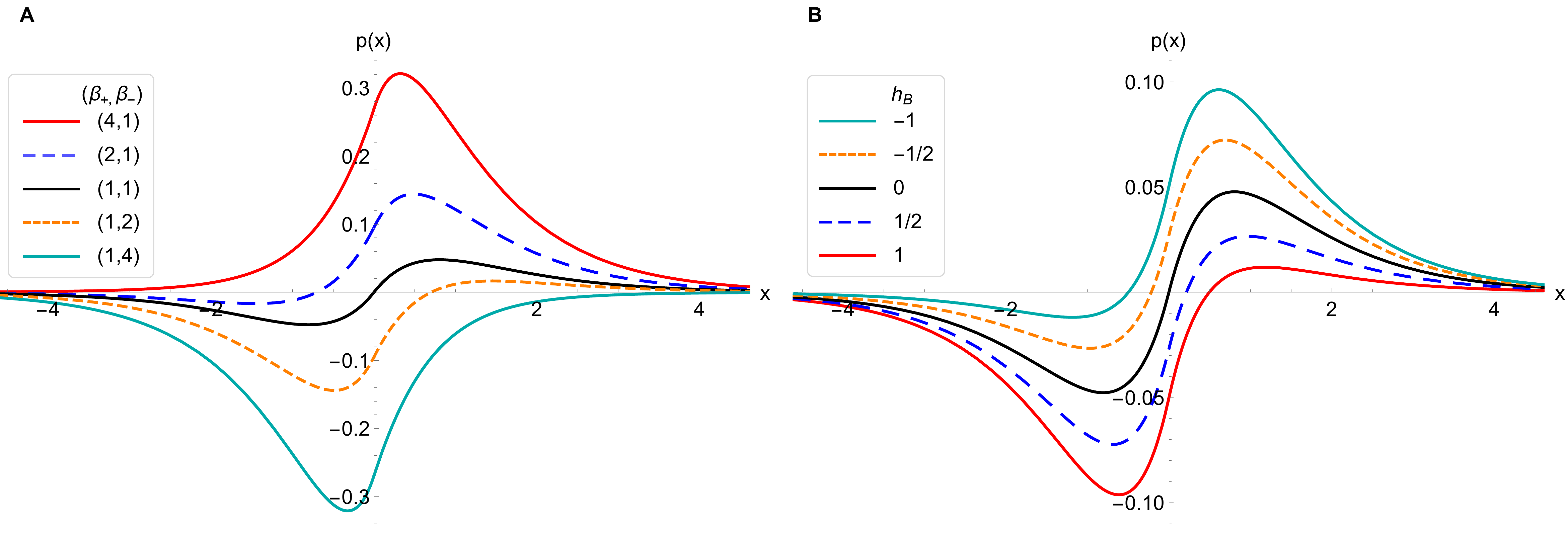}
\caption{The first-order perturbation $p(x)$ as a function of $x$ for different parameter combinations. {\bf A}. The selection intensities at locus $\B$ in the two environments, $(\be_+,\be_-)$, are varied for fixed $\a_+=\a_-=1$. Dominance is absent ($h_A=h_B=0$). {\bf B}. The selection intensities at loci $\A$ and $\B$ are fixed ($\a_+=\a_-=\be_+=\be_-=1$), and the strength and sign of dominance $\dom_B$ at locus $\B$ are varied. Dominance is absent at locus $\A$ ($h_A=0$). In all cases, $\la=1$. The perturbation $p(x)$ is obtained from \eqref{p(x)}, \eqref{define_k_pm}, and \eqref{k_0} by numerical integration using Lemma \ref{p(x)_num_int}. We note that in panel A, $p(x)$ does not change sign if $(\be_+,\be_-)=(4,1)$ or $(\be_+,\be_-)=(1,4)$. In the first case, we have $\ka_-\approx0.09$, $\ka_+\approx1.32$; in the second case, $\ka_-\approx-1.32$, $\ka_+\approx-0.09$; cf.\ Corollary \ref{Cor:prop_p(0)_gen}(iii).}\label{fig:p_perturb}
\end{figure}

In Section \ref{sec:one_locus}, we have shown that the slope of the cline in its center, $P'(0)$, is independent of dominance. This is no longer so in the two-locus case, because $p'(0)$ in \eqref{p'(0)_gen} depends on $I_+(0)$ and $ I_-(0)$, and both integrals depend on $\dom_A$ and $\dom_B$; see also \eqref{p'(0)_equivalent_dom}.

Because $P$ and $p$ satisfy the invariance properties \eqref{invariance_P} and \eqref{invariance_p}, respectively, it follows immediately from \eqref{approx_eps2} that $p_A^{(\ep)}=p_A^{(1/\rh)}$ satisfies the invariance property 
\begin{equation}
	p_A^{(\frac{1}{c_2\rh})}(x,c_1\la) = p_A^{(\frac{c_1}{c_2\rh})}(\sqrt{c_1}x,\la) = p_A^{(\frac{1}{\rh})}(\sqrt{c_2}x,\frac{c_1}{c_2}\la)\,.
\end{equation}
This parallels \eqref{invariance}, as is to be expected.

For equivalent loci, an explicit representation of $p(x)$ will be derived in Section \ref{sec:explicit}. Figure \ref{fig:p_perturb} displays $p(x)$ for various parameter choices. All cases clearly show that $p'(0)>0$ \eqref{p'(0)_gen}. However, also each of the three possibilities in statements (ii) and (iii) of Corollary \ref{Cor:prop_p(0)_gen} are exemplified. The accuracy of the approximation for the cline $(p_A,p_B,D)$ derived in Theorem \ref{thm:p(x)} will be investigated numerically in Section \ref{sec:numerics}.

\subsection{Asymptotic properties}\label{sec:asymptotics}
Here, we derive asymptotic properties of $P(x)$, $p(x)$, and $p_A(x)$ for $x\to\pm\infty$. We recall that two functions $f(x)$ and $g(x)$ are asymptotically equal as $x\to\infty$, i.e., $f(x)\sim g(x)$ as $x\to\infty$, if $\lim_{x\to\infty} f(x)/g(x)=1$, and analogously if $x\to-\infty$.

We start with the single-locus case and assume $-1<\dom_A<1$ and $-1<\dom_B<1$. From \eqref{P(x)_dom}, we obtain immediately the asymptotic equalities
\begin{subequations}\label{P(x)_asymp}
\begin{alignat}{2}
		1-P(x)&\sim 6(1-\dom_A)A_+^{-1}e^{-\xix_+} &&\quad\text{as } x\to+\infty\,, \\
		P(x) &\sim 6(1+\dom_A)A_- e^{\xix_-} 	&&\quad\text{as } x\to-\infty\,,
\end{alignat}
\end{subequations}
and, after differentiation of \eqref{P(x)_dom},
\begin{equation}\label{P'(x)_asymp}
	P'(x) \sim 6a_\pm (1\mp\dom_A)^{3/2} A_\pm^{\mp1} e^{\mp\xix_\pm} \quad\text{as } x\to\pm\infty\,.
\end{equation}
As a trivial consequence, we note that
\begin{subequations}\label{P(x)_asymp_P'(x)}
\begin{alignat}{2}
		1-P(x)&\sim \frac{1}{a_+\sqrt{1-\dom_A}}P'(x) &&\quad\text{as } x\to+\infty\,, \label{P(x)_asymp_P'(x)_a}\\
		P(x) &\sim \frac{1}{a_-\sqrt{1+\dom_A}}P'(x) 	  &&\quad\text{as } x\to-\infty\,.  \label{P(x)_asymp_P'(x)_b}
\end{alignat}
\end{subequations}

In analogy to \eqref{abbrevs_Adom}, we introduce the abbreviations
\begin{subequations}\label{abbrevs_B}
\begin{gather}
	B_+ = F_+(b_0,\dom_B) \,,\quad B_- = F_-(b_0,\dom_B) \,, \label{B_pm_dom} \\
	\yix_+ = x b_+\sqrt{1-\dom_B} \quad\text{ if } \dom_B<1\,, \label{yix}\\
	\yix_- = x b_-\sqrt{1+\dom_B}  \quad\text{ if } \dom_B>-1\,, \label{yix-}
\end{gather}
\end{subequations}
where
\begin{equation}\label{b0}
	b_0=Q(0)
\end{equation}
is the unique solution in $(0,1)$ of the quartic equation
\begin{equation}\label{get_b0_dom}
	b_0^2\phi(b_0,\dom_B) = \frac{\be_+}{\be_+-\be_-}\,;
\end{equation}
cf.\ \eqref{get_a0_dom} and \eqref{phi}.

The main result of this section is the following.

\begin{proposition}\label{prop:asymptotic}
Let $-1<\dom_A<1$ and $-1<\dom_B<1$. 
The asymptotic rates of convergence of the marginal allele-frequency cline $p_A$ are proportional to those of the single-locus cline $P$, i.e., 
\begin{subequations}\label{pA_asymp}
\begin{align}
	\lim_{x\to\infty} \frac{1-p_A(x)}{1-P(x)} &= 1 - \ep a_+\ka_+\sqrt{1-\dom_A} \,,\label{pA_asymp_p} \\
	\lim_{x\to-\infty} \frac{p_A(x)}{P(x)} &=  1 + \ep a_-\ka_- \sqrt{1+\dom_A}	\,. \label{pA_asymp_m}
\end{align}
\end{subequations}
In fact, the following stronger result holds:
\begin{equation}\label{p(x)_asymp}
	\lim_{x\to\pm\infty} \frac{p(x) -\ka_\pm P'(x)}{P'(x) Q'(x)} =  \frac{\mp 2b_\pm\sqrt{1\mp\dom_B}}{2a_\pm\sqrt{1\mp\dom_A}+b_\pm\sqrt{1\mp\dom_B}} \,.
\end{equation}
Asymptotic equalities in terms of exponential functions can be obtained straightforwardly by applying \eqref{P'(x)_asymp}. 
\end{proposition}

\begin{proof}
The proof of \eqref{pA_asymp} is very simple. From \eqref{approx_eps} and \eqref{p(x)}, and after invoking \eqref{kpm_limits} and \eqref{P(x)_asymp_P'(x)_a},
we obtain
\begin{equation*}
	\lim_{x\to\infty}\frac{1-p_A(x)}{1-P(x)} = 1 - \ep\lim_{x\to\infty}\frac{P'(x)k(x)}{1-P(x)} = 1 - \ep\ka_+a_+\sqrt{1-\dom_A}\,,
\end{equation*}
and analogously \eqref{pA_asymp_m}.

To prove \eqref{p(x)_asymp}, we have to work harder. Mimicking \eqref{P'(x)_asymp}, we introduce
\begin{subequations}\label{P'Q'_asymp}
\begin{alignat}{2}
		\tilde P'(x) &= 6a_\pm (1\mp\dom_A)^{3/2} A_\pm^{\mp1} e^{\mp\xix_\pm}   &&\quad\text{if } x\gtrless0 \,,\\
		\tilde Q'(x) &= 6b_\pm (1\mp\dom_B)^{3/2} B_\pm^{\mp1} e^{\mp\yix_\pm}   &&\quad\text{if } x\gtrless0 \,.
\end{alignat}
\end{subequations}
Furthermore, we define
\begin{subequations}
\begin{align}
		\tilde I_+(x) &= \int_x^\infty \left[\tilde P'_+(y)\right]^2\tilde Q'_+(y)[1+\dom_B-2\dom_B \tilde Q_+(y)]\, dy \,, \\
		\tilde I_-(x) &= \int_{-\infty}^x \left[\tilde P'_-(y)\right]^2\tilde Q'_-(y)[1+\dom_B-2\dom_B \tilde Q_-(y)]\, dy \,.
\end{align}
\end{subequations}
Then straightforward integration yields
\begin{equation}\label{Ipm_asymp}
	\tilde I_\pm(x) = \frac{6^3 a_\pm^2b_\pm (1\mp\dom_A)^3(1\mp\dom_B)^{5/2}}{2a_\pm\sqrt{1\mp\dom_A}+b_\pm\sqrt{1\mp\dom_B}}\,A_\pm^{\mp2}B_\pm^{\mp1} e^{\mp(2\xix_\pm + \yix_\pm)}  \quad\text{as } x\to\pm\infty\,.
\end{equation}

Because
\begin{equation}
	I_\pm(x) \sim \tilde I_\pm(x) \quad\text{as } x\to\pm\infty\,,
\end{equation}
we obtain from \eqref{Ipm_asymp} and \eqref{P'Q'_asymp}
\begin{equation}
	I_\pm(x)  \sim \frac{1\mp\dom_B}{2a_\pm\sqrt{1\mp\dom_A}+b_\pm\sqrt{1\mp\dom_B}} [P'(x)]^2 Q'(x) \quad\text{as } x\to\pm\infty\,.
\end{equation}
Therefore,
\begin{equation}\label{Ipm/P'_exists}
	\frac{I_\pm(x)}{[P'(x)]^2} \sim \frac{6b_\pm(1\mp\dom_B)^{5/2}}{2a_\pm\sqrt{1\mp\dom_A}+b_\pm\sqrt{1\mp\dom_B}}B_\pm^{\mp1} e^{\mp \yix_\pm} \quad\text{as } x\to\pm\infty\,
\end{equation}
and
\begin{subequations}\label{**}
\begin{alignat}{2}
	\int_x^\infty \frac{I_+(y)}{[P'(y)]^2}\,dy &\sim \frac{6(1-\dom_B)^2}{2a_+\sqrt{1-\dom_A}+b_+\sqrt{1-\dom_B}}\,B_+^{-1} e^{-\yix_+}   &&\quad\text{as } x\to+\infty\,,\\
	\int_{-\infty}^x \frac{I_-(y)}{[P'(y)]^2}\,dy &\sim \frac{6(1+\dom_B)^2}{2a_-\sqrt{1+\dom_A}+b_-\sqrt{1+\dom_B}}\,B_- e^{\yix_-}  &&\quad\text{as } x\to-\infty\,.
\end{alignat}
\end{subequations}
Substitution into \eqref{k_pm_kappa} yields
\begin{equation}
	k_\pm(x)-\ka_\pm \sim  \frac{\mp 12b_\pm^2(1\mp\dom_B)^2}{2a_\pm\sqrt{1\mp\dom_A}+b_\pm\sqrt{1\mp\dom_B}}\, B_\pm^{\mp1} e^{\mp \yix_\pm} \quad\text{as } x\to\pm\infty\,,
\end{equation}
which can be written as
\begin{equation}
	k_\pm(x) -\ka_\pm \sim  \frac{\mp  2b_\pm\sqrt{1\mp\dom_B}}{2a_\pm\sqrt{1\mp\dom_A}+b_\pm\sqrt{1\mp\dom_B}}\, Q'(x) \quad\text{as } x\to\pm\infty\,.
\end{equation}
Now \eqref{p(x)} establishes \eqref{p(x)_asymp}. 
\end{proof}

The equations \eqref{pA_asymp} may leave the impression that these limits are independent of $h_B$. This is not the case, because $\ka_\pm$ depends both on $h_A$ and $h_B$, and of course on $\a_\pm$ and $\be_\pm$. If $\ka_-<0<\ka_+$, as is the case for loci of equal effects (see eqs.\ \eqref{kappapm_exp_nodom} and Appendix \ref{sec:derive_k_pm_dom}), then \eqref{pA_asymp} shows that the cline $p_A$ in the two-locus system converges faster to its limits 1 or 0 as $x\to\infty$ or $x\to-\infty$ than the single-locus cline $P$. If $\ka_-$ and $\ka_+$ have the same sign (for examples, see Fig.\ \ref{fig:p_perturb}), then convergence may be faster at one end but slower at the other. 

Equations \eqref{kappapm_exp_nodom} and \eqref{kappa_dom_explicit} provide explicit expressions for $\ka_\pm$ for equivalent loci without and with dominance, respectively.

In many of the above asymptotic formulas, the compound constants $A_\pm^{\mp1}$ (and the analogous $B_\pm^{\mp1}$) occur. 
From the definitions of $A_\pm$ in \eqref{A_pm_dom} and the monotonicity of $F_\pm$ in \eqref{F_pm}, we observe that $A_+^{-1}=A_+^{-1}(a_0)$ is monotone decreasing in $a_0$, and $A_-=A_-(a_0)$ is monotone increasing. They satisfy $A_+^{-1}(0)=(2+\sqrt3\sqrt{1-h_A})^{-1}$, $A_+^{-1}(1)=0$, $A_-(0)=0$, and $A_-(1)=(2+\sqrt3\sqrt{1+h_A})^{-1}$.

\begin{remark}\rm
The cases where one or both of $\dom_A$ and $\dom_B$ equal $\pm1$ can be treated easily. The results, however, differ drastically from intermediate dominance. We present the case $\dom_A=\dom_B=1$. For the single-locus clines, we obtain from \eqref{P(x+)_dom}
\begin{equation}
	1-P(x)  \sim \frac{3}{a_+^2}\,x^{-2} \quad\text{as } x\to\pm\infty\,,
\end{equation}
and 
\begin{equation}
	P'(x) \sim \frac{6}{a_+^2}\,x^{-3}   \quad\text{as } x\to\pm\infty\,.
\end{equation}
Simple calculations yield
\begin{equation}
	I_+(x) \sim \frac{648}{5a_+^4b_+^4}\,x^{-10} \quad\text{as } x\to\pm\infty\,,
\end{equation}
and
\begin{equation}
	\int_x^\infty \frac{I_+(y)}{[P'(y)]^2}\,dy \sim \frac{6}{5b_+^2}\,x^{-3}  \quad\text{as } x\to\pm\infty\,.
\end{equation}
Therefore, we find
\begin{equation}
	k_+(x) -\ka_+ \sim - \frac{12}{5b_+^2}\,x^{-3} \quad\text{as } x\to\pm\infty\,.
\end{equation}
Thus, \eqref{pA_asymp_p} and \eqref{p(x)_asymp} are replaced by the qualitatively different 
\begin{equation}
	\frac{1-p_A(x)}{1-P(x)} -1 \sim - \frac{2\ep a_+\ka_+}{x} \quad\text{as } x\to+\infty\,,
\end{equation}
and
\begin{equation}
	p(x) -\ka_+P'(x) \sim - \frac{2}{5} P'(x)Q'(x) \quad\text{as } x\to+\infty\,,
\end{equation}
respectively.
\end{remark}

\subsection{Explicit solution for equivalent loci}\label{sec:explicit}
We assume $b_\pm=a_\pm$ and $-1<h_A=h_B<1$, i.e., the loci have equal effects. Then we can calculate the term $p(x)$ in $p_A(x)=P(x)+\ep p(x)+O(\ep^2)$ explicitly. Since the expressions with dominance are very complicated, we relegate them to Appendix \ref{app:explicit}, which also contains the derivations. After presenting the explicit results in the absence of dominance, we describe and illustrate the influence of dominance on $p'(0)$.

\begin{theorem}\label{thm_explicit}
Let $b_\pm=a_\pm$ and $h=h_A=h_B=0$. Then $p(x)$ is given by \eqref{p(x)}, where
\begin{equation}
	P'(x) = \begin{cases} \dfrac{6a_+(Z_+-Z_+^{-1})}{(Z_++2+Z_+^{-1})^2} &\text{if } x\ge0\,, \vspace{.3cm}\\
			\dfrac{6a_-(Z_-^{-1}-Z_-)}{(Z_+-2+Z_-^{-1})^2} &\text{if } x\ge0\,, \end{cases}
\end{equation}
\begin{equation}
	Z_\pm = A_\pm e^{x a_\pm}\,,
\end{equation}
\begin{equation}
	A_+ = \frac{2+a_0 + \sqrt{3}\sqrt{1+2a_0}}{1-a_0}\,, \quad A_- = \frac{3-a_0 - \sqrt{3}\sqrt{3-2a_0}}{1-a_0}\,,
\end{equation}
\begin{subequations}\label{k_pm_explicit}
\begin{align}
	k_+(x) &= k_0 + 2a_+ \Bigl[\frac{Z_+^{-1}-2}{Z_+ - Z_+^{-1}} - \frac{A_+^{-1}-2}{A_+-A_+^{-1}}  \Bigr] \quad\text{if } x\ge0\,, \label{kp(x)_exp} \\
	k_-(x) &=  k_0 - 2a_- \Bigl[\frac{A_-^{-1}-2}{A_--A_-^{-1}} - \frac{Z_-^{-1}-2}{Z_--Z_-^{-1}} \Bigr]   \quad\text{if } x<0\,, \label{km(x)_exp}
\end{align}
\end{subequations}
and
\begin{subequations}
\begin{align}
	k_0 &= \frac{2a_+(1-a_0)(2a_0-1)}{\sqrt3\sqrt{1+2a_0}(3-2a_0)} 
		= \frac{2a_-a_0(2a_0-1)}{\sqrt3\sqrt{3-2a_0}(1+2a_0)}\,.
\end{align}
\end{subequations}
In addition,
\begin{subequations}\label{kappapm_exp_nodom}
\begin{align}
	\ka_+ = \lim_{x\to\infty}k_+(x)&= a_+\left(1-\frac{2-a_0}{3-2a_0}\sqrt{\frac{1+2a_0}{3}} \right) \quad\text{if } x\ge0\,, \label{kappap_exp_nodom} \\
	\ka_- = \lim_{x\to-\infty}k_-(x)&= -a_-\left(1-\frac{1+a_0}{1+2a_0}\sqrt{\frac{3-2a_0}{3}} \right) \quad\text{if } x<0\,. \label{kappam_exp_nodom} 
\end{align}
\end{subequations}
\end{theorem}

\begin{remark}\label{properties of kpm}\rm
It is easy to show that $\ka_+$ is strictly monotone decreasing in $a_0\in[0,1]$ and satisfies $1.23a_+\approx 2a_+(1-\frac{2}{3\sqrt3})\ge \ka_+ \ge0$. Similarly, $\ka_-$ is strictly monotone decreasing in $a_0\in[0,1]$ and satisfies 
$0\ge \ka_-\ge -2a_-(1-\frac{2}{3\sqrt3}) \approx-1.23a_-$.
\end{remark}

The following corollary specifies several important properties of $p(x)$.

\begin{corollary}\label{Cor:prop_p(0)_explicit_nodom}
Under the assumptions of Theorem \ref{thm_explicit}, $p(x)$ has the following properties:

{\rm (i)}
\begin{equation}
	p(0) = 8a_+^2\,\frac{(1-a_0)^2(2a_0-1)}{3(3-2a_0)}
	 = \frac83(a_+^2+a_-^2)a_0^2(1-a_0^2)(2a_0-1)\,.
\end{equation}
Therefore, $p(0)\gtrless0$ if and only if $a_0\gtrless\frac12$, which is the case if and only if $a_+\gtrless a_-$.

{\rm (ii)}
\begin{equation}\label{p'(0)_equivalent}
	p'(0) = \frac{8a_+^3(1-a_0)^3\sqrt{3+6a_0}}{3-2a_0} = 8\sqrt{3} a_+a_- \sqrt{a_+^2+a_-^2}\, a_0^2(1-a_0^2)\,,
\end{equation}
which is always positive, in accordance with the general case. 

{\rm (iii)} The frequency $p(x)$ changes sign once. If $p(x_k)=0$, then $p(x)>0$ if $x>x_k$ and $p(x)<0$ if $x<x_k$.
\end{corollary}

\begin{proof}
Statements (i) and (ii) follow immediately from the corresponding statements in Corollary \ref{Cor:prop_p(0)_gen} by applying \eqref{I_+-(0)_domexp}.

(iii) From Remark \ref{properties of kpm}, we obtain $\ka_-<0<\ka_+$. Because $k_+$ and $k_-$ are strictly monotone increasing on their domains (see eq.\ \ref{define_k_pm}), $k(x)=k_\pm(x)$, $x\gtrless0$, changes sign once, say at $x=x_k$. Because $P'>0$ on $\Reals$, $p$ changes sign once, at $x=x_k$, and satisfies $p(x)>0$ if $x>x_k$ and $p(x)<0$ if $x<x_k$.
\end{proof}

Explicit asymptotic expansions of $p(x)$ and $p_A(x)$ as $x\to\pm\infty$ follow immediately from Proposition \ref{prop:asymptotic} by setting $b_\pm=a_\pm$, $h=h_A=h_B=0$, and using \eqref{kappapm_exp_nodom}.

Corollary \ref{Cor:prop_p(0)_explicit_dom} generalizes Corollary \ref{Cor:prop_p(0)_explicit_nodom} to intermediate dominance. In contrast to single-locus clines, in two-locus systems the slope $p_A'(0) \approx P'(0) + \rh^{-1}p'(0)$ of the allele-frequency clines depends on $h$ because $p'(0)=p'(0,h)$ depends on $h$. Figure \ref{fig:p'(0,h)/p'(0,0)} displays the ratio $p'(0,h)/p'(0,0)$ for different values of $\tilde\a$. That this ratio depends only on $\tilde\a=\a_+/(\a_++\a_-)$, and not on $\a_+$ and $\a_-$ separately, is a consequence of the scaling property shown in Lemma \ref{lem:invariance} because $\tilde\a$ remains unchanged if and only if both $\a_+$ and $\a_-$ are multiplied by the same constant. The estimates \eqref{p'(dom)_bounds} show that with intermediate dominance, the slope $p'(0,h)$ is at most 20\% larger or smaller than without dominance. Since our analysis is based on the assumption of large $\rh$, the influence of dominance on $p_A'(0)$ will be small, at least for loci of equal or similar effects.

\begin{figure}[t]
	\centering
	\includegraphics[width=0.6\textwidth]{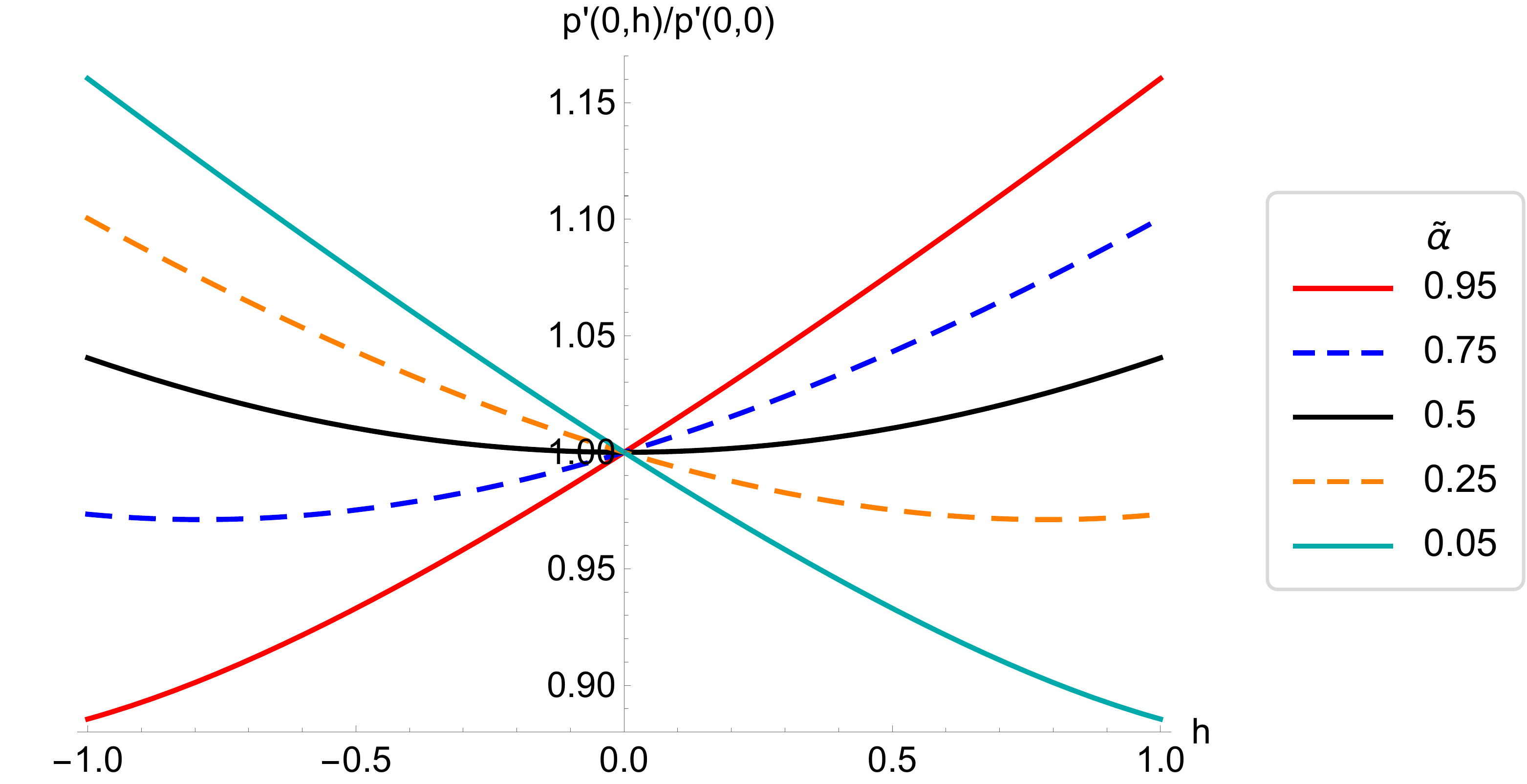}
\caption{The effects of dominance on the slope of two-locus clines for equivalent loci. Displayed is the ratio $p'(0,h)/p'(0,0)$ in \eqref{p'(dom)_bounds_1} as a function of  $h$ for different values of $\tilde\a=\a_+/(\a_++\a_-)$. The dependence of $p'(0)=p'(0,h)$ on $h$ is indicated explicitly.}\label{fig:p'(0,h)/p'(0,0)}
\end{figure}

\subsection{A global measure for the steepness of a cline}\label{sec:global_steepness}
The slope of a cline in its center, $P'(0)$ or $p_A'(0)$, provides a local measure for its steepness, which is the inverse of  its width. Here, we study the global measure
\begin{equation}\label{def:steepness}
	s(\pi) = \int_{-\infty}^\infty \pi'(x)^2 \,dx
\end{equation}
for the steepness of a cline $\pi(x)$ and compare with the slope $\pi'(0)$. The function
$s(\pi)$ is the square of the measure introduced by Liang and Lou (2011), who proved for a more general model on a bounded domain that it is a strictly monotone increasing function of $\la$; see also Lou et al.\ (2013). We mention that $s(\pi)$ is one of the two terms in the Lyapunov function that is minimized by the solution (Fleming 1975).
To shape intuition, we note that $s(\tanh(ax)) = \frac{4a}{3}$.

\subsubsection{A single locus}
We start by collecting properties of the steepness $s(P)$ for the single-locus cline \eqref{P(x)_dom}. We write $P_\la$ to indicate that for given $(\a_+,\a_-)$ each $\la$ defines a different function. Then the invariance property \eqref{invariance_P} implies $P_{c\la}'(x)=cP_\la'(\sqrt{c}x)$. A simple integral transformation yields
\begin{equation}\label{steep_scale_P}
	s(P_{c\la}) = \sqrt{c}\, s(P_\la)\,.
\end{equation}
Therefore, as expected, the cline becomes steeper with stronger selection, i.e., larger $\la$.

The steepness of the single-locus cline $P(x)$ can be computed explicitly. The expression, however, is very complicated and given by equations \eqref{steep_P_dom} in Appendix \ref{app:steepness}. Figure \ref{fig:steep_dom} displays $s(P)$ as a function of $h$ for various parameter combinations, which are chosen such that all resulting clines have the same slope $P'(0)$. Whereas, $P'(0)$ does not depend on the degree of dominance, the steepness $s(P)$ does depend on $h$, although not very strongly. The reason is that the region near the environmental step contributes most to $s(P)$ because it is there that $P'(x)$ is highest and (nearly) independent of $h$.

In the absence of dominance ($h=0$), the steepness simplifies to
\begin{equation}\label{steep_P_nodom}
	s(P) = \frac{3}{5}\Bigl[(a_++a_-) - \frac{\sqrt{a_+^2+a_-^2}}{\sqrt3}\sqrt{(3-2a_0)(1+2a_0)}  \Bigr]\,;
\end{equation}
see Appendix \ref{app:steepness}.
Evidently, \eqref{steep_P_nodom} satisfies \eqref{steep_scale_P}.
If we fix $\la$ and the step size $\a_++\a_-$, or equivalently $a_+^2+a_-^2$, we can consider $s(P)$ as a function of $a_+$, or of $a_0$. 
Then \eqref{steep_P_nodom} shows that $s(P)$ decreases monotonically to 0 as $a_+\to0$ (because then $a_0\to0$) or as $a_+\to\sqrt{a_+^2+a_-^2}$ (because then $a_0\to1$). In addition, $s(P)$ is symmetric about $a_+=a_-$, or $a_0=1/2$, and $s(P)$ is maximized at this point, at which it assumes the value
\begin{equation}\label{steep_P_a0=1/2}
	s(P) = \frac{2}{5}(3-\sqrt6)a_+\,.
\end{equation}

In Appendix \ref{app:steepness}, we provide the simple and accurate approximation \eqref{s(P)_approx_nodom} to \eqref{steep_P_nodom}, and we show that in the absence of dominance the ratio $P'(0)/s(P)$ varies between quite narrow bounds, i.e., 
\begin{equation}\label{P'(0)/s(P)_bounds}
1.564\approx \frac{5}{23} (2 + 3\sqrt3) \le \frac{P'(0)}{s(P)} \le \frac{5}{12}(2+\sqrt6) \approx 1.854\,.
\end{equation}

\begin{figure}[t]
	\centering
	\includegraphics[width=0.6\textwidth]{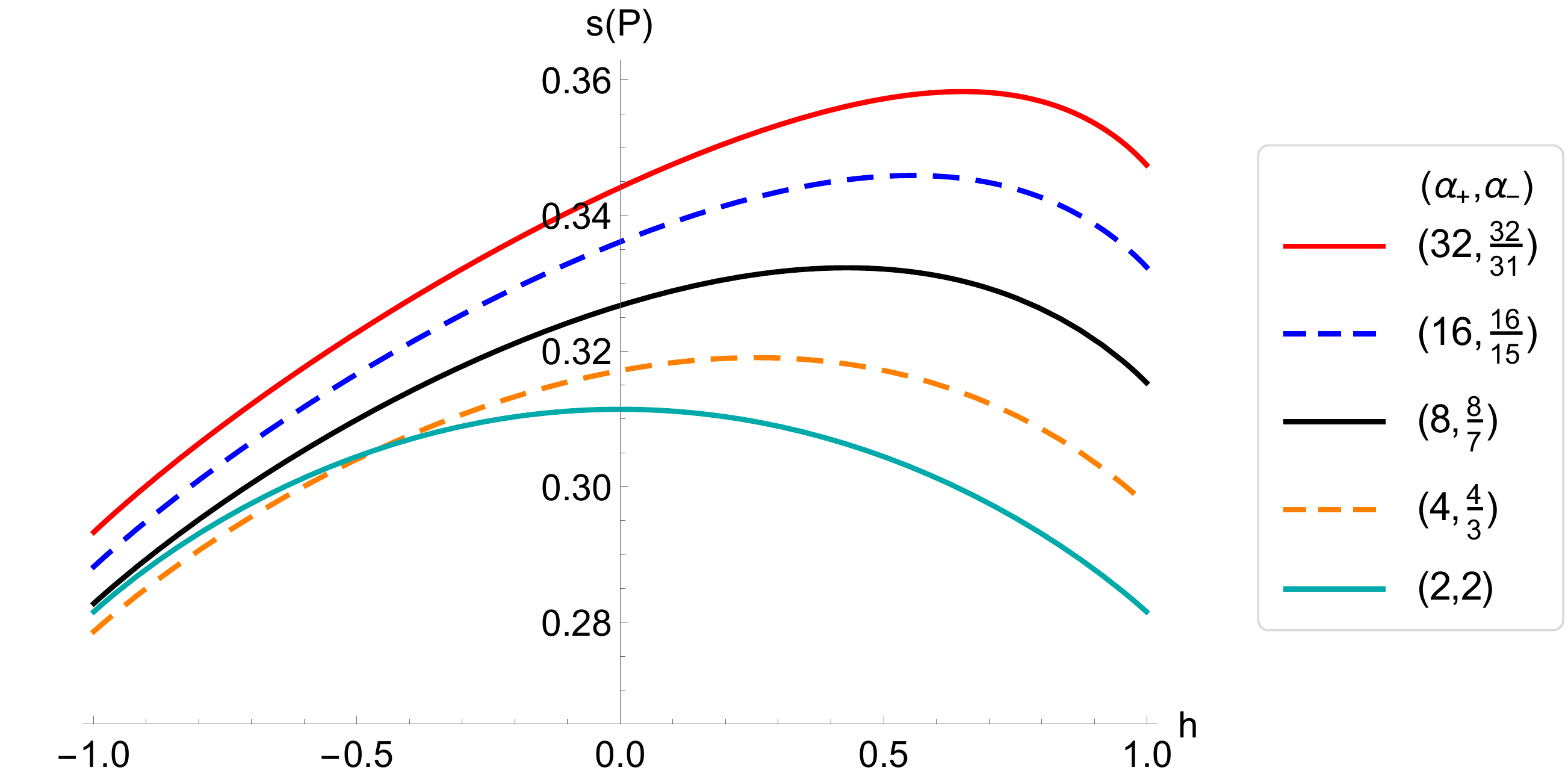}
\caption{Steepness of single-locus clines, $s(P)$, as a function of the degree of dominance, $h$. The curves show $s(P)$ in \eqref{steep_P_dom} as a function of $h$ for the five clines emerging from the pairs of values $(\a_+,\a_-)$ given in the legend, where $\la=1$. These pairs are chosen such that they have the same harmonic mean of $2$. By \eqref{P'(0)}, the corresponding clines have the same slope in the center, i.e., $P'(0)=\sqrt{2/6}\approx 0.577$.}\label{fig:steep_dom}
\end{figure}

\subsubsection{Two recombining loci}
We write $p_{A,(\la,\rh)}$ to indicate the dependence of the cline on $\la$ and $\rh$. 
For two recombining loci the invariance property \eqref{invariance} yields, in a similar way as for a single locus, 
\begin{equation}\label{scale_s_pA}
	s[p_{A,(c_1\la,c_2\rh)}] = \sqrt{c_1}s[p_{A,(\la,\frac{c_2}{c_1}\rh)}] = \sqrt{c_2}s[p_{A,(\frac{c_1}{c_2}\la,\rh)}]\,.
\end{equation}
In particular, this shows that the steepness increases by a factor of $\sqrt{c}$ if both $\la$ and $\rh$ are multiplied by the same constant $c$. 

For two loosely linked loci, we obtain more detailed insight. From \eqref{def:steepness} and \eqref{approx_eps}, we infer
\begin{align}\label{steep_pA}
	s(p_A) &= \int_{-\infty}^\infty P'(x)^2 dx + \frac{2}{\rh} J + O(\rh^{-2}) \notag \\
			&= s(P) + \frac{2}{\rh} J + O(\rh^{-2}) \,,
\end{align}
where
\begin{equation}\label{J_def}
	J = \int_{-\infty}^\infty P'(x)p'(x) dx\,.
\end{equation}
Therefore, $J$ is a measure of the influence of locus $\B$ on the steepness of the cline at locus $\A$. It contributes to $s(p_A)$ in an analogous way as $p'(0)$ to $p_A'(0)$.
Equation \eqref{steep_pA} demonstrates that the allele-frequency cline (at each locus) gets steeper as $\rh$ decreases if and only if $J>0$.  

To study the integral $J$, we write $J(\la)$ to make the dependence on $\la$ explicit and obtain
\begin{align}\label{steep_scale_J}
	J(c\la) &= \int_{-\infty}^\infty P'(x,c\la)p'(x,c\la) dx \notag \\
			&= \int_{-\infty}^\infty \sqrt{c} P'(\sqrt{c} x,\la) c\sqrt{c} p'(\sqrt{c}x,\la) dx \notag \\
			&= \frac{c^2}{\sqrt{c}}\int_{-\infty}^\infty P'(y,\la) p'(y,\la) dy \notag\\
			&= c^{3/2} J(\la)\,.
\end{align}
From \eqref{steep_pA}, \eqref{steep_scale_P} and \eqref{steep_scale_J}, and again \eqref{steep_pA},
we derive
\begin{align}\label{scale_s_pA_rh}
	s(p_{A,(c\la,\rh)}) &= s(P_{c\la}) + \frac{2}{\rh}J(c\la) + O(\rh^{-2})\notag \\
			&= \sqrt{c}\,s(P_{\la}) + \frac{2c^{3/2}}{\rh}J(\la) + O(\rh^{-2}) \notag \\
			&= \sqrt{c}\,s(p_{A,(\la,\rh)}) + \frac{2\sqrt{c}}{\rh}(c-1)J(\la) + O(\rh^{-2})\,.
\end{align}
This is not only in accordance with \eqref{scale_s_pA} but shows, provided $J(\la)>0$, that stronger selection ($c>1$) causes a stronger increase in the steepness of the cline for smaller values of $\rh$ than for larger values.

We conjecture that $J$ is always positive. This conjecture is based on a proof for loci of equal effects without dominance (Appendix \ref{app:steep_explicit}) and comprehensive numerical evaluation of $J$ (e.g., Figs.\ \ref{fig:J_steep_general} and \ref{fig:plot_J_dom}).

\begin{figure}[t]
	\centering
	\includegraphics[width=0.6\textwidth]{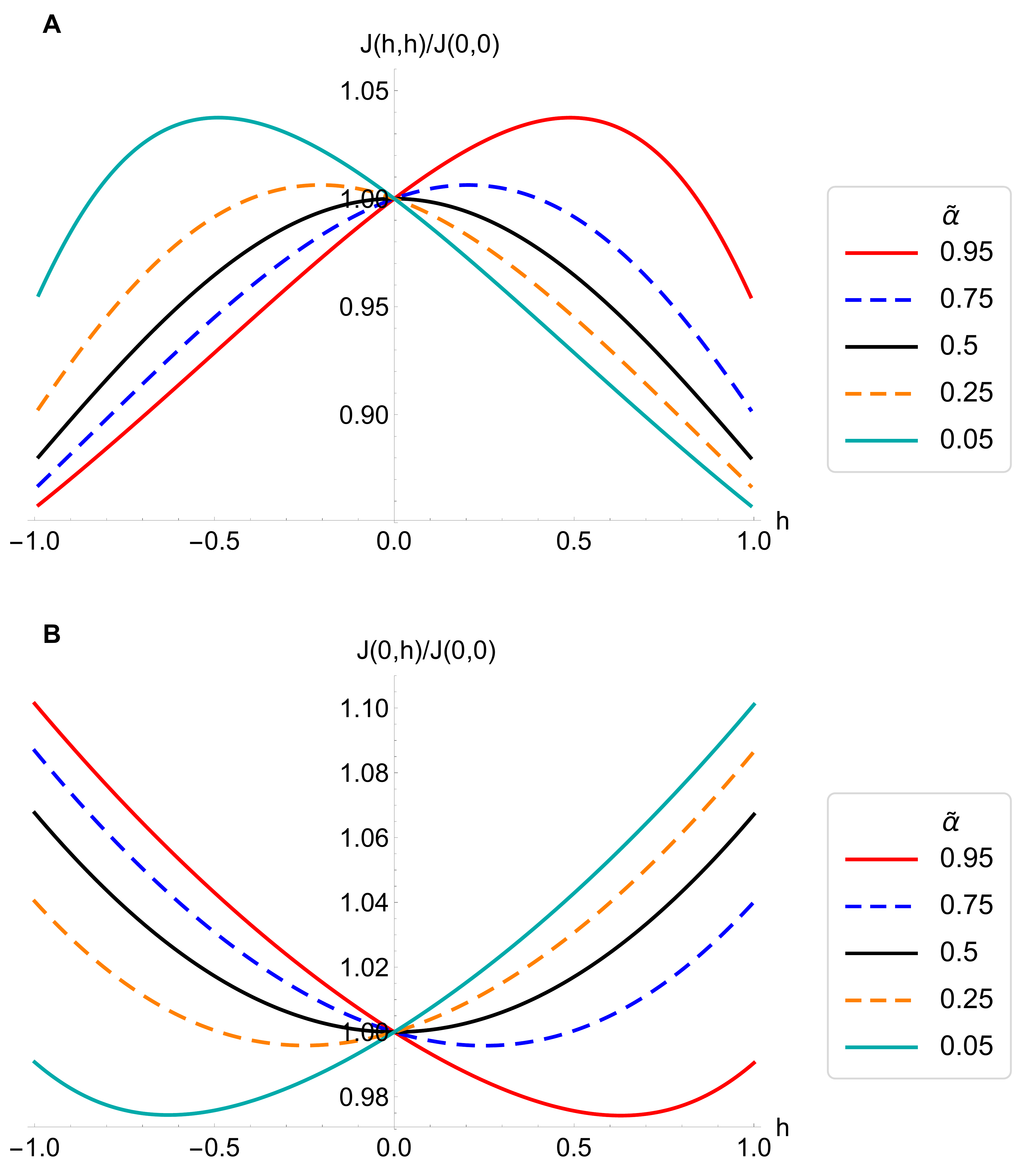}
\caption{The effects of dominance on the steepness of two-locus clines. {\bf A} The ratio $J(h,h)/J(0,0)$ as a function of  $h$ for different values of $\tilde\a=\a_+/(\a_++\a_-)$. {\bf B} The ratio $J(0,h)/J(0,0)$ as a function of $h$ for different values of $\a$. The scaling property \eqref{steep_scale_J} shows that the ratios displayed in the figure do not depend on $\a_+$ and $\a_-$ separately, but only on $\a$. In both panels, both loci have the step sizes $\be_\pm=\a_\pm$, so that selection on homozygotes is equally strong at both loci.}\label{fig:plot_J_dom}
\end{figure}

Figure \ref{fig:plot_J_dom} illustrates the effects of dominance on $J=J(h_A,h_B)$ for loci of equal effects (panel A) and for the case where only locus $\B$ exhibits dominance (panel B). Panel A is analogous to Fig.\ \ref{fig:p'(0,h)/p'(0,0)} and shows that $J$  depends on $h$ in a qualitatively different way than $p'(0)$. If dominance is absent at locus $\A$, dominance at locus $\B$ affects the steepness $s(p_A)$ in \eqref{steep_pA} as depicted by Fig.\ \ref{fig:plot_J_dom}B.  
For loci of equal effects and in the absence of dominance, Fig.\ \ref{fig:steepness_vs_slope} displays the ratio $p_A'(0)/s(p_A)$ as a function of $\a_+$ for given step size $\a_++\a_-$ and three (large) recombination rates. 

We conclude that the global measure $s(p_A)$ for the steepness has properties very similar to the slope $p_A'(0)$ if there is no dominance. In the presence of dominance, the steepness $s(p_A)$ depends on the degree of dominance and therefore reflects its effects much better than the slope. This is most conspicuous if the `primary' locus (here, $\A$) exhibits dominance because in this case the slope in the center is independent of $h_A$. The indirect effects of dominance at a loosely linked locus are weak; this is true for both measures of steepness if both loci are under similarly strong selection.  Of course, if the second locus is under much stronger selection than the first, then dominance at this locus may have a bigger effect on the first locus. 

\section{No recombination}\label{sec:norec}
If recombination is absent ($\rh=0$), the system of partial differential equations \eqref{dynamics_pi} for the gamete frequencies becomes formally equivalent to a one-locus four-allele model in which the gametes play the role of the alleles. With the fitness functions \eqref{step_functions} and the assumption \eqref{sign_ab}, one-locus theory (Section \ref{sec:one_locus}) implies that on each of the edges ($p_i+p_j=1$, $i\ne j$) of the state space there exists a unique cline. 

By our assumptions, $AB$ is the gamete with highest fitness ($\a_++\be_+$) if $x\ge0$ and that with lowest fitness ($-\a_--\be_-$) if $x<0$, whereas $ab$  is the gamete with lowest fitness ($-\a_+-\be_+$) if $x\ge0$ and that with highest fitness ($\a_-+\be_-$) if $x<0$. The gametes $Ab$ and $aB$ have intermediate fitness everywhere. We conjecture that the cline formed by the gametes $AB$ and $ab$ (satisfying $p_2=p_3=0$, $p_1>0$, $p_4>0$) is globally asymptotically stable. The frequency $p_1$ of gamete $AB$ of this cline is obtained from the one-locus formula \eqref{P(x)} by substituting in \eqref{abbrevs_A} $\a_++\be_+$ for $\a_+$, and $\a_-+\be_-$ for $\a_-$. Then the two-locus cline, denoted by $p_{AB}=(p_A,p_B,D)$, satisfies $p_A=p_B=p_1$ and $D=p_1(1-p_1)$.

This conjecture is supported by Corollary 4.7 of Lou and Nagylaki (2004) which implies global stability for an analogous model defined on a \emph{bounded} domain if $\la$ is sufficiently large. Because $\rh=0$, Lemma \ref{lem:invariance}(iii) shows that if our cline $p_{AB}$ is globally asymptotically stable for one value of $\la$, then it is globally asymptotically stable for every $\la>0$. 
Numerical integration of the time-dependent equation \eqref{eq:ABD_add} also supports this conjecture. 

In principle, a regular perturbation analysis of the equilibrium $p_{AB}$ can be performed to obtain an approximation of the two-locus cline for small values of $\rh$. This leads to a system of three coupled linear PDEs that seems to be difficult to analyze. In fact, already for a continent-island model (B\"urger and Akerman 2011) and a two-deme model (Akerman and B\"urger 2014), the resulting approximations were very complicated and not immediately intuitive.

\section{Numerical results for arbitrary recombination}\label{sec:numerics}
Here, we present numerical results to (i) test the accuracy of the approximation for the cline $(p_A,p_B,D)$ derived in Theorem \ref{thm:p(x)}, and to (ii) illustrate the dependence of the properties of the two-locus cline on the recombination rate over the full range of recombination rates 

Stationary solutions were obtained by numerical integration of the system \eqref{eq:ABD_add} of partial differential equations from given initial conditions. We used the \emph{Mathematica} function $\mathsf{NDSolve}$ with sufficiently small values of $\mathsf{MaxStepFrac}$ (mostly $=0.0005$).  We solved this system for $(x,t)\in[-L,L]\times[0,T]$, where for given parameters $\la$ (usually $\la=1)$, $\a_\pm$, and $\be_\pm$, the lengths $T$ and $L$ of the time and spatial intervals were chosen sufficiently large in the following sense. $T$ was chosen such that $\max_{-L\le x\le L}\abs{p_A(x,T)-p_A(x,T/2)}<5\times10^{-7}$ (and analogously for $p_B$). Our choice of $L$ yields
$\max_{-2L/3\le x\le 2L/3}\abs{p_A(x,T,\rh=\infty)-P(x)}<1.2\times10^{-3}$, where $P(x)$ is the analytically calculated cline of the one-locus case (and analogously for $p_B$). We imposed zero-flux boundary conditions; cf.\ \eqref{boundary_cond1}. Deviations of the numerically computed solution on $[-L,L]$ from the analytical solution on $\Reals$ close to the boundaries of this interval cannot be minimized arbitrarily because on a finite interval stationary allele frequencies are always strictly positive at the boundary. Values of $T$ and $L$ are given in the figure legends.

The reader may recall from \eqref{scaled_pars} that we use scaled parameters, i.e., $\rh=r\la=2r/\si^2$ is the recombination \emph{rate} relative to the diffusion constant $d=1/\la=\frac12\si^2$, which is half the migration variance. 
Following \eqref{cline_length}, we use the harmonic mean of $\a_+$ and $\a_-$, $H(\a_+,\a_-)$, as a measure for the selection intensity (before scaling by $d$). Then the ratio of recombination rate to selection intensity at locus $A$ is 
\begin{equation}
	\frac{r}{H(\a_+,\a_-)} = \frac{1}{\ep\la H(\a_+,\a_-)}\,.
\end{equation}
We expect that the strong-recombination approximation derived in Section \ref{sec:strong_rec} is accurate if this ratio is sufficiently bigger than 1.
For the parameters in Figs.\ \ref{fig:p_perturb_accur} and \ref{fig:D_perturb_accur}, we obtain $r/H(\a_+,\a_-) = 25/(16\ep) = 0.5625/\ep$. For those in Fig.\ \ref{fig:pApBD_rho}, we obtain $r/H(\a_+,\a_-) = 9/(8\ep) = 1.125/\ep$ and $r/H(\be_+,\be_-) = 15/(4\ep) = 3.75/\ep$. Indeed, Fig.\ \ref{fig:slope_function_of_rho} shows that the slope at the center of the allele-frequency clines is reasonably well approximated by the weak-recombination approximation if $r/H\gtrsim1$.

\begin{figure}[t]
	\centering
	\includegraphics[width=0.7\textwidth]{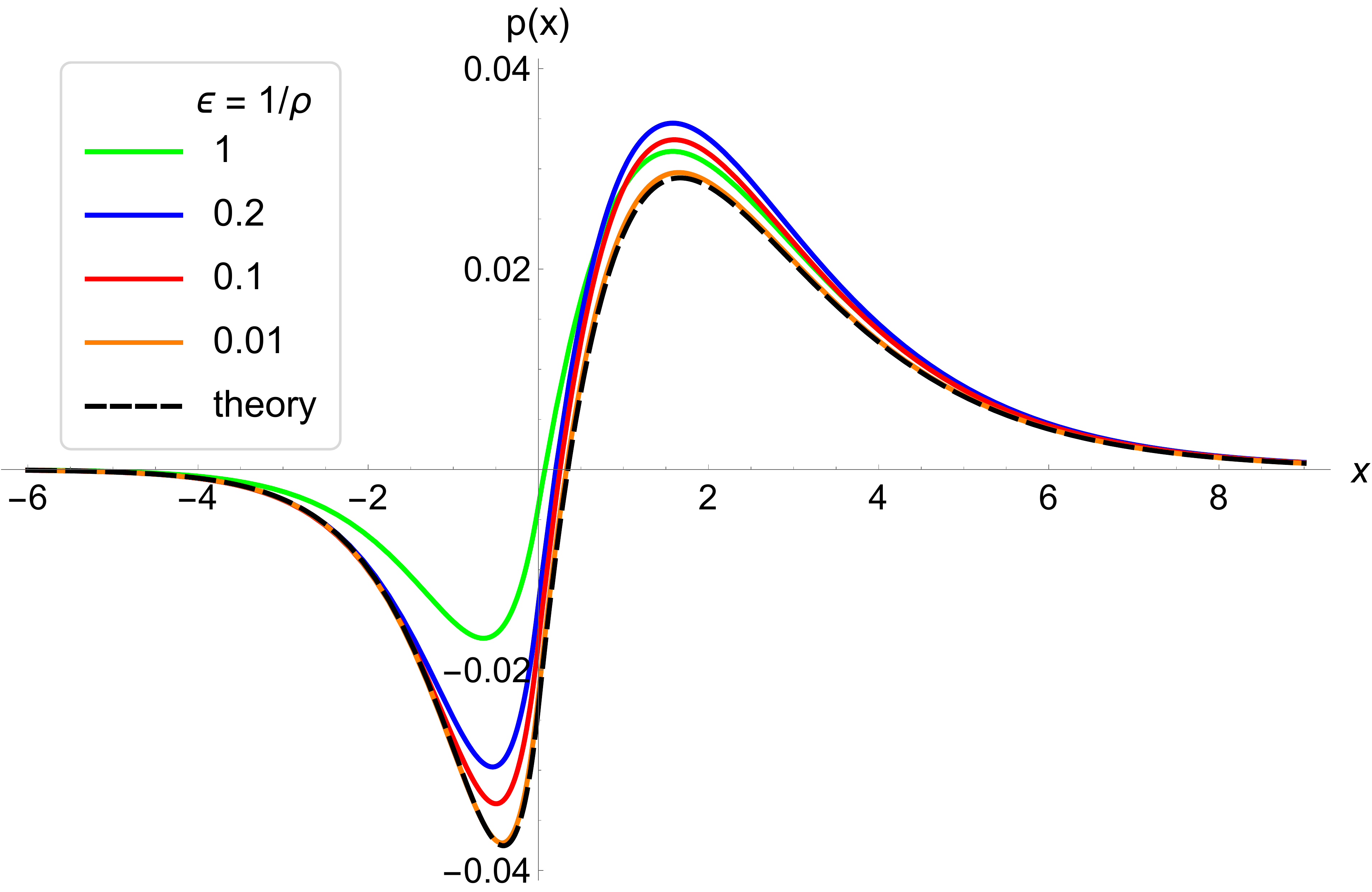}
\caption{Accuracy of the first-order perturbation $p(x)$ in the absence of dominance ($h_A=h_B=0)$. The selection parameters are $\la=1$, $\a_+=\be_+=0.4$, $\a_-=\be_-=1.6$. Because the loci are equivalent, $p(x)$ (dashed black line) was computed from the explicit formulas given in Theorem \ref{thm_explicit}. The solid lines shows $[p_A(x)-P(x)]/\ep$, where $p_A(x)$ and $P(x)$ were obtained by numerical integration as described in the main text. We chose $T=100$ and $L=12$. The ordering of lines in the legend coincides with their order near $x=-1$.}\label{fig:p_perturb_accur}
\end{figure}

\subsection{Accuracy of the strong-recombination approximation}\label{sec:accuracy}
For a range of (scaled) recombination rates $\rh=1/\ep$, Fig.\ \ref{fig:p_perturb_accur} compares the analytical approximation $p(x)$ from Theorem \ref{thm:p(x)} with $\frac{1}{\ep}(p_A(x)-P(x))$, where $p_A(x)$ and $P(x)$ where obtained by numerical integration of the PDE as described above. Here, the loci are equivalent and dominance is absent. Qualitatively similar results were found for a variety of parameter combinations including strong dominance and loci of different effects. If the selection strength is similar to that in Fig.\ \ref{fig:p_perturb_accur}, then the accuracy of the approximation $p(x)$ to $\frac{1}{\ep}(p_A(x)-P(x))$ is also similar to that in the figure (results not shown).

Figure \ref{fig:D_perturb_accur} shows $D(x)$ and its approximation $\ep d(x)$ in \eqref{D_ep}. Whereas $D(x)$ has a continuous first derivative at $x=0$, $d(x)$ has not. Also recall from Sect.\ \ref{sec:model} that $D(x)$ has different one-sided second derivatives at $x=0$. Note that in contrast to its approximation $\ep d(x)$, $D(x)$ is in general not maximized at $x=0$.

\begin{figure}[t]
	\centering
	\includegraphics[width=0.7\textwidth]{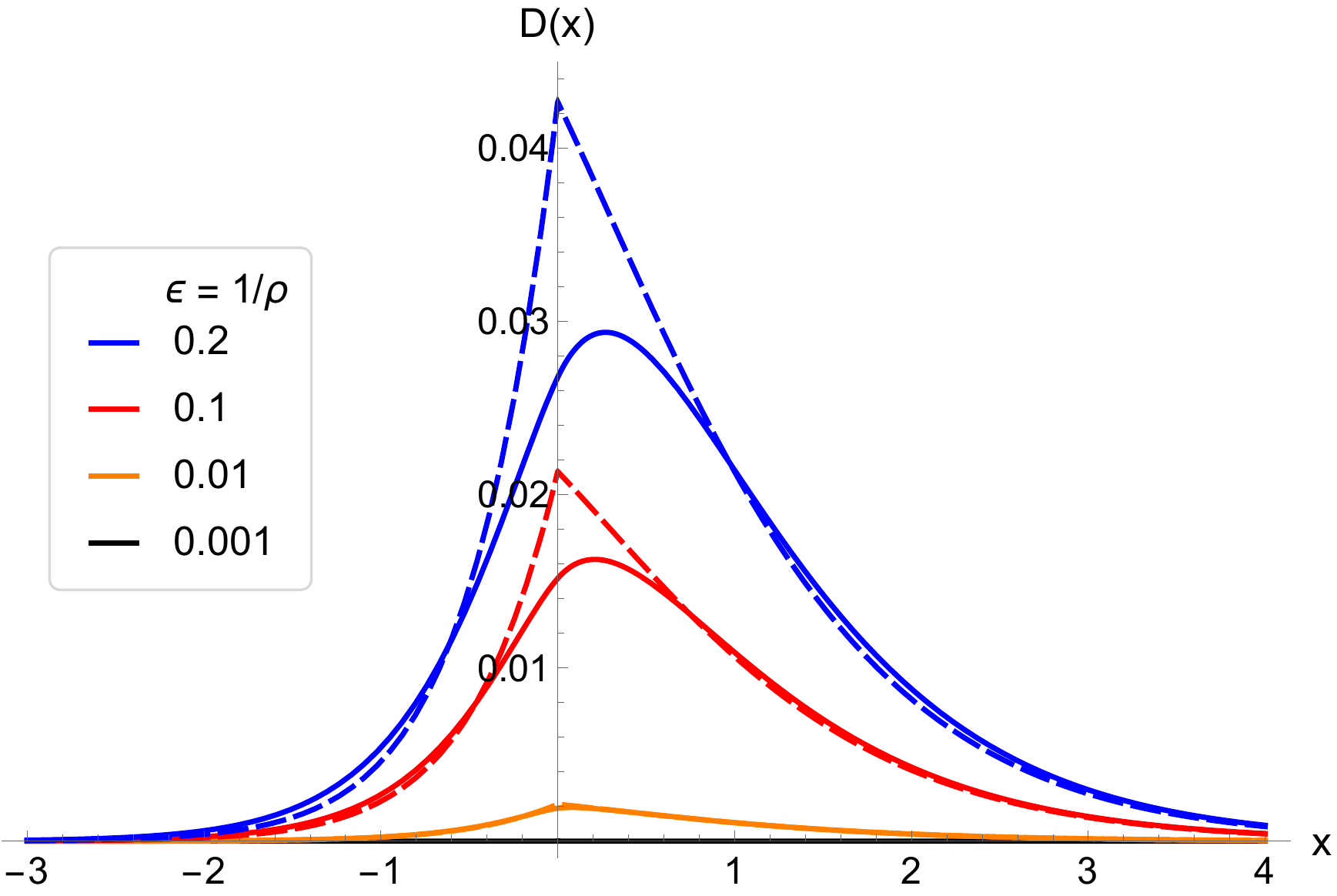}
\caption{Accuracy of the first-order perturbation $\ep d(x)$ in \eqref{D_ep} of $D(x)$. The latter was obtained by numerical solution of \eqref{eq:ABD_add}. Solid lines show the true $D(x)$, dashed lines the corresponding approximation $\ep d(x)$. The selection parameters are as in Fig.\ \ref{fig:p_perturb_accur}. We chose $T=100$ and $L=12$.}\label{fig:D_perturb_accur}
\end{figure}

\subsection{Dependence of the two-locus cline on the recombination rate}\label{sec:reco_rate_general}
To investigate the dependence of the shape of two-locus clines on the recombination rate for given parameters $\a_\pm$, $\be_\pm$, it is sufficient to vary $\rh$ and keep $\la$ constant. This follows from the scaling property \eqref{invariance}. 

Figure \ref{fig:pApBD_rho} illustrates the dependence on the scaled recombination rate $\rh$ of the clines in allele frequencies at each locus and of the linkage disequilibrium between the loci. As already shown by Slatkin (1975), stronger linkage has a stronger effect on the shape, in particular on the slope, of the cline at the locus under weaker selection. This is in accordance with intuition because the locus under stronger selection exerts stronger indirect selection on the weaker locus than vice versa.

\begin{figure}[t]
\begin{center}
\begin{adjustbox}{minipage=\linewidth,scale=0.8}
	\centering
	\includegraphics[height=\textheight]{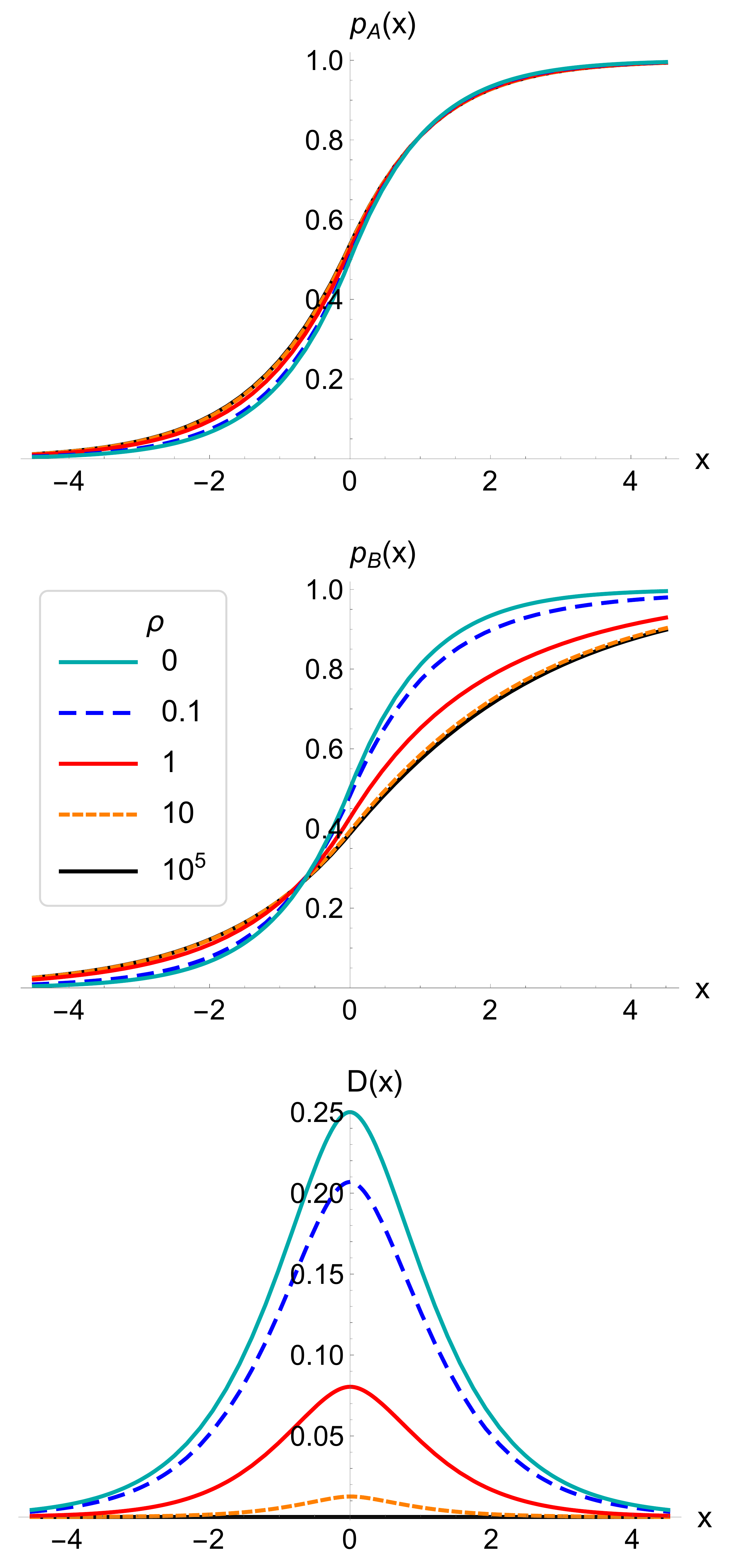}
	\end{adjustbox}
 \end{center}
\caption{Dependence of clines in $p_A$, $p_B$, and $D$ on the scaled recombination rate $\rh$. The other parameters are $\la=1$, $\a_+=2$, $\a_-=1.6$, $\be_+=0.4$, $\be_-=0.8$. For $\rh=0$, the maximum relative deviation on the interval $[-8,8]$ of the numerically calculated $p_A=p_B$ from the corresponding exact one-locus solution is $<2\times10^{-5}$. For $\rh=10^5$, the maximum relative deviations on the interval $[-8,8]$ of $p_A$ and $p_B$ from the corresponding exact one-locus solutions are about $8\times10^{-4}$ and $2.5\times10^{-3}$, respectively.  The ordering of lines in the legend coincides with their order for $x>0$ in the same panel and with the order in the bottom panel. We chose $T=200$ and $L=12$.}\label{fig:pApBD_rho}
\end{figure}
\clearpage

If the loci are unlinked ($\rh\to\infty$, black lines in Fig.\ \ref{fig:pApBD_rho}), the width \eqref{cline_width} of the cline at locus $\A$ ($\B$) is $\om_A=\frac{3\sqrt3}{2}\approx 2.598$ ($\om_B=3\sqrt\frac{5}{2}\approx 4.743$). For completely linked loci ($\rh=0$, green lines in Fig.\ \ref{fig:pApBD_rho}), we obtain $\om_A=\om_B=\sqrt{5}\approx2.236$ because the width can be calculated from the added step-size parameters, i.e., $\a_++\be_+=\a_-+\be_-=1.2$ (Section \ref{sec:norec}). Indeed, most gene-frequency change at loci $\A$ and $\B$ occurs in the intervals $(-\om_A,\om_A)$ and $(-\om_B,\om_B)$, respectively.

Figure \ref{fig:slope_function_of_rho} displays the slope in the center as a function of $\log_{10}(\rh)$ for various parameter combinations. In particular, it compares analytical results from the strong-recombination approximation with results obtained from numerical integration of the PDE \eqref{eq:ABD_add}. The solid lines, based on the analytical approximation, always overestimate the true slope (dots), but provide an accurate approximation if, approximately, $\rh>\sqrt{10}$, i.e., $\log_{10}(\rh)>0.5$. For small recombination rates, the approximation diverges, whereas the true slope approaches that of one-locus clines with values $\a_++\be_+$ and $-(\a_-+\be_-)$ for the step function (dashed horizontal lines on the left).

\begin{figure}[t]
	\centering
	\includegraphics[width=0.7\textwidth]{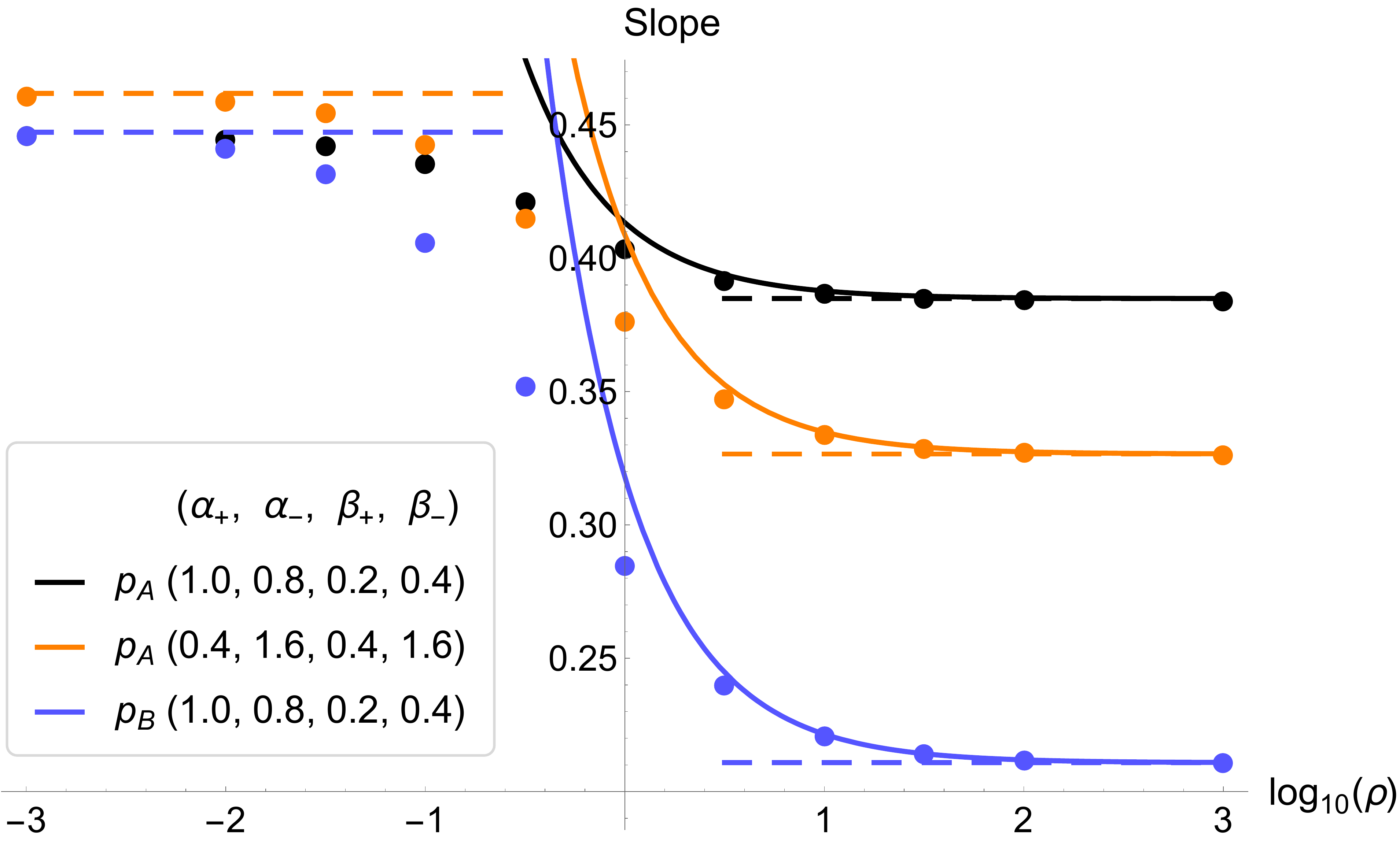}
\caption{Slope of allele-frequency clines in the center as a function of the (decadic) logarithm of the recombination rate $\rh$. Black and orange solid lines show $p_A'(0)=P'(0)+\frac{1}{\rh}\,p'(0)$ for the parameter combinations given in the legend (all with $\la=1$), where $P'(0)$ and $p'(0)$ are calculated from \eqref{P'(0)} and \eqref{p'(0)_equivalent_la}. The blue solid line shows $p_B'(0)$ and was calculated analogously. For the other parameter combination (orange line) loci are equivalent, whence $p_B=p_A$. The ordering of lines in the legend coincides with their order for $\log_{10}(\rh)>0$. The dots are the slopes obtained from numerically computed solutions of the PDE. The dashed horizontal lines on the right are for independent loci ($\rh\to\infty$) and are computed from \eqref{P'(0)}. The dashed horizontal lines on the left are for $\rh=0$ and are computed from \eqref{P'(0)} by using $\a_++\be_+$ and $-(\a_-+\be_-)$ as values for the step function \eqref{step_function_alpha}. The black dashed line on the left side is invisible because it coincides with the blue one since the loci are equivalent. }\label{fig:slope_function_of_rho}
\end{figure}

\section{Discussion}
The shape of a cline in allele, genotype, or phenotype frequencies has many determinants. Among these are the migration distribution, the spatial dependencies of fitnesses, but also dominance relations among alleles and linkage to other loci under selection. In this paper, we focus on the latter two and base our analysis on a simple but classical scenario for migration and selection (Haldane 1948; Slatkin 1973, 1975; Nagylaki 1975). We model migration by homogeneous, isotropic diffusion on the real line and let selection act additively on two loci, $\A$ and $\B$, such that in one region ($x\ge0$) one allele at each locus ($A_1$, $B_1$) is advantageous, whereas in the other region ($x<0)$ the other allele ($A_2$, $B_2$) is advantageous. We admit arbitrary intermediate, spatially independent dominance at both loci, and we assume that the two loci are recombining. Therefore, linkage disequilibrium is generated. We start with a summary of main results.

We studied the non-constant stationary, or equilibrium, solutions of the system \eqref{eq:ABD_add} of partial differential equations for the allele frequencies and the linkage disequilibrium. This system is obtained from the original set of differential equations \eqref{dynamics_pi} for the gamete frequencies (Slatkin 1975) by transforming the variables and rescaling time and parameters appropriately; see \eqref{scaled_pars}. The equilibrium equations resulting from system \eqref{eq:ABD_add} are a special case of a multilocus cline model developed by Barton and Shpak (2000).

Already Haldane showed that the slope of the cline in the center, i.e., at the environmental step, is proportional to the square root of the ratio of the selection coefficient and the migration variance. Therefore, the slope of the cline or its inverse, the width, can be used to infer the selection intensity if an estimate of the migration variance is available. Haldane's results reveal that the asymptotic behavior of the clines for no or complete dominance differ drastically; see also Nagylaki (1975), who concentrated on semi-infinite clines, or an environmental pocket. 

In generalization of Haldane's work, we derived an explicit solution for the single-locus cline with dominance (Theorem \ref{thm:1-locus_cline}). It shows that the slope of the cline in the center, $P'(0)$, is independent of the degree of dominance (eq.\ \eqref{P'(0)} and Figure \ref{fig:plot_dom_1loc}). Since we admit non-symmetric selection, i.e., the selection strength in the two regions may differ, the appropriate measure of selection intensity is the harmonic mean of the step sizes ($\a_+$, $\a_-$; eq.\ \ref{step_function_alpha}). Slatkin's (1973) characteristic length for allele-frequency variation and the width of a cline are generalized accordingly (Remark \ref{rem:width}).

The asymptotic behavior of the single-locus cline follows immediately from the explicit representation in \eqref{P(x)_dom} and is given by eqs.\ \eqref{P(x)_asymp}. The approach to constancy (to 0 as $x\to-\infty$, to 1 as $x\to+\infty$) is exponential, except when dominance is complete. Then the approach is inverse quadratic in the region where the advantageous allele is recessive. In the general case, the exponent of the exponential that describes convergence as $x\to\pm\infty$ is $a_\pm\sqrt{1\mp h}$, where $h$ is the dominance coefficient and
$a_\pm$ is the square root of the ratio of step size to half the migration variance; see eqs.\ \eqref{a_pm} and \eqref{scaled_pars}. Therefore, strong dominance does have a strong effect on the asymptotic properties of the cline. This is also reflected by a global measure for the steepness of a cline, which is investigated in Section \ref{sec:global_steepness}; cf.\ Figure \ref{fig:steep_dom}.  

Sections \ref{sec:strong_rec}, \ref{sec:norec}, and \ref{sec:numerics} are primarily dedicated to the exploration of the influence of linkage on the shape of two-locus clines. Most of the analysis is performed for arbitrary intermediate dominance. Our main result is Theorem \ref{thm:p(x)}. For strong recombination relative to selection and diffusion, it provides the approximations for the clines of (marginal) one-locus allele frequencies and of the linkage disequilibrium. In other words, the quasi-linkage-equilibrium approximation of the two-locus cline is derived. The allele-frequency cline at locus $\A$ is $p_A \approx P+(1/\rh)p$, where $P$ is the single-locus cline at $\A$. The perturbation term $p$ is the solution of the linear second-order differential equation \eqref{p''inhomog} and can be obtained by integration; see \eqref{p(x)}, \eqref{define_k_pm}, \eqref{k_0}. The linkage disequilibrium is approximately $(1/\rh)d(x)$, where $d(x)$ is given by \eqref{D_ep}. 

Corollary \ref{Cor:prop_p(0)_gen} summarizes simple, general properties of the perturbation term $p(x)$. It shows that $p'(0)$ is always positive. Hence, the cline gets steeper in the center with tighter linkage, as is expected intuitively (see Slatkin 1975 and Barton 1983). Depending on the parameters, $p(0)$ can be positive, negative, or zero. In fact, $p(x)$ can be positive on the real line, negative, or change sign once (Figure \ref{fig:p_perturb}). Thus, linkage to other loci may affect the shape of the allele-frequency clines in complex ways. We also note that, in contrast to $P'(0)$, $p'(0)$ does depend on dominance, although only weakly.

If multiple additive loci contribute to fitness, then terms resulting from all pairwise and high-order linkage disequilibria have to be added to the right-hand side of the differential equations for the allele frequencies, i.e., to \eqref{eq:ABD_a} and the corresponding equations for the other loci. If recombination is strong, only pairwise linkage disequilibria will matter and the equilibrium equations for the allele frequencies will have the same form as eqs.\ (24) and (25) in Barton and Shpak (2000). Also differential equations for the linkage disequilibria can be deduced. A quasi-linkage-equilibrium approximation for the pairwise linkage disequilibria was derived by Barton and Shpak (2000). As already noted above, our equation \eqref{D_ep} is a special case of their equation (15). Putting all this together, a multilocus generalization of the second-order differential equation \eqref{p''inhomog} for the perturbation term $p(x)$ in \eqref{approx_eps} can be derived. This has a more complicated inhomogeneous term, resulting from all pairwise linkage disequilibria, and is awaiting analysis. Based on the investigation of Barton (1983) of a model with hybrid inferiority (and spatially uniform selection), we expect that the cumulative effect of many loci may considerably steepen the allele-frequency clines.

In Section \ref{sec:asymptotics}, we derived the asymptotic properties of the perturbation term $p(x)$ and of the quasi-linkage-equilibrium approximation of $p_A$. Proposition \ref{prop:asymptotic} shows that, unless there is complete dominance, $p_A$ approaches constancy (0 or 1) as $x\to\pm\infty$ at the same exponential rate as the single-locus cline $P$. What differs is the multiplicative factor in front of the exponential, which depends on the recombination rate and the other parameters. This factor may increase or decrease with $\rh$ \eqref{pA_asymp}.

For equivalent loci, we calculated an explicit expression for $p(x)$, which is relatively simple if dominance is absent (Theorem \ref{thm_explicit}), but complicated in general (Appendix \ref{app:explicit}). For the slope $p'(0)$ and the value $p(0)$ simple expressions are obtained if there is no dominance (Corollary \ref{Cor:prop_p(0)_explicit_nodom}), and somewhat more complicated expressions if there is dominance (Corollary \ref{Cor:prop_p(0)_explicit_dom}). They all are easily evaluated numerically and provide insight into the effects of dominance on the slope in the center of two-locus clines (Figure \ref{fig:p'(0,h)/p'(0,0)}). These corollaries also show that for equivalent loci, $p(x)$ changes sign once and is positive for large $x$, and negative for small $x$. Therefore, $p_A(x) \approx P(x) + (1/\rh)p(x) > P(x)$ holds above a threshold, and $p_A(x) < P(x)$ below that threshold.

In Section \ref{sec:global_steepness}, we studied the global, cumulative measure of steepness $s(P)$ in \eqref{def:steepness}. It is the square of a measure investigated by Liang and Lou (2011). For a quite general one-locus cline model on a bounded domain, these authors had shown that this measure is a strictly monotone increasing function of $\la$, i.e., stronger selection or weaker migration increase the steepness of the cline. This is also true in our setting, and \eqref{steep_scale_P} quantifies the dependence of $s(P)$ on $\la$. In contrast to the slope in the center, this cumulative measure depends on dominance, although in a complicated way \eqref{steep_P_dom}. This is visualized in Figure \ref{fig:steep_dom}. 

For two recombining loci, we showed that an increase in the strength of selection causes an increase in the steepness $s(p_A)$ of the one-locus cline $p_A$ \eqref{scale_s_pA_rh}. Moreover, this increase is stronger the smaller the recombination rate $\rh$ is. We conjecture, but could not prove, that $s(p_A)$ increases if $\rh$ decreases. This conjecture is supported by numerical evaluations (e.g., Figures \ref{fig:J_steep_general} and \ref{fig:plot_J_dom}) and by a proof for the case of loci of equal effects without dominance (Appendix \ref{app:steep_explicit}). In summary, the dependence of the global measure $s(p_A)$ on the parameters is qualitatively the same as that of the slope in the center, $p_A'(0)$, if there is no dominance. However, in contrast to the slope, which is a local measure, the global measure reflects the influence of dominance. The advantage of the slope in the center is that it will be easier to estimate in practice.

In Section \ref{sec:norec}, we briefly treated the case of no recombination, when the model becomes formally equivalent to a one-locus four-allele system. We conjecture that the cline formed by the gametes $AB$ and $ab$ is globally asymptotically stable, i.e., the intermediate gametes $Ab$ and $aB$ are eventually lost. This cline can be computed from the one-locus formula \eqref{P(x)}. For weak migration, stability of this cline was proved for a finite number of demes (Akerman and B\"urger 2014) and also for a bounded continuous habitat (Su and B\"urger, unpublished). For strong migration, clines usually vanish in bounded discrete or continuous habitats (see below).

Finally, in Section \ref{sec:numerics}, we solved the system \eqref{eq:ABD_add} of PDEs numerically to (i) check the accuracy of our  quasi-linkage-equilibrium approximation (Section \ref{sec:accuracy}, Figures \ref{fig:p_perturb_accur} and \ref{fig:D_perturb_accur}) and to (ii) explore the whole range of possible recombination rates (Section \ref{sec:reco_rate_general}, Figures \ref{fig:pApBD_rho} and \ref{fig:slope_function_of_rho}). The latter results complement and extend findings by Slatkin (1975) for essentially the same model, and by Barton (1983, 1999), Barton and Shpak (2000), and Geroldinger and B\"urger (2015) for related models, who showed numerically that tighter linkage steepens clines (in allele frequency or mean phenotype) and increases linkage disequilibrium. Unfortunately, no general theory is currently available to prove or quantify these findings for arbitrary recombination.

The model investigated in this paper assumes that the habitat is unbounded and, in particular, that each allele is advantageous in an unbounded region (of infinite measure). As a consequence, a cline exists for arbitrarily strong migration relative to selection. However, as the migration variances increases, the cline becomes increasingly flatter. In our model, this follows easily from the scaling, or invariance, properties derived in Lemma \ref{lem:invariance}; for the single-locus case, see also Remark \ref{rem:width}. For very general fitness functions, but still assuming that each allele is advantageous in an unbounded region, Conley (1975) and Fife and Peletier (1977) proved that clines generally exist independently of the strength of migration. As shown by Nagylaki's (1975) analysis of an environmental pocket, an allele that is favored in a bounded region and disadvantageous in unbounded region will be lost under sufficiently strong migration. This is in line with the apparently generic property of migration-selection models on bounded domains that clinal or, more general, genetic variation can be maintained only if selection is sufficiently strong relative to migration. Otherwise, i.e., below a threshold value of the ratio of selection intensity to migration rate, the type with the highest spatially averaged fitness will swamp the whole population. For recent reviews of cline models, we refer to Nagylaki and Lou (2008) and Lou et al.\ (2013), and for reviews of models with discrete demes to Lenormand (2002), Nagylaki and Lou (2008), and B\"urger (2014).

We conjecture that for the present step environment on the real line, the two-locus cline exists, is unique, and globally asymptotically stable for all parameter combinations satisfying our assumptions. Our conjecture is based on the invariance property derived in Lemma \ref{lem:invariance} and the apparent existence and stability for strong recombination and arbitrary $\lambda$ (Theorem \ref{thm:p(x)}).

Lenormand et al.\ (1998) estimated the migration variance and the selection pressure from geographic gradients (clines) in allele frequencies at two loosely linked insecticide resistance loci in the mosquito \emph{Culex pipiens pipiens}. Their simulation results showed that different hypotheses of dominance (recessiveness, no dominance, complete dominance) do not have an important effect on the estimates of migration variance, whereas the recombination rate has a significant effect. Our study provides analytical support for their finding. However, our model is not directly applicable to the case study of Lenormand et al.\ because their organisms are adapted to an environmental pocket. Then the slope of the cline at the transition between the environments and the maximum allele frequencies are needed to estimate the migration variance (selection intensity) if independent estimates of the selection intensity (migration variance) are available (Nagylaki 1975). It would be of interest to extend the present analysis to study clines caused by environmental pocket. In particular, the effects of dominance should be identifiable by studying the asymptotic properties of a cline (see Section \ref{sec:asymptotics}) or the global measure of steepness that was investigated in Section \ref{sec:global_steepness}.

\section*{Acknowledgments}
I am very grateful to Professor Tom Nagylaki and to Dr.\ Linlin Su for perceptive and inspiring communication. It is a pleasure to thank Professor Nick Barton, Professor Tom Nagylaki, and Dr.\ Swati Patel for useful comments on a previous version. Financial support by the Austrian Science Fund (FWF) through Grant P25188-N25 is gratefully acknowledged.

\begin{appendix}
\numberwithin{figure}{section}
\section{Appendix}\label{sec:Appendix}
\subsection{Generality of the fitness scaling \eqref{fitscheme}}\label{app:fitness}
We assume absence of epistasis and assign spatially dependent fitnesses to one-locus genotypes as follows:
\begin{equation}\label{one-loc_fit}
\begin{tabular}{ccc}
	$AA$ & $Aa$ & $aa$ \\
\hline
	$2\a_1(x)$ & $\a_1(x)+\a_2(x) + \dom_A(\a_1(x)-\a_2(x))$ & $2\a_2(x)$ 
\end{tabular}\;,
\end{equation}
and analogously for the other locus. The dominance coefficient $\dom_A$ could also depend on $x$, but in view of our applications, we assume constant dominance coefficients.
Therefore, suppressing the variable $x$, we obtain the genotypic fitnesses 
\begin{equation}\label{fitscheme1}
\begin{tabular}{c|ccc}
	&  $BB$ & $Bb$ & $bb$\\
\hline
	$AA$ & $2\a_1 +2\be_1 $ & $2\a_1 +\be_1 +\be_2 + \dom_B(\be_1-\be_2) $ & $2\a_1 +2\be_2 $\\
	$Aa$ & $\begin{cases}\a_1 +\a_2 + 2\be_1 \\+\dom_A(\a_1-\a_2)\end{cases}$ 
		& $\begin{cases}\a_1 +\a_2 +\be_1 +\be_2 \\+\dom_A(\a_1-\a_2)+\dom_B(\be_1-\be_2) \end{cases}$ 
		& $\begin{cases}\a_1 +\a_2 + 2\be_2 \\ +\dom_A(\a_1-\a_2) \end{cases}$\\
	$aa$ & $2\a_2 +2\be_1 $ & $2\a_2 +\be_1 +\be_2 +\dom_B(\be_1-\be_2) $ & $2\a_2 +2\be_2 $\\
\end{tabular}\;.
\end{equation}

Because in the continuous-time model \eqref{dynamics_pi_a}, the same function of $x$ can be added to all genotypic fitness functions $w_{ij}(x)$ without changing the dynamics, we subtract $\a_1 +\a_2 +\be_1 +\be_2$ from all entries in \eqref{fitscheme1} and obtain
\begin{equation}\label{fitscheme2}
\begin{tabular}{c|ccc}
	&  $BB$ & $Bb$ & $bb$\\
\hline
	$AA$ & $\a_1 -\a_2 +\be_1 -\be_2 $ & $\a_1 -\a_2 + \dom_B(\be_1-\be_2) $ & $\a_1 -\a_2 -\be_1 +\be_2 $\\
	$Aa$ & $\dom_A(\a_1-\a_2)+ \be_1 -\be_2 $ & $\dom_A(\a_1-\a_2)+ \dom_B(\be_1-\be_2)$ & $\dom_A(\a_1-\a_2)-\be_1 +\be_2 $\\
	$aa$ & $-\a_1 +\a_2 +\be_1 -\be_2 $ & $-\a_1 +\a_2 + \dom_B(\be_1-\be_2)$ & $-\a_1 +\a_2 -\be_1 +\be_2 $\\
\end{tabular}\;.
\end{equation}
Defining
\begin{equation}
	\a(x) = \a_1(x)-\a_2(x) \quad\text{and}\quad \be(x) = \be_1(x)-\be_2(x)\,,
\end{equation}
we arrive at \eqref{fitscheme}. 

Assuming that $\a_i(x)$ and $\be_i(x)$ are step functions with a single step at $x=x_0$, and that $A$ and $B$ are advantageous on $[x_0,\infty)$ and disadvantageous on $(-\infty,x_0)$, we obtain the assumptions \eqref{step_functions} and \eqref{sign_ab}. Without loss of generality, we choose $x_0=0$.

\subsection{The integrals in \eqref{define_k_pm}}\label{app:one-fold_int}
We show that the integrals occurring in the definition \eqref{define_k_pm} of $k_\pm(x)$ can be obtained by one-fold integration.
We define
\begin{equation}
	\ps_+(x) = \int_0^x \frac{1}{[P'(y)]^2}\,dy \quad\text{and}\quad \ps_-(x) = -\int_x^0 \frac{1}{[P'(y)]^2}\,dy\,.
\end{equation}

\begin{lemma}\label{p(x)_num_int}
\begin{subequations}\label{int_I_Pp^2}
\begin{alignat}{2}
	\int_0^x \frac{I_+(y)}{[P'(y)]^2}\,dy	&= \ps_+(x)I_+(x) + \int_0^x \ps_+(y)[P'(y)]^2Q'(y)[1+\dom_B-2\dom_B Q(Y)]\,dy	&&\quad\text{if } x\ge0 \,, \\
	\int_x^0 \frac{I_-(y)}{[P'(y)]^2}\,dy	&= -\ps_-(x)I_-(x) - \int_x^0 \ps_-(y)[P'(y)]^2Q'(y)[1+\dom_B-2\dom_B Q(Y)]\,dy	&&\quad\text{if } x\le0\,,
\end{alignat}
\end{subequations}
where
\begin{equation}
	\ps_\pm(x) = \tilde\ps_\pm(x) - \tilde\ps_\pm(0)\,,
\end{equation}
and 
\begin{subequations}\label{tilde_psi}
\begin{align}
	\tilde\ps_+(x) &= \frac{1}{72a_+^3(1-h_A)^{7/2}}\biggl\{12xa_+\sqrt{1-h_A}(5-21h_A+36h_A)^2 \notag\\
	&\quad+ \frac{1}{Z_+-(1+3h_A)Z_+^{-1}}\Bigl[Z_+^3 + 16(1-3h_A)Z_+^2 \notag\\ 
	&\quad -3(1+3h_A)^{-1}(43-126h_A+51h_A^2+288h_A^3+432h_A^4)Z_+ \notag\\
	&\quad- 32 (1-3h_A)(5-3h_A+18h_A^2) - (1+3h_A)^2Z^{-1}  +16(1-3h_A)(1+3h_A)^2 Z^{-2} \notag\\
	&\quad + (1+3h_A)^3 Z^{-3}\Bigr] \biggr\} \qquad\text{ if } x\ge0\,,\; h_A<1\,, \\
	\tilde\ps_-(x) &= \frac{1}{72a_-^3(1+h_A)^{7/2}}\biggl\{12a_-\sqrt{1+h_A}(5+21h_A+36h_A)^2 \notag\\
	&\quad+ \frac{1}{(1-3h_A)Z_--Z_-^{-1}}\Bigl[(1-3h_A)^3 Z_-^3 + 16(1+3h_A)(1-3h_A)^2Z_-^2 \notag\\ 
	&\quad -3(43+126h_A+51h_A^2-288h_A^3+432h_A^4)Z_- \notag\\
	&\quad- 32 (1+3h_A)(5+3h_A+18h_A^2) - (1-3h_A)Z^{-1}  +16(1+3h_A) Z^{-2} \notag\\
	&\quad + Z^{-3}\Bigr] \biggr\}    \qquad\text{ if } x\le0\,, \; h_A>-1  \,.
\end{align}
\end{subequations}
We recall that $Z_\pm$ is defined in \eqref{Z_pm}.
\end{lemma}

\begin{proof}
The equations \eqref{int_I_Pp^2} follow immediately by partial integration because $\ps_\pm(0)=0$ and $I_\pm(0)$ are finite. Equations \eqref{tilde_psi} were derived with the help of \emph{Mathematica}. They can be checked by differentiation.
\end{proof}

Expressions for $h_A=\pm1$ are available in a \emph{Mathematica} notebook.

\subsection{Explicit solution for equivalent loci with dominance}\label{app:explicit}
We assume $b_\pm=a_\pm$ and $-1<\dom=\dom_A=\dom_B<1$. Under these assumptions, we can calculate the first-order term $p(x)$ of $p_A(x)$ explicitly. 
In the following, we provide the most important steps of the derivation. Our main result are equations \eqref{k_pm_explicit_dom}, which present $k_\pm(x)$. Together with \eqref{p(x)}, they provide $p(x)$. Corollary \ref{Cor:prop_p(0)_explicit_dom} gives explicit expressions for $p(0)$ and $p'(0)$. The dependence of $p'(0)$ on the degree of dominance is analyzed thereafter.

\subsubsection{Derivation of $k_\pm(x)$}\label{sec:derive_k_pm_dom}
First, we compute $I_+(x)$ in \eqref{IpIm_gen}. Therefore we need to evaluate integrals of the form $\int_x^\infty [f'(y)]^3[1+\dom-2\dom f(y)]\,dy$, where 
\begin{equation}\label{f''}
	f''(x)=a_1f(x)+a_2f(x)^2+a_3f(x)^3.
\end{equation}
With $f(x) = P(x)$, \eqref{1-locus ODE} yields
\begin{equation}\label{a_123}
	a_1= -a_+^2(1+\dom)\,, \; a_2 = a_+^2(1+3\dom)\,, \; a_3 = -2a_+^2\dom\,.
\end{equation}
We start with the following simple observations:
\begin{subequations}
\begin{align}
	\left[(f')^2\right]' &= 2f'(a_1f+a_2f^2+a_3f^3)\,, \\
	\left[f(f')^2\right]' &= 2f'(a_1f^2+a_2f^3+a_3f^4)+(f')^3\,, \\
	\int_x^\infty f(y)^nf'(y)\,dy &= \frac{1}{n+1}(1-f(x)^{n+1})\,.	
\end{align}
\end{subequations}
The first two identities follow from \eqref{f''}, and the last from one-fold partial integration and the boundary condition $f(\infty)=1$.

Therefore, we obtain by one-fold partial integration and the boundary condition $f'(\infty)=0$:
\begin{subequations}
\begin{align}
	\int_x^\infty [f'(y)]^3\,dy &= -f(x)[f'(x)]^2- 2\left(\frac{a_1}{3}+\frac{a_2}{4}+\frac{a_3}{5}\right) \notag\\
			&\quad+ 2f(x)^3\left(\frac{a_1}{3}+\frac{a_2f(x)}{4}+\frac{a_3f(x)}{5}\right) \,, \\
	\int_x^\infty [f'(y)]^3f(y)\,dy &= -f(x)^2[f'(x)]^2- 2\left(\frac{a_1}{4}+\frac{a_2}{5}+\frac{a_3}{6}\right) \notag\\
			&\quad+ 2f(x)^3\left(\frac{a_1}{4}+\frac{a_2f(x)}{5}+\frac{a_3f(x)}{6}\right) - \int_x^\infty [f'(y)]^3f(y)\,dy\,.
\end{align}
\end{subequations}
By collecting terms and substituting \eqref{a_123}, we arrive at 
\begin{align}\label{int_f_genp}
	&\int_x^\infty [f'(y)]^3[1+\dom-2\dom f(y)]\,dy =\frac{a_+^2}{30}(1-f)^2 \bigl[5(1+2f+3f^2) \notag \\
	&\quad\qquad +\dom(1+2f+3f^2-36f^3) - 20\dom^2(1-f)f^3\bigr] - f(1+\dom-\dom f)(f')^2\,,
\end{align}
where on the right-hand side we abbreviated $f(x)$ by $f$. With $f(x)=P(x)$ and \eqref{P(x)_dom}, \eqref{int_f_genp} yields
\begin{subequations}\label{I_+-(x)_domexp}
\begin{align}
	I_+(x) &= \frac{72a_+^2(1-h)^5}{5\bigl[Z_++2(1-3h)+(1+3h)Z_+^{-1}\bigr]^6} \biggl[ 5Z_+^3 + 15Z_+^2 + (15-27h)Z_+ \notag\\
	&\quad+ 2(5-27h-54h^2) + 3(1+3h)(5-9h)Z_+^{-1} + 15(1+3h)^2Z_+^{-2} \notag\\
	&\quad+ 5(1+3h)^3Z_+^{-3} \biggr] \quad\text{if } x\ge0 \,,\\
	I_-(x) &= \frac{72a_-^2(1+h)^5}{5\bigl[(1-3\dom)Z_- +2(1+3\dom)+Z_-^{-1}\bigr]^6} \biggl[5(1-3h)^3Z_-^3 + 15(1-3h)^2Z_-^2 \notag\\
	&\quad + 3(1-3h)(5+9h) Z_- + 2(5+27h-54h^2) + 3(5+9h)Z_-^{-1} \notag\\
	&\quad + 15 Z_-^{-2} + 5Z_-^{-3}\biggr] \quad\text{if } x<0 \,,
\end{align}
\end{subequations}
where the expression for $I_-(x)$ follows from an analogous computation using the boundary conditions at $-\infty$.

At $x=0$, these expressions simplify to
\begin{subequations}\label{I_+-(0)_domexp}
\begin{align}
	I_+(0) &= \frac{a_+^2}{30}(1-a_0)^3[5(1+a_0) + h(1-7a_0-14a_0^2) + 10h^2a_0^3] \,,\\
	I_-(0) &= \frac{a_-^2}{30}a_0^3[5(2-a_0) + h(20-35a_0+14a_0^2) + 10h^2(1-a_0)^3] \,.
\end{align}
\end{subequations}
With \eqref{I_+-(0)_domexp}, we obtain from \eqref{k_0},
\begin{subequations}
\begin{align}
	k_0 &= \frac{1}{5\sqrt3\sqrt{a_+^2+a_-^2}[1-h(2a_0-1)]} \notag \\
	&\quad\times\Biggl(\frac{a_+^3(1-a_0)^2}{a_-a_0}[5(1+a_0) + h(1-7a_0-14a_0^2) + 10h^2a_0^3]  \notag \\
	&\qquad -\frac{a_-^3a_0^2}{a_+(1-a_0)}[5(2-a_0) + h(20-35a_0+14a_0^2) + 10h^2(1-a_0)^3] \Biggr) \label{k0_dom1} \\
	&=  \frac{a_+(1-a_0)}{5\sqrt3\sqrt{1+2a_0-3ha_0^2}[3-2a_0+3h(1-a_0)^2][1-h(2a_0-1)]} \notag \\
	&\quad\times \Bigl[10(2a_0-1) - 2h(10-37a_0+37a_0^2) + h^2(2a_0-1)(10-43a_0+43a_0^2] \notag \\
	&\qquad + 3h^3a_0(1-a_0)(3-10a_0+10a_0^2) \Bigr]\,, 
\end{align}
\end{subequations}
where the second equality is obtained by using \eqref{get_a0_dom} to express $a_-$ in terms of $a_+$ and $a_0$.

We note that
\begin{equation}
	\pder{k_0}{h}\Bigl|_{h=0} = a_+\,\frac{a_0(1-a_0)(2-61a_0+88a_0^2-44a_0^3)}{5\sqrt3(3-2a_0)^2(1+2a_0)^{3/2}}.
\end{equation}
This is negative if $a_0\gtrsim 0.0345$.

Defining
\begin{subequations}
\begin{align}
	\ph_+(z) & = \sqrt{1-h}\Biggl[\frac{(7+9h)z^{-1}+8}{(1+3h)z^{-1}-z} - \frac{10(1-3h+(1+3h)z^{-1})}{z+2(1-3h)+(1+3h)z^{-1}} \Biggr] \,,\\
	\ph_-(z) & = \sqrt{1+h}\Biggl[\frac{(7-9h)z+8}{z^{-1}-(1-3h)z} + \frac{10(1+3h+(1-3h)z)}{z^{-1}+2(1+3h)+(1-3h)z} \Biggr] \,,
\end{align}
\end{subequations}
the following indefinite integrals can be written as
\begin{subequations}\label{IntIpmOverPp2expl}
\begin{align}
	15a_+\int^x\frac{I_+(y)}{[P'(y)]^2}\,dy &= \ph_+(Z_+) + \frac{4}{\sqrt{h}}\arctan\left(\frac{Z_++1-3h}{3\sqrt{1-h}\sqrt{h}}\right) - \frac{2\pi\sgn h}{\sqrt{h}}  \quad\text{if } x\ge0\,, \label{IntIpOverPp2expl}\\
	15a_-\int^x\frac{I_-(y)}{[P'(y)]^2}\,dy &= \ph_-(Z_-) + \frac{4}{\sqrt{h}}\arctanh\left(\frac{Z_-^{-1}+1+3h}{3\sqrt{1+h}\sqrt{h}}\right) + \frac{2 i\pi}{\sqrt{h}}  \quad\text{if } x<0\,. \label{IntImOverPp2expl}
\end{align}
\end{subequations}
These integrals can be found with \emph{Mathematica} after rearranging terms. Checking by differentiation is easily done in \emph{Mathematica}.
We note that $\eqref{IntImOverPp2expl}$ is obtained from \eqref{IntIpOverPp2expl} by the joint substitutions $Z_+\to Z_-^{-1}$, $a_+\to a_-$, and $h\to-h$. The constants $\frac{2\pi\sgn h}{\sqrt{h}}$ and $\frac{2 i\pi}{\sqrt{h}}$ were added only to have real-valued right-hand sides for every $h\in(-1,1)$. (Note that the argument of $\arctanh$ is greater than 1 if $h>0$, and imaginary if $h<0$.)

Then we obtain from \eqref{define_k_pm},
\begin{subequations}\label{k_pm_explicit_dom}
\begin{alignat}{2}
	k_+(x) &= k_0 + \frac{2a_+}{15}[\ph_+(Z_+)-\ph_+(A_+)] \notag \\
		&\quad + \frac{8a_+}{15\sqrt{h}}\Bigl[\arctan\left(\frac{Z_++1-3h}{3\sqrt{1-h}\sqrt{h}}\right) -\arctan\left(\frac{A_++1-3h}{3\sqrt{1-h}\sqrt{h}}\right)\Bigr] &&\quad\text{if } x\ge0\,, \label{k_p_explicit_dom}\\
	k_-(x) &= k_0 - \frac{2a_-}{15}[\ph_-(A_-)-\ph_-(Z_-)] \notag\\
		&\quad- \frac{8a_-}{15\sqrt{h}}\Bigl[\arctanh\left(\frac{A_-^{-1}+1+3h}{3\sqrt{1+h}\sqrt{h}}\right) - \arctanh\left(\frac{Z_-^{-1}+1+3h}{3\sqrt{1+h}\sqrt{h}}\right) \Bigr] &&\quad\text{if } x<0\,. \label{k_m_explicit_dom}
\end{alignat}
\end{subequations}
Although the terms with $\arctan$ and $\arctanh$ are non-real for certain ranges of values of $h$, their respective differences divided by $\sqrt{h}$ are always real and finite. This follows from the above comment or from the identity $\arctan z = -i \arctanh(iz)$. 

Combining \eqref{k_pm_explicit_dom} with \eqref{P'(x)_dom}, we find from \eqref{p(x)} explicit, but complicated, expressions for $p(x)$.

For the limits $\ka_\pm = \lim_{x\to\pm\infty} k_\pm(x)$, we obtain
\begin{subequations}\label{kappa_dom_explicit}
\begin{align}
	\ka_+ &= k_0 - \frac{2a_+}{15}\ph_+(A_+) + \frac{8a_+}{15\sqrt{h}}\Bigl[\frac{\pi \sgn h}{2} -\arctan\left(\frac{A_++1-3h}{3\sqrt{1-h}\sqrt{h}}\right)\Bigr] \,, \label{kappap_dom_explicit}\\
	\ka_- &= k_0 - \frac{2a_-}{15}\ph_-(A_-) - \frac{8a_-}{15\sqrt{h}}\Bigl[\arctanh\left(\frac{A_-^{-1}+1+3h}{3\sqrt{1+h}\sqrt{h}}\right) + \frac{i \pi}{2} \Bigr] \,. \label{kappam_dom_explicit}
\end{align}
\end{subequations}
We note that $(\ka_\pm)/a_\pm$ depends only on $a_0$ and $h$. It can be shown that $\lim_{a_0\uparrow1}\frac{\ka_+}{a_+}=0$ and
\begin{align}\label{kappap_exp_dom_limit}
	\lim_{a_0\downarrow0}\frac{\ka_+}{a_+} &= \frac{2\sqrt3(1-9h)+9\sqrt{1-h}(1+7h)}{45(1+3h)} \notag \\
			&\quad - \frac{8}{15\sqrt{h}}\left[\arctan\left(\frac{\sqrt{3(1-h)}+3(1-h)}{3\sqrt{1-h}\sqrt{h}}\right) - \frac{2\pi\sgn h}{\sqrt{h}}\right] \ge0\,.
\end{align}

We have shown numerically that $\ka_+>0$ whenever $0<a_0<1$, and that $\ka_+/a_+$ is a decreasing function of $a_0$ and of $h$. Similarly, $\ka_-<0$ if $0<a_0<1$, and $\ka_-/a_-$ is a decreasing function of $a_0$ and of $h$. 

\subsubsection{Properties of $p(x)$}
Corollary \ref{Cor:prop_p(0)_explicit_nodom} can be generalized as follows.

\begin{corollary}\label{Cor:prop_p(0)_explicit_dom}
Let $b_\pm=a_\pm$ and $-1<\dom=\dom_A=\dom_B<1$. Then $p(x)$ has the following properties:

{\rm (i)}
\begin{align}
	p(0) &=\frac{a_0^2(1-a_0^2)(a_+^2+a_-^2)}{15[1+(1-2a_0)h]}\, \bigl[10(2a_0-1) - 2(10-37a_0+37a_0^2)h \notag \\
	&\quad+(2a_0-1)(10-43a_0+43a_0^2)h^2 + 3a_0(1-a_0)(3-10a_0+10a_0^2)h^3\bigr] \,.
\end{align}
If $h=0$, then $p(0)\gtrless0$ if and only if $a_0\gtrless\frac12$, which is the case if and only if $a_+\gtrless a_-$.

{\rm (ii)}
\begin{equation}\label{p'(0)_equivalent_dom}
	p'(0) = \frac{a_+a_- \sqrt{a_+^2+a_-^2}}{5\sqrt{3}} \, a_0^2(1-a_0^2)[15+18(1-2a_0)h + (3-20a_0+20a_0^2)h^2]\,,
\end{equation}
which is always positive, as in the general case. 

{\rm (iii)} The frequency $p(x)$ changes sign once. If $p(x_k)=0$, then $p(x)>0$ if $x>x_k$ and $p(x)<0$ if $x<x_k$.
\end{corollary}

The proof is analogous to that of Corollary \ref{Cor:prop_p(0)_gen} and based on the expressions derived in Section \ref{sec:derive_k_pm_dom}.

\subsubsection{Dependence of the slope $p'(0)$ on dominance}
To indicate the dependence of the slope $p'(0)$ on the dominance parameter $h$, we write $p'(0,h)$. Our aim is to investigate the ratio $p'(0,h)/p'(0,0)$. This has the advantage that this ratio depends only on $\tilde\a=\a_+/(\a_++\a_-)$, and not on $\a_+$ and $\a_-$ separately. Indeed, \eqref{p'(0)_equivalent_dom} implies
\begin{align}\label{p'(dom)_bounds_1}
	&\frac{p'(0,h)}{p'(0,0)} \\
	&= \frac{a_0(h,\tilde\a)^2(1-a_0(h,\tilde\a))^2[15+18h(1-2a_0(h,\tilde\a)) + h^2(3-20a_0(h,\tilde\a)+20a_0(h,\tilde\a)^2)]}{a_0(0,\tilde\a)^2(1-a_0(0,\tilde\a))^2}\,, \notag
\end{align}
where $a_0$ depends on $h$ and $\tilde\a$; see \eqref{get_a0_dom}. 

Figure \ref{fig:p'(0,h)/p'(0,0)} displays $p'(0,h)/p'(0,0)$ as a function of $h$ for different values of $\tilde\a$. It suggests that the minimum is obtained as either $(h,\tilde\a)\to(-1,0)$ or $(h,\tilde\a)\to(1,1)$, and the maximum is obtained as $(h,\tilde\a)\to(1,0)$ or $(h,\tilde\a)\to(-1,1)$. We will show that for every $h\in(-1,1)$,
\begin{equation}\label{p'(dom)_bounds}
	\frac{4}{5} < \frac{p'(0,h)}{p'(0,0)} < \frac{6}{5}
\end{equation}
holds for every $\tilde\a$. 

First, we determine the limits $\tilde\a\to0$ and $\tilde\a\to1$, from which we obtain the minimum and maximum values.
Corollary \ref{cor:properties of P}(iv) informs us that for each given $h$, $a_0(h,\tilde\a)$ is a monotone increasing function of $\tilde\a$; in addition, $a_0(h,\tilde\a)\to0$ as $\tilde\a\to0$ and  $a_0(h,\tilde\a)\to1$ as $\tilde\a\to1$. It is easy to show that if $\tilde\a$ is sufficiently close to 0, then $a_0(h,\tilde\a)=\sqrt{\tilde\a/(3+3h)}$ to leading order in $\tilde\a$. Therefore, taking the limit $\tilde\a\to0$ in \eqref{p'(dom)_bounds_1}, we obtain
\begin{equation}
	\frac{p'(0,h)}{p'(0,0)} \xrightarrow{\tilde\a\to0} \frac{5+h}{5}\,.
\end{equation}
If $\tilde\a$ is sufficiently close to 1, then $a_0(h,\tilde\a)=1-\sqrt{\tilde\a/(3-3h)}$ to leading order. Then, taking the limit $\tilde\a\to1$ in \eqref{p'(dom)_bounds_1}, we obtain
\begin{equation}
	\frac{p'(0,h)}{p'(0,0)} \xrightarrow{\tilde\a\to1} \frac{5-h}{5}\,.
\end{equation}

We have not yet proved that $p'(0,h)/p'(0,0)$ is always between the bounds given by \eqref{p'(dom)_bounds}. Figure \ref{fig:p'(0,h)/p'(0,0)} suggests that for every given $h$, $p'(0,h)/p'(0,0)$ is monotone in $\tilde\a$, and for every given $\tilde\a$, $p'(0,h)/p'(0,0)$ is monotone in $h$. The first statements seems to be true, but we cannot prove it. The second statement is valid if $a_0(h,\tilde\a)<\frac{1}{35}(34-\sqrt{421})\approx0.385$; then $p'(0,h)/p'(0,0)$ is strictly monotone increasing in $h$. It is also valid if $a_0(h,\tilde\a)>\frac{1}{35}(1+\sqrt{421})\approx0.615$; then $p'(0,h)/p'(0,0)$ is strictly monotone decreasing in $h$. If $\frac{1}{35}(34-\sqrt{421})<a_0(h,\tilde\a)<\frac{1}{35}(1+\sqrt{421})$, then $p'(0,h)/p'(0,0)$ is concave in $h$. This follows from differentiation of \eqref{p'(dom)_bounds_1} with respect to $h$, which yields
\begin{align}
	\pder{}{h}\frac{p'(0,h,\tilde\a)}{p'(0,0,\tilde\a)} &= \frac{a_0(h,\tilde\a)^2(1-a_0(h,\tilde\a))^2}{a_0(0,\tilde\a)^2(1-a_0(0,\tilde\a))^2} \notag \\
		&\quad\times \bigl[33(1-2a_0(h,\tilde\a)) + 9h -70ha_0(h,\tilde\a)(1-a_0(h,\tilde\a)) \bigr]\,,
\end{align}
and straightforward analysis of the term in brackets. Therefore, critical points of $p'(0,h,\tilde\a)/p'(0,0,\tilde\a)$ are on the curve given by equating this bracket to 0. Substituting this curve into $\pder{}{\tilde\a}\frac{p'(0,h,\tilde\a)}{p'(0,0,\tilde\a)}$ shows that the only critical point in the interior is at $h=0$ and $a_0=1/2$. The minima of the concave function $h\to p'(0,h,\tilde\a)/p'(0,0,\tilde\a)$ are assumed at $h=-1$ or $h=1$, and these values lie within the bounds given by \eqref{p'(dom)_bounds}. Therefore, we have shown \eqref{p'(dom)_bounds}.

\subsubsection{No dominance}
Let $h=0$. Then, instead of \eqref{IntIpmOverPp2expl}, we obtain the much simpler expressions:
\begin{subequations}
\begin{align}
	\int^x\frac{I_+(y)}{[P'(y)]^2}\,dy &= \frac{1}{a_+}\frac{Z_+^{-1}-2}{Z_+-Z_+^{-1}}  \quad\text{if } x\ge0\,, \label{IntIpOverPp2explh0} \\
	\int^x\frac{I_-(y)}{[P'(y)]^2}\,dy &=  \frac{1}{a_-}\frac{Z_-^{-1}-2}{Z_--Z_-^{-1}}  \quad\text{if } x<0\,. \label{IntImOverPp2explh0}
\end{align}
\end{subequations}
These integrals can be obtained either by direct integration of $I_+(y)/[P'(y)]^2$ if $h=0$, or from \eqref{IntIpmOverPp2expl} by taking the limit $h\to0$. Moreover, $k_0$ \eqref{k0_dom1} simplifies to
\begin{subequations}
\begin{align}
	k_0 &= \frac{a_+^4(1-a_0)^3(1+a_0) - a_-^4a_0^3(2-a_0)}{\sqrt3 a_+a_-\,\sqrt{a_+^2+a_-^2}a_0(1-a_0)} \\
		&= \frac{2a_+(1-a_0)(2a_0-1)}{\sqrt3\sqrt{1+2a_0}(3-2a_0)} 
		= \frac{2a_-a_0(2a_0-1)}{\sqrt3\sqrt{3-2a_0}(1+2a_0)}\,.
\end{align}
\end{subequations}
The expressions for $k_\pm(x)$ and $\ka_\pm$ follow easily and are given in \eqref{k_pm_explicit} and \eqref{kappapm_exp_nodom}, respectively. 

\subsection{Steepness of single-locus clines}\label{app:steepness}
For the single-locus cline $P(x)$ in \eqref{P(x)_dom}, the integral \eqref{def:steepness} can be computed with \emph{Mathematica}. We assume $-1<h<1$. The steepness $s(P)$ is
\begin{subequations}\label{steep_P_dom}
\begin{equation}
	s(P) = \int_0^\infty [P'(y)]^2\,dy + \int_{-\infty}^0 [P'(y)]^2\,dy \,,
\end{equation}
where, after rearrangement,
\begin{align}		
	&\int_0^\infty [P'(y)]^2\,dy = a_+\Biggl\{\frac{\sqrt{1-h}}{9h^2\bigl[A_++2(1-3h)+(1+3h)A_+^{-1}\bigr]^3} \notag \\
	&\quad\times \Bigl[(1-9h^2)A_+^2 + (5-15h+117h^2-189h^3+162h^4)A_+ \notag \\
	&\quad+2(1-3h)(1+3h^2)(5-3h+18h^2) +2(1-3h)^2(1+3h)(5+6h+9h^2)A_+^{-1} \notag \\
	&\quad+(1-3h)(1+3h)^2(5-3h+18h^2)A_+^{-2} +(1+3h)^3(1+3h^2)A_+^{-3}\Bigr] \notag \\
	&\quad+\frac{1-9h^2}{27h^{5/2}}\Bigl[\arctan\left(\frac{A_++1-3h}{3\sqrt{1-h}\sqrt{h}}\right) -\frac{\pi \sgn h}{2}\Bigr]\Biggr\} \,, \label{steep_P_doma}\\
	 &\int_{-\infty}^0 [P'(y)]^2\,dy = a_-\Biggl\{\frac{\sqrt{1+h}}{9h^2\bigl[(1-3h)A_- +2(1+3h)+A_-^{-1}\bigr]^3}  \notag \\
	&\quad\times \Bigl[(1-3h)^3(1+3h^2)A_-^3 + (1-3h)^2(1+3h)(5+3h+18h^2)A_-^2 \notag \\
	&\quad+2(1-3h)(1+3h)^2(5-6h+9h^2)A_- + 2(1+3h)(1+3h^2)(5+3h+18h^2) \notag \\
	&\quad+(5+15h+117h^2+189h^3+162h^4)A_-^{-1} + (1-9h^2)A_-^{-2} \Bigr] \notag \\
	&\quad + \frac{1-9h^2}{27h^{5/2}} \Bigl[\arctanh\left(\frac{A_-(1-3h)+1+3h}{3\sqrt{1+h}\sqrt{h}}\right) -\arctanh\left(\frac{1+3h}{3\sqrt{1+h}\sqrt{h}}\right)\Bigr] \Biggr\}\,. \label{steep_P_domb}
\end{align}
\end{subequations}
Although these expressions are too complicated to yield analytical insight, they can be computed readily. In the limit $h\to0$, they simplify to
\begin{subequations}\label{steep_P_nodom1}
\begin{align}		
	&\int_0^\infty [P'(y)]^2\,dy = \frac{6a_+(1+5A_+-5A_+^2+15A_+^3)}{5(1+A_+)^5} \,,\\	
	 &\int_{-\infty}^0 [P'(y)]^2\,dy = \frac{6a_-A_-^2(15-5A_-+5A_-^2+15A_-^3)}{5(1+A_-)^5}\,.
\end{align}
\end{subequations}
Employing the substitutions \eqref{A_pm_dom}, we obtain 
\begin{subequations}\label{steep_P_nodom2}
\begin{align}		
	&\int_0^\infty [P'(y)]^2\,dy = \frac{a_+}{15}[9 - (2-a_0)(1+2a_0)\sqrt3\sqrt{1+2a_0}] \,,\\	
	 &\int_{-\infty}^0 [P'(y)]^2\,dy = \frac{a_-}{15}[9 - (1+a_0)(3-2a_0)\sqrt3\sqrt{3-2a_0}]\,.
\end{align}
\end{subequations}
We emphasize that taking the limits $h\to1$ in \eqref{steep_P_doma} or $h\to-1$ in \eqref{steep_P_domb} does not yield valid results. In fact, this is obvious from \eqref{A_pm_dom} and \eqref{F_pm}.

For the rest of this subsection, we assume $h=0$.
To derive \eqref{steep_P_nodom}, we observe from \eqref{get_a0} that
\begin{equation}\label{a_+nodom1}
	a_+ = a_0\sqrt{3-2a_0}\sqrt{a_+^2+a_-^2} \,.
\end{equation}
Because \eqref{get_a0} is equivalent to
\begin{equation}\label{geta0_from_a-}
	(1-a_0)^2(1+2a_0) = \frac{a_-^2}{a_+^2+a_-^2}\,,
\end{equation}
we obtain
\begin{equation}\label{a_+nodom2}
	a_- = (1-a_0)\sqrt{1+2a_0}\sqrt{a_+^2+a_-^2}\,.
\end{equation}
Substituting \eqref{a_+nodom1} and \eqref{a_+nodom2} into \eqref{steep_P_nodom2}, we find \eqref{steep_P_nodom}.

There seems to be no simple representation of $s(P)$ in \eqref{steep_P_nodom} solely in terms of $a_+$ and $a_-$. However, we found the following accurate approximation:
\begin{equation}\label{s(P)_approx_nodom}
	s(P) \approx \frac{6-2\sqrt6}{5}\,\sqrt{H(a_+^2,\a_-^2)}\,,
\end{equation}
where $(6-2\sqrt6)/5\approx 0.220$ (results not shown). It is proportional to $P'(0)$ in \eqref{P'(0)}.

To compare the global and the local measure of steepness of a cline, we assume that $\la$ and $\a_++\a_-$, or equivalently $a_+^2+a_-^2$, are fixed but arbitrary. We already know that then both $P'(0)$ and $s(P)$ are maximized at $a_+=a_-$, or $a_0=\tfrac12$, and are symmetric about this value. From \eqref{P'(0)} we have $P'(0)=\frac{1}{\sqrt6}a_+$ if $a_+=a_-$. Together with \eqref{steep_P_a0=1/2}, this yields
\begin{equation}
	\frac{P'(0)}{s(P)} = \frac{5}{12}(2+\sqrt6)\approx1.854\,.
\end{equation}
Furthermore, starting from \eqref{P'(0)} and \eqref{steep_P_nodom}, straightforward calculations yield
\begin{equation}
	\lim_{\a_+\to0}\frac{P'(0)}{s(P)} = \frac{5}{23} (2 + 3\sqrt3)\approx1.564
\end{equation}
and \eqref{P'(0)/s(P)_bounds}.

\subsection{Steepness of two-locus clines}
In the first subsection, we derive an explicit expression for the integral $J$ in \eqref{J_def} for loci of equal effects without dominance and show that $J>0$.
We also compare the two measures of steepness of $p_A$, $p_A'(0)\approx P'(0)+\frac{1}{\rh}p'(0)$ and $s(p_A) \approx s(P)+\frac{2}{\rh}J$, and show that their ratio varies only between narrow bounds. 

In the second subsection, we provide a method for numerical computation of $J$ in the general case and apply it to support our conjecture that $J$ is always positive.

\subsubsection{Properties of $J$ for loci of equal effects without dominance}\label{app:steep_explicit}
We assume $a_\pm=b_\pm$ and $h_A=h_B=0$. We use partial integration and \eqref{1-locus ODE} to obtain
\begin{align}
	J &= - \int_{-\infty}^\infty P''(x) p(x) dx \notag\\
		&= \la\a_+ \int_0^\infty P(x)(1-P(x)) p(x) dx - \la\a_- \int_{-\infty}^0 P(x)(1-P(x)) p(x) dx\,.
\end{align}
With the explicit expression for $P(x)$ and $p(x)$ given in Sects.\ \ref{sec:one_locus} and \ref{sec:strong_rec}, these integrals can be computed: 
\begin{align}
	&\int_0^\infty P(x)(1-P(x)) p(x) dx \notag \\
	&\quad = \frac{54(3-2a_0) - \sqrt3\sqrt{1+2a_0}(86-127a_0+114a_0^2-77a_0^3+22a_0^4)}{63(1+2a_0)}
\end{align}
and
\begin{align}
	&\int_{-\infty}^0 P(x)(1-P(x)) p(x) dx \notag \\
	&\quad= \frac{54(1+2a_0) - \sqrt3\sqrt{3-2a_0}(18+42a_0+15a_0^2-11a_0^3+22a_0^4)}{63(3-2a_0)}\,.
\end{align}
Straightforward algebra using \eqref{get_a0}, \eqref{geta0_from_a-}, and $A_0=a_0^2(3-2a_0) = 1-(1-a_0)^2(1+2a_0) = 1 - (1-A_0)$, yields
\begin{align}\label{J_equal_general}
	J &= \frac{\sqrt2}{21}\la^{3/2}\left(\frac{\a_++\a_-}{2}\right)^{3/2} \Bigl[18\bigl(A_0^{3/2}+(1-A_0)^{3/2}\bigr)  \notag \\
			&\quad -\sqrt3\sqrt{3-2a_0}\sqrt{1+2a_0}(6-4a_0-19a_0^2+46a_0^3-23a_0^4)\Bigr]\,.
\end{align}
Obviously, $J$ is proportional to $\la^{3/2}$. Moreover, it is easily shown that $J>0$ because $\a_+>0$ and $\a_- > 0$. Indeed, as a function of $a_0$, $J$ is symmetric about $a_0=\tfrac12$ and is maximized at $a_0=1/2$, i.e., for symmetric selection ($\a_+=\a_-)$. Then it simplifies to
\begin{equation}\label{J_at_a0=1/2}
	J = \frac{1}{56}(48-19\sqrt6)(\la\a_+)^{3/2}\,,
\end{equation}
where $\frac{1}{56}(48-19\sqrt6)\approx0.0261$. $J$ is minimized and equals 0 if either $a_0=0$ or $a_0=1$, i.e., if $\a_+=0$ or $\a_- =0$, respectively (see also the more general Figure \ref{fig:J_steep_general}).

\begin{figure}[t]
	\centering
	\includegraphics[width=0.6\textwidth]{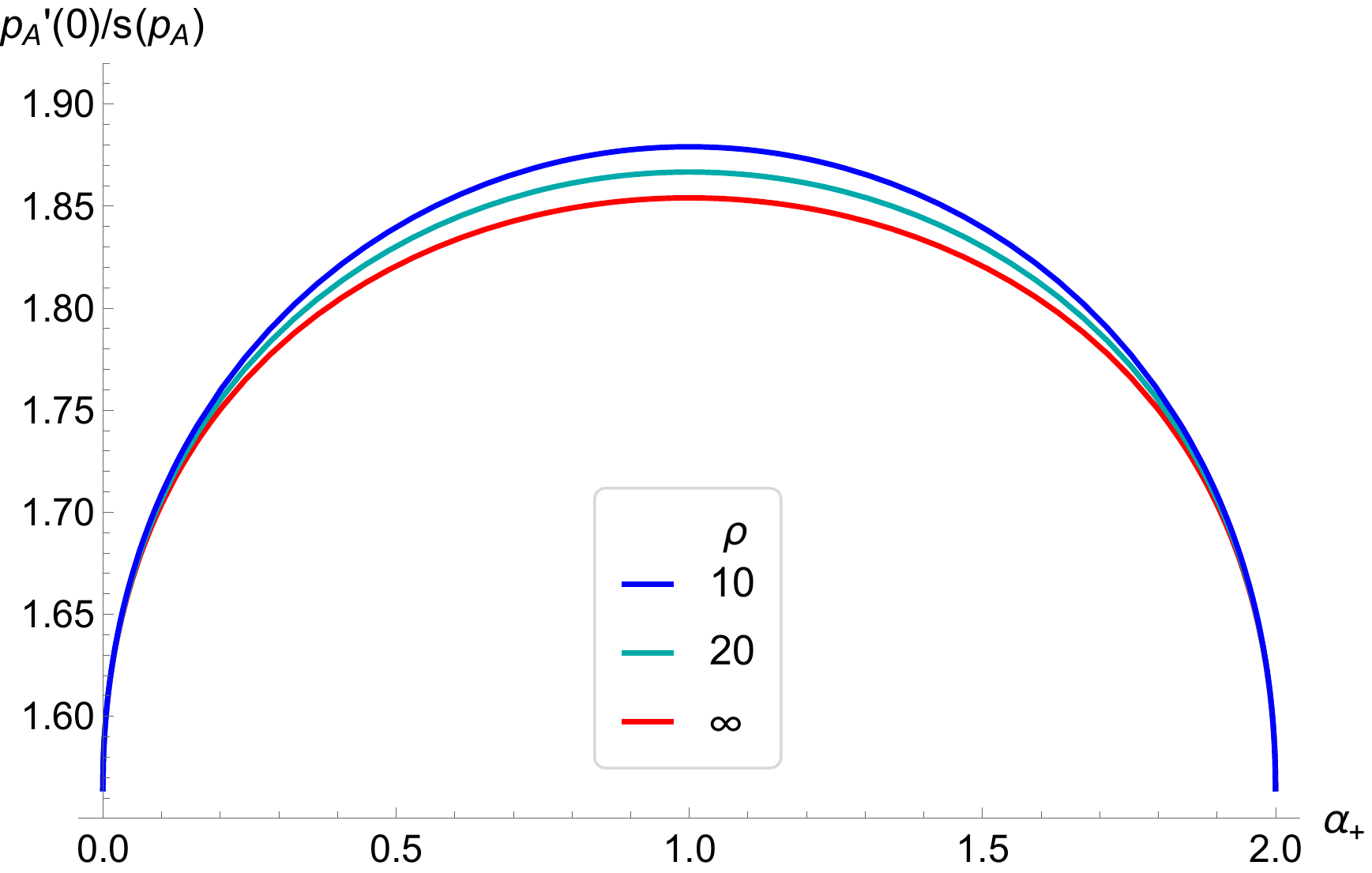}
\caption{Comparison of measures of steepness of a cline. This figure shows $p_A'(0)/s(p_A)$ as a function of $\a_+$ for $\rh=\infty$ (independent loci), $\rh=20$, and $\rh=10$. The loci have equal effects, $\la=1$, and $\a_++\a_-=2$. The slope $p_A'(0) \approx P'(0) + \frac{1}{\rh}\,p'(0)$ is computed from \eqref{P'(0)} and \eqref{p'(0)_equivalent_la}, and the steepness $s(p_A)$ in \eqref{steep_pA} from \eqref{steep_P_nodom} and \eqref{J_equal_general}.}\label{fig:steepness_vs_slope}
\end{figure}

Now, we compare the two measures of steepness of $p_A$, $p_A'(0)\approx P'(0)+\frac{1}{\rh}p'(0)$ and $s(p_A) \approx s(P)+\frac{2}{\rh}J$, for two loosely linked loci.

We assume that the loci are equivalent. Then we obtain from the first equality in \eqref{p'(0)_equivalent}, by substituting \eqref{geta0_from_a-} raised to the power $3/2$,
\begin{equation}\label{p'(0)_equivalent_la}
	p'(0) = \frac{\sqrt3}{\sqrt8}\frac{\la^{3/2} H(\a_+,\a_-)^{3/2}}{(3-2a_0)(1+2a_0)}\,.
\end{equation}
For fixed $\la$ and $\a_++\a_-$, and considered as a function of $a_0$ (or $\a_+$), $p'(0)$ is symmetric about $a_0=\tfrac12$ and maximized at $a_0=0$. There it assumes the value
\begin{equation}
	p'(0) = \frac{1}{8}\frac{\sqrt3}{\sqrt2}(\la\a_+)^{3/2} \approx 0.153(\la\a_+)^{3/2}\,,
\end{equation}
and it decays to zero as $a_0\to0$ or $a_0\to1$. From \eqref{J_at_a0=1/2} we find that at $a_0=\tfrac12$, or $\a_+=\a_-$,
\begin{equation}
	\frac{p'(0)}{2J} \approx 2.937\,.
\end{equation}
However, we find that $\lim_{\a_+\to0} \frac{p'(0)}{2J} = 0$. Nevertheless, 
if $10^{-5}\le\a_+\le1-10^{-5}$, then
\begin{equation}
	2.406 < \frac{p'(0)}{2J} < 2.937\,.
\end{equation}
Therefore, unless selection against one of the alleles essentially vanishes, $p'(0)/(2J)$ varies only between quite narrow bounds.

Figure \ref{fig:steepness_vs_slope} displays the ratio $p_A'(0)/s(p_A)$ as a function of $\a_+$ for given step size $\a_++\a_-$ and three values of recombination. For the dependence of the slope of a cline on the recombination rate in its full range, we refer to Figure \ref{fig:slope_function_of_rho}.

\begin{figure}[t]
	\centering
	\includegraphics[width=0.6\textwidth]{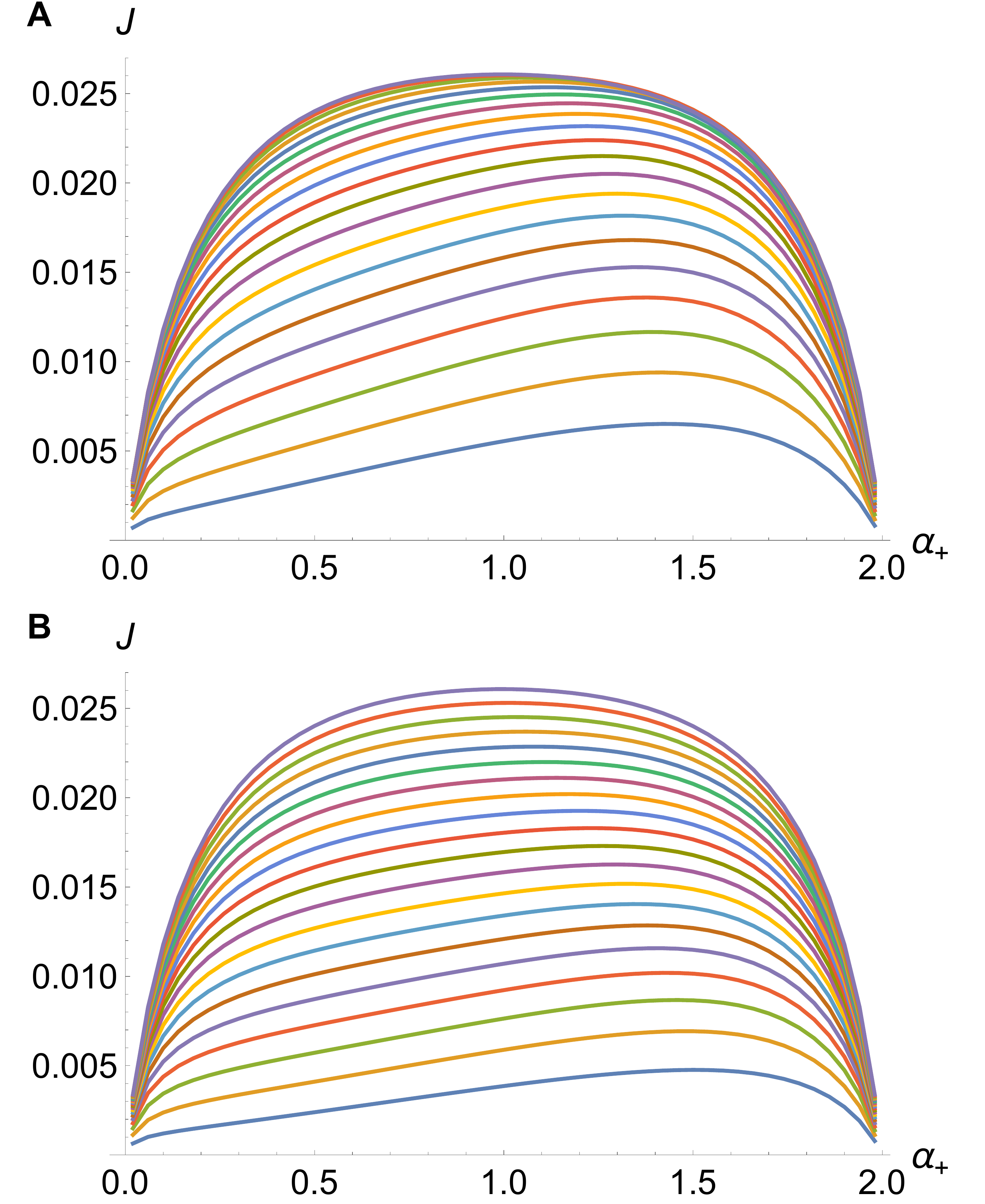}
\caption{The integral $J$ in \eqref{J_def} as a function of various parameters. In both panels, $J$ is shown as a function of $\a_+$ for the constant step size $\a_++\a_-=2$ and different combinations of $\be_+$ and $\be_-$. In {\bf A}, $\be_++\be_-=2$ and $\be_+$ is varied from $0.05$ to $1$ in steps of $0.05$ (curves from bottom to top). In {\bf B}, $\be_-=1$ and $\be_+$ is varied from $0.05$ to 1 in steps of $0.05$ (curves from bottom to top). In both panels, $\la=1$, and there is no dominance, $h_A=h_B=0$. Note that by \eqref{steep_scale_J}, a change in $\la$ affects all curves in the same way.}\label{fig:J_steep_general}
\end{figure}

\subsection{Computation and properties of $J$ in the general case}\label{app:J_general}
To compute the steepness of two-locsu clines in the general case, we use \eqref{steep_pA} and evaluate $J$ as follows:
\begin{align}
	J &=\int_{-\infty}^\infty P'(x)p'(x)\, dx \notag \\
		&= P(+\infty)p'(+\infty) - P(-\infty)p'(-\infty) - \int_{-\infty}^\infty P(x)p''(x)\, dx \notag\\
		&= - \int_{-\infty}^\infty P(x)p''(x)\, dx\,,
\end{align}
where we used the boundary conditions \eqref{boundary_cond_pqd}. Since $p''(x)$ satisfies \eqref{p''inhomog}, we obtain
\begin{subequations}\label{J_general}
\begin{align}
	- \int_0^\infty P(x)p''(x)\, dx &= a_+^2\int_0^\infty [1-2P(x)+h_A(1-6P(x)+6P(x)^2)]P(x)p(x)\,dx \notag \\
		&\quad+2b_+^2 \int_0^\infty P(x)P'(x)Q'(x)(1+h_B-2h_BQ(x))\, dx \label{Jplus_gen}
\end{align}
and
\begin{align}
	- \int_{-\infty}^0 P(x)p''(x)\, dx &= -a_-^2\int_0^\infty [1-2P(x)+h_A(1-6P(x)+6P(x)^2]P(x)p(x)\,dx \notag \\
		&\quad-2b_-^2 \int_0^\infty P(x)P'(x)Q'(x)(1+h_B-2h_BQ(x))\, dx \label{Jminus_gen}\,.
\end{align}
\end{subequations}
Both expressions can be computed numerically from the formulas given for $P(x)$, $P'(x)$, and $p(x)$ in the main text. We used \emph{Mathematica} for these evaluations.

In general, we cannot prove that $J$ is always positive. In the absence of dominance, multiplication of all four fitness effects ($\a_\pm$,$\be_\pm$) by the same positive constant $c$ can be compensated for by multiplying $\la$ by $c$. Therefore, \eqref{steep_scale_J} shows that, up to a multiplicative factor, $J$ depends only on three independent parameters. In Figure \ref{fig:J_steep_general}, we fixed the step size $\a_++\a_-$ and varied the other parameters as explained in the legend. The figure suggests convincingly that $J$ is always positive. Thus, as expected from intuition, tighter linkage apparently always increases the measures $s(p_A)$ and $s(p_B)$ of steepness of the allele-frequency clines. Based on additional numerical computations for a wide variety of parameter combinations (not shown), we conjecture that $J$ is always positive.

\end{appendix}

\section*{References}
\parindent=0pt
\normalbaselineskip16pt
\baselineskip16pt

{\everypar={\hangindent=0.8cm \hangafter=1}
%
Aeschbacher, S., B\"urger, R. 2014. The effect of linkage on establishment and survival of locally beneficial mutations. Genetics 197, 317-336.

Akerman, A., B\"urger, R. 2014. The consequences of gene flow for local adaptation and differentiation: a two-locus two-deme model. J.\ Math.\ Biol.\ 68, 1135–1198.

Barton, N.H. 1983. Multilocus clines. Evolution 37, 454-471.

Barton, N.H. 1986. The effects of linkage and density-dependent regulation on gene flow. Heredity 57, 415–426.

Barton, N.H. 1999. Clines in polygenic traits. Genet.\ Res.\ Camb.\ 74, 223-236.

Barton, N.H., Shpak, M. 2000. The effect of epistasis on the structure of hybrid zones. Genet.\ Res.\ Camb.\ 75, 179-198.

B\"urger, R. 2000. The Mathematical Theory of Selection, Recombination, and Mutation. Chichester: Wiley.

B\"urger, R. 2014. A survey of migration-selection models in population genetics. Discrete Cont.\ Dyn.\ Syst.\ B 19, 883 - 959.

B\"urger, R., Akerman, A. 2011. The effects of linkage and gene flow on local adaptation: A
two-locus continent-island model. Theor.\ Popul.\ Biol., 80, 272-288.

Conley, C.C. 1975. An application of Wazewski's method to a non-linear boundary value problem which arises in population genetics. J.\ Math.\ Biol.\ 2, 241-249.

Endler, J.A. 1977. Geographic Variation, Speciation, and Clines. Princeton Univ.\ Press, Princeton, New Jersey.

Fife, P.C., Peletier, L.A. 1977. Nonlinear diffusion in population genetics. Arch.\ Rat.\ Mech.\ Anal.\ 64, 93-109.

Fife, P.C., Peletier, L.A. 1981. Clines induced by variable selection and migration. Proc.\ R.\ Soc.\ Lond.\ B 214, 99-123.

Fleming, W.H., 1975. A selection-migration model in population genetics. J.\ Math.\ Biol.\ 2, 219-233.

Geroldinger, L., B\"urger, R. 2015. Clines in quantitative traits: The role of migration patterns and
selection scenarios. Theor.\ Popul.\ Biol.\ 90, 43-66.

Gilbarg, D., Trudinger, N.S. 2001. Elliptic Partial Differential Equations of Second Order. Springer, Berlin.

Haldane, J.B.S. 1948. The theory of a cline. J.\ Genetics 48, 277-284.

Henry, D. 1981. Geometric Theory of Semilinear Parabolic Equations. In: Lecture Notes in Mathematics, vol.\ 840. Springer, Berlin.

Hoffmann, A.A., Anderson A., Hallas R. 2002. Opposing clines for high and low temperature resistance in \emph{Drosophila melanogaster}. Ecology Letters 5, 614-618.

Kolmogoroff, A., Petrovsky, I., Piscounoff, N. 1937. \'Etude de l'\'equation de la diffusion avec croissance de la quantite de mati\'ere et son application \`a un probl\`eme biologique. Bull.\ Univ.\ Etat Moscou, Ser.\ Int., Sect.\ A, Math.\ et Mecan.\ {\bf 1}, Fasc.\ 6, 1-25.

Lenormand, T., 2002. Gene flow and the limits to natural selection. Trends Ecol.\ Evol.\ 17, 183–189.

Lenormand, T., Guillemaud, T., Bourguet, D., Raymond, R. 1998. Evaluating gene flow using selected markers: A case study. Genetics 149, 1383-1392.

Liang, S., Lou, Y., 2011. On the dependence of the population size on the dispersal rate. Disc.\ Cont.\ Dynam.\ Sys.\ Series B 17, 2771-2788.

Lohman B.K., Berner D., Bolnick D.I. 2017. Clines arc through multivariate morphospace. Amer.\ Nat.\ 189, 354-367. 

Lou, Y., Nagylaki, T., 2002. A semilinear parabolic system for migration
and selection in population genetics. J.\ Differential Equations 181, 388-418.

Lou, Y., Nagylaki, T., 2004. Evolution of a semilinear parabolic system for migration
and selection in population genetics. J.\ Differential Equations 204, 292-322.

Lou, Y., Nagylaki, T., 2006. Evolution of a semilinear parabolic system for migration
and selection without dominance.. J.\ Differential Equations 225, 624-665.

Lou, Y., Nagylaki, T., Ni, W.-M. 2013. An introduction to migration-selection PDE models.
Disc.\ Cont.\ Dyn.\ Syst.\ A 33, 4349-4373.

Lou, Y., Ni, W.-M., Su, S. 2010. An indefinite nonlinear diffusion problem in population genetics, II: Stability and multiplicity.
Disc.\ Cont.\ Dyn.\ Syst.\ A 27, 643-655.

Mallet, J., Barton, N. 1989. Inference from clines stabilized by frequency-dependent selection.  Genetics 122, 967-976.

Nagylaki, T., 1975. Conditions for the existence of clines. Genetics 80, 595-615.

Nagylaki, T., 1976. Clines with variable migration. Genetics 83, 867-886.

Nagylaki, T., 1978. Clines with asymmetric migration. Genetics 88, 813-827.

Nagylaki, T. 1989. The diffusion model for migration and selection. Pp.\ 55-75 
in: Hastings, A. (Ed.), Some Mathematical Questions in Biology. Lecture Notes
on Mathematics in the Life Sciences, vol. 20. American Mathematical Society,
Providence, RI.

Nagylaki, T., Lou, Y. 2008. The dynamics of migration-selection models.
In: Friedman, A.\ (ed) Tutorials in Mathematical Biosciences IV.
Lect.\ Notes Math.\ 1922, pp.\ 119 - 172. 
Berlin Heidelberg New York: Springer.

Nakashima, K., Ni, W.-M., Su, S. 2010. An indefinite nonlinear diffusion 
problem in population genetics, I: Existence and limiting profiles.
Disc.\ Cont.\ Dyn.\ Syst.\ A 27, 617-641.

Slatkin, M. 1973. Gene flow and selection in a cline. Genetics 75, 773-756.

Slatkin, M. 1975. Gene flow and selection on a two-locus system. Genetics 81: 787–802.

}

\end{document}